\documentclass[final]{siamart171218}

\usepackage{lipsum}
\usepackage{amsfonts}
\usepackage{graphicx}
\usepackage{epstopdf}
\usepackage{algorithmic}
\ifpdf
  \DeclareGraphicsExtensions{.eps,.pdf,.png,.jpg}
\else
  \DeclareGraphicsExtensions{.eps}
\fi

\newsiamremark{remark}{Remark}
\newsiamremark{hypothesis}{Hypothesis}
\crefname{hypothesis}{Hypothesis}{Hypotheses}
\newsiamthm{claim}{Claim}

\headers{Stochastic Pairwise Interactions}{T. Dag\`es and A. M. Bruckstein}

\title{Doubly Stochastic Pairwise Interactions for Agreement and Alignment
\thanks{\textbf{Funding:} {This research was partially supported by the Israel Science Foundation grant no. 2306/18.}}
}

\author{Thomas Dag\`es\thanks{Department of Computer Science, Technion Israel Institute of Technology, Haifa, ISRAEL (\email{thomas.dages@cs.technion.ac.il}).}
\and Alfred M. Bruckstein\footnotemark[2]}

\usepackage{amsopn}
\DeclareMathOperator{\diag}{Diag}
\DeclareMathOperator\arccosh{arccosh}

\usepackage{bbm} 
\usepackage{amssymb} 

\usepackage{mathtools}
\DeclarePairedDelimiter{\ceil}{\lceil}{\rceil}
\DeclarePairedDelimiter{\floor}{\lfloor}{\rfloor}
\DeclareMathOperator*{\argmax}{argmax}

\usepackage[caption=false]{subfig} 

\usepackage{nicefrac} 

\usepackage{enumitem}

\makeatletter
\newcommand*{\addFileDependency}[1]{
  \typeout{(#1)}
  \@addtofilelist{#1}
  \IfFileExists{#1}{}{\typeout{No file #1.}}
}
\makeatother

\ifpdf
\hypersetup{
  pdftitle={Agreement and Alignment via Stochastic Pairwise Interactions},
  pdfauthor={T. Dag\`es and A. M. Bruckstein}
}
\fi

\begin{document}

\maketitle

\begin{abstract}
    Random pairwise encounters often occur in large populations, or groups of mobile agents, and various types of local interactions that happen at encounters account for emergent global phenomena. In particular, in the fields of swarm robotics, sociobiology, and social dynamics, several types of local pairwise interactions were proposed and analysed leading to spatial gathering or clustering and agreement in teams of robotic agents coordinated motion, in animal herds, or in human societies. We here propose a very simple stochastic interaction at encounters that leads to agreement or geometric alignment in swarms of simple agents, and analyse the process of converging to consensus. Consider a group of agents whose ``states'' evolve in time by pairwise interactions: the state of an agent is either a real value (a randomly initialised position within an interval) or a vector that is either unconstrained (e.g. the location of the agent in the plane) or constrained to have unit length (e.g. the direction of the agent's motion). The interactions are doubly stochastic, in the sense that, at discrete time steps, pairs of agents are randomly selected and their new states are independently and uniformly set at random in (local) domains or intervals defined by the states of the interacting pair. We show that such processes lead, in finite expected time (measured by the number of interactions that occurred) to agreement in case of unconstrained states and alignment when the states are unit vectors.
\end{abstract}



\section{Introduction}

We consider a group of $N$ agents with states described by the set $\{x_1,\cdots,x_N\}$. The states $x_i$ can be either real scalar values in some interval $I\subset \mathbb{R}$, or vectors in a $D$-dimensional box $C\subset\mathbb{R}^D$, or unit vectors on the circle $S$. Typical agents could be people in a social group or a large population, ants in a colony, man-made robots designed to act in a swarm, fish in a school, gas molecules moving around in a container, or even software bots migrating from computer to computer on a network like the internet. 

The  ``state'' of an agent may therefore be the opinion of a person on some issue, which can be measured by a real value on $\mathbb{R}$, such as how much you like a product on a (continuous) scale from $0$ to $10$, or where on the political spectrum you are from the far left ($-\infty$) to the far right ($+\infty$), or the location of an ant or a robot in a planar domain $C\subset\mathbb{R}^2$ or the direction of motion of a mobile robot. Of course, ``state'' might also stand for the classical memory content of a (finite-state) machine or a bot-program but we shall not consider such discrete states here.

We assume that agents are identical in their capabilities and behaviour and their states change in time only due to interactions with other agents. The interaction rules must be given and depend only on the current states of the interacting agents, not on their identities. We say that the agents are identical, anonymous, and oblivious. Given some rules of interaction and their timing schedule, we are interested in the evolution of the states of the agents, the evolution reflecting some ``emergent behaviours'' of the swarm of agents, like agreement in a community of people, or gathering, grouping, or clustering of robotic agents, or coordinated motion due to alignment of directions of movement in some herd of animals such as bison or insects like locusts.

In this paper we shall analyse a particular rule of interaction: we assume that each agent can interact with any other agent at all times and that interactions are pairwise only between randomly selected agents and occur at distinct and discrete times, denoted sequentially as $t_k$ for $k = 0,1,2$...  At the interaction moments, two randomly selected agents exchange information about their states and decide on how to to update their own state.
This leads to an evolution of the set of states in time and hopefully to  convergence to some interesting globally emergent swarming state.

\section{An overview of previous results}

To set the stage for our proposed rules of interaction and the consequent emerging phenomena, we shall first describe several previously proposed pairwise interaction rules and the results obtained on the consequent global behaviours.

A fundamental problem in distributed computation, as well as in opinion dynamics, is to achieve agreement or consensus, via a sequence of local interactions. Suppose that $N$ agents have as initial states randomly selected real values and we would like the agents to eventually agree on a common real value. If agents could see all others' states, they could agree say on the average value of all the states. Suppose however that agents can only see neighbouring agents or agents connected to them, such as in a given fixed neighbourhood graph, determined by the designer of the network, or by geometrical proximity. Then agents can average sequentially but locally, only within their neighbourhood. The question is: will this process eventually converge? This problem is not too difficult, and we can in many cases prove that, indeed, in time, the agents will agree on a value that is the average of the initial states.

However, consider a stochastic setting in that at distinct time instances $t_k$ ordered increasingly for $k = 0,1,2\hdots$ random pairs of agents are selected and they replace their states by the averages of the corresponding values. How does this random pairwise averaging process behave? The result is a stochastic state equalising process and all states will converge to the  average of the initial states of the agents. This was first analysed in \cite{tarmy1981random,proschan1984random}. There is extensive research work on such processes under the name of ``distributed gossiping''. These ``gossiping'' works analyse the evolution in time of the gathering process to the average value of the initial states. They consider the interaction moment when the sum of squared individual state departures from the average reaches a value less than a small preset $\varepsilon$. The conclusion of these results is an upper and lower bound on the so defined time to convergence proportional to $\log \left(\frac{1}{\varepsilon}\right)$ where the constant factors depend on the number of agents $N$ and on the size of the initial spread of the states, see \cite{boyd2006randomized,dimakis2010gossip}.

Random pairwise interactions were also proposed as suitable models for achieving consensus in social studies on opinion dynamics in populations. Several studies proposed to consider societies of individuals as holding initial opinions, or states, quantified by some real values in an interval $I$ of the real line $\mathbb{R}$, and the following rule of evolution: at discrete time instants, if two random members of the society meet, they change their opinion so as to ``approach'' each other by a deterministic fixed proportion of the size of the difference between their opinions 
\cite{hegselmann2002opinion,aldous2012lecture,chatterjee2019note}.
In a more complex and realistic setting, this is done only if the difference between their opinions is smaller than a certain threshold, otherwise the meeting results in no changes of opinion whatsoever. This later idea is the so-called pairwise interactions based ``bounded confidence model'' considered by the ``French School'' of opinion dynamics led by Deffuant, see
\cite{deffuant2000mixing,haggstrom2012pairwise}.
These works also lead to clustering and convergence opinions either to the average of the initial states (if no bounded confidence threshold is assumed) or to several clusters in bounded confidence models 
\cite{deffuant2000mixing,lorenz2007continuous}.
Recently, \cite{gantert2020strictly} studied a generalisation of the unbounded confidence model in non convex opinion spaces, e.g. on the unit circle, but simplified the interaction graph to a ring coined as the ``compass model''. We we focus on the complete graph for interactions.

The idea of using pairwise interaction models in analysing the emergence of various collective dynamics phenomena is also prevalent in swarm robotics. It was, for example, proposed to model cooperative localisation processes in swarms of robots, to improve their self location estimates by averaging those at random pairwise encounters, when the agents know that they are co-located, hence their estimates should ideally coincide \cite{elor2012thermodynamic}. The idea of Encounter Averaging of self location $2D$-vector estimates was there shown to significantly improve the cooperative odometric location evaluations, even under the assumption that the pairwise agent encounters are totally random, which obviously is not the case. This idea is, of course, prevalent in physics. In thermodynamics, one considers gas particles (molecules) moving and colliding, their self-propelling motion manifested as thermal energy and their collisions modelled with several types of deterministic or randomised interactions. The emergent collective behaviour in this case is quantified by globally measured properties of the system of molecules like variations in density, temperature, and pressure as functions of container geometry and external, perhaps even temporally changing, factors
\cite{bertin2006boltzmann}.

\section{Brief overview of results}

We consider here three types of problems concerning systems with multi-agent pairwise interactions. The interactions that we define are stochastic and we prove that a desired behaviour eventually emerges. We also provide evaluations about the expected time (measured by the number of interactions) to the convergence to a state that is very close to the desired global behaviour. The problems are the following: systems of $N$ agents, with states defined by either real numbers in an interval $I\subset\mathbb{R}$ or by real vectors in a $D$-dimensional box in $\mathbb{R}^D$ or by unit vectors on the circle are considered to evolve due to random pairwise interactions that result in changes of the states of the interacting agents. The rules of the evolution are:
\begin{enumerate}
    \item the interaction moments are discrete times $t_1, t_2, \hdots$ starting from $t_0 = 0$ when a random initialisation is done,
    \item at each moment $t_k$ a random pair of agents is selected uniformly from the $\frac{N(N-1)}{2}$ possible pairs of agents,
    \item the selected agents $i$ and $j$ uniformly choose new states in the ``interval'' defined by their states $\{x_i^k ,x_j^k\}$ as follows:
    \begin{itemize}
        \item if the states are real numbers, the ``interval'' is just chosen to be as  $\left[\min\{x_i^k,x_j^k\},\max\{x_i^k,x_j^k\}\right]$,
        \item if the states are $D$-dimensional vectors, the ``interval'' is the one dimensional line segment  $\{\lambda x_i^k+(1-\lambda)x_j^k\mid \lambda\in\left[0,1\right]\}$ embedded in $\mathbb{R}^D$,
        \item if the states are two unit vectors, the ``interval'' is the geodesic circle arc between the two points defined by $x_i^k$ and $x_j^k$ on the unit circle.
    \end{itemize}
\end{enumerate}

The main question we address is: how does such a stochastic system evolve in time, measured by the indices of the interaction times $t_0, t_1, t_2 \hdots$ (i.e. $0, 1, 2\hdots$)? We prove that in all the cases above, the system gathers the agents' states, with probability one, to a common random point on the real line in in the first case, to a single random point in $\mathbb{R}^D$ for the second case, and to a random unit vector in the third case. We also show that the expected time to $\varepsilon$-convergence is finite and provide bounds on it, where $\varepsilon$-convergence is defined as the expected time (or number of iterations) for which the spread of the system state is smaller than $\varepsilon$. 

We list the main results below:
\begin{itemize}
    \item \textbf{Evolution of real values in an interval.} We prove that $\varepsilon$-convergence is achieved almost surely and in finite expected time, by deriving an upper  bound of $O(N\log\left(\frac{N}{\varepsilon^2}\right))$ on the expected $\varepsilon$-gathering. We illustrate our theory with extensive numerical simulations and they prove the quality and tightness of our bound.
    \item \textbf{Evolution of real values in a $D$-dimensional box.} Similarly to the $1D$ case, we again prove almost sure $\varepsilon$-convergence in finite expected time, by deriving a $O\left(N\log\left(\frac{DN}{\varepsilon^2}\right)\right)$ upper bound. Extensive numerical simulations are performed to show the quality of the bound.
    \item \textbf{Evolution of unit vectors on the unit circle.} We prove almost sure $\varepsilon$-convergence in finite expected time. This problem is significantly more challenging: we here provide a simplistic approach yielding a crude and upper bound, as revealed in extensive experiments. We also provide and detail several promising approaches for deriving a more reasonable upper bound but leave the refinement issue as an open challenge.
\end{itemize}

\section{1-Dimensional case}

A preliminary simple model to study social gathering is to assume that people's opinions solely depends on a unique parameter $x_{opinion}$ that lives on the real line $\mathbb{R}$. This model suits well systems where opinions exist along a simple spectrum, with notions of ``left wing'' and ``right wing'' opinions. The larger $x_{opinion}$ is, the more the opinion is ``right-wing'' and the smaller $x_{opinion}$ becomes, the more the opinion is ``left-wing''. It is important to note that in this model, if $x_{opinion}$ increases, then the opinion becomes more and more ``right wing''.  We may also assume that the space of opinions is either bounded, which models well systems will limited ``left'' and ``right wing extremism'', or unbounded, which models better systems with unlimited ``extremism'' in one or both directions.

Mathematically, we define the space of opinions to be an interval $I\subset \mathbb{R}$ of the real line. Our system comprises of $N\ge 2$ indistinguishable agents, each with their own opinion: $x_i\in I$ for $i\in\{1,\cdots,N\}$. The initial distribution of the opinions is a given set $\{x_i^0\}_{1\le i\le N}$ where $x_i^0\in I$ for all $i$. Opinion dynamics are modelled in discrete time, conditionally to the state of opinions at the previous time step $X_k = (x_1^k,\cdots,x_N^k)^T$. The evolution law at time step $t_{k+1}$ for all $i\in\{1,\cdots,N\}$ is:
\begin{equation}
    \label{eq:1D evolution}
    x_i^{k+1} = \mathbbm{1}_{i\notin\{A_{k+1},B_{k+1}\}}x_i^k + \mathbbm{1}_{i=A_{k+1}}U_{1}^{k+1}+\mathbbm{1}_{i=B_{k+1}}U_{2}^{k+1},
\end{equation}
where $(A_{k+1},B_{k+1})$ is a random uniform sampling of two indices of $\{1,\cdots,N\}$ without replacement independent of the past, $\mathbbm{1}_\mathcal{A}$ is the indicator function of the event $\mathcal{A}$, and where conditionally to $X_k$, $A_{k+1}$, and $B_{k+1}$, $U_1^{k+1}$ and $U_2^{k+1}$ are independent random uniform variables in the interval $[\min\{x_{A_{k+1}}^k,x_{B_{k+1}}^k\},\max\{x_{A_{k+1}}^k,x_{B_{k+1}}^k\}]$. Concretely, at each time step, two random agents $A_{k+1}$ and $B_{k+1}$ are selected and they then independently and uniformly resample their opinion in the interval between their previous opinions. See \cref{fig: 1D evolution} for an example.

\begin{figure}[tbhp]
  \centering
    \includegraphics[width=\textwidth]{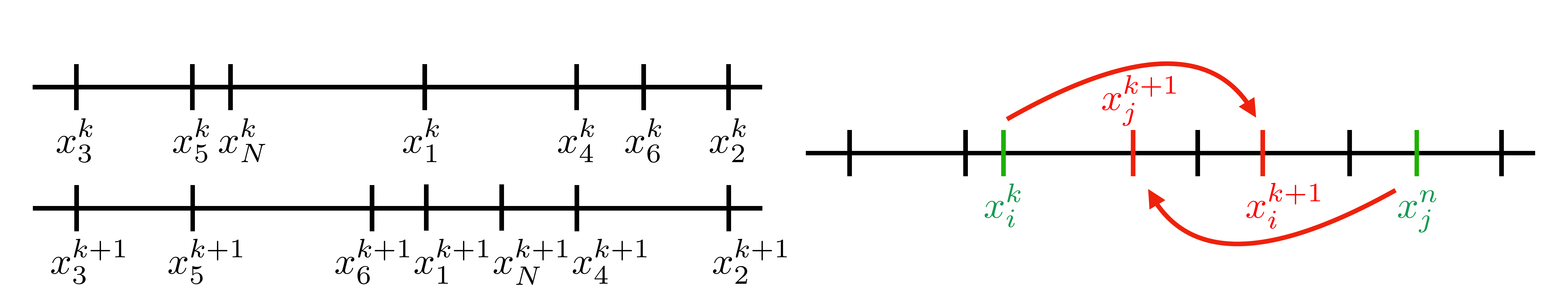}
    \caption{One-step opinion evolution in the one dimensional model.}
    \label{fig: 1D evolution}
\end{figure}

\begin{proposition}
    \label{prop: 1D max dist non increasing}
    The quantity $\max\limits_{i\neq j}\left|x_i^k-x_j^k\right|$ is a positive non increasing function and converges when $k\xrightarrow[]{}\infty$.
\end{proposition}

\begin{proof}
    The result is immediate due to the rules of motion: no point can become more extreme than the already most extreme points. Mathematically, in the inclusion sense, the smallest closed interval containing all points is non increasing and will thus converge to a  non empty limit interval. 
\end{proof}

\begin{proposition}
    \label{prop: 1D expec lyapunov n+1}
    For all $k\in\mathbb{N}$, we have:
    \begin{equation*}
        \mathbb{E}\Big(\sum\limits_{i\neq j}(x_i^{k+1}-x_j^{k+1})^2\mid X_k\Big) = \Big(1-\frac{2N+1}{3N(N-1)}\Big)\sum\limits_{i\neq j}(x_i^k-x_j^k)^2.
    \end{equation*}
\end{proposition}

\begin{proof}
	We here give an overview of the proof, for which a detailed one can be found in the supplementary material \cref{sm: proof 1D expec lyapunov n+1}.

    Denote $\mathcal{L}_{i,j}^k$ and $\mathcal{L}^k$ the studied quantities:
    \begin{align}
        \mathcal{L}_{i,j}^k &= (x_i^k-x_j^k)^2 \\
        \mathcal{L}^k &= \sum\limits_{i\neq j}(x_i^k-x_j^k)^2 = \sum\limits_{i\neq j}\mathcal{L}_{i,j}^k.
    \end{align}

    We can calculate the expected $\mathcal{L}^{k+1}$ by conditioning on the chosen pair $(m,l)$ for $(A^{k+1},B^{k+1})$ and use the linearity of the expectation to focus on the conditional expectation of the $(i,j)$ term $\mathcal{L}_{i,j}^{k+1}$. By expanding the square in this term, and once again using the linearity of the expectation, we thus only need to know the first two moments of the updates of the opinions, conditioned on the choice of pair for update $(m,l)$. These are simply given by the well-known first order moments of one dimensional uniform random variables (see the supplementary material \cref{prop: E(x_i) and E(x_i^2)}). Some straightforward calculations give the final result.
\end{proof}

\begin{proposition}
    \label{prop: 1D bounds lyap}
    For all $k\in\mathbb{N}$, we have:
    \begin{equation*}
        N\max\limits_{i\neq j} \left|x_i^k - x_j^k \right|^2 \le \mathcal{L}^k \le \frac{N^2}{2}\max\limits_{i\neq j} \left|x_i^k - x_j^k \right|^2 .
    \end{equation*}
\end{proposition}

\begin{proof}
    This is a well known result from \cite{popoviciu1935equations} for the upper bound and \cite{nagy1918algebraische} for the lower bound, after noticing that up to normalisation and a constant factor, the Lyapunov sum of square differences of $N$ points $x_1,\cdots, x_N$ is the (biased) empirical variance of the points. We re-derive the proof as supplementary material in \cref{sm: proof 1D bounds lyap}.
\end{proof}

\begin{definition}
    \label{def: 1D stopping time convergence}
    For any $\varepsilon>0$, we denote $T_\varepsilon$ the stopping time with respect to the natural filtration induced by the $(X_k)$ sequence defined as:
    $$T_\varepsilon =  \min\{k\in\mathbb{N}\mid \max\limits_{i\neq j} \left|x_i^k - x_j^k\right| \le \varepsilon\}.$$
\end{definition}

\begin{definition}
    \label{def: 1D stopping time lyapunov}
    For any $\varepsilon>0$, we denote $T_\varepsilon'$ the stopping time with respect to the natural filtration induced by the $(X_k)$ sequence defined as:
    $$T_\varepsilon' = \min\{k\in\mathbb{N}\mid \sum\limits_{i\neq j}(x_i^k -x_j^k)^2 \le N\varepsilon^2\}.$$ 
\end{definition}

\begin{proposition}
    \label{prop: 1D T<=T'}
    For all $\varepsilon>0$, $T_\varepsilon\le T_\varepsilon'$.
\end{proposition}

\begin{proof}
    The result follows from \cref{prop: 1D bounds lyap}. If $r>\varepsilon$, the minimum possible configuration for $\mathcal{L}^k$ given that the range is $r$, denoted $\mathcal{L}_{\min}$, is necessarily strictly larger than $N\varepsilon^2$. A reciprocal argument gives that if $\mathcal{L}_{\min}\le N\varepsilon^2$, then $r\le \varepsilon$.
    
    Now assume our system has evolved to a configuration with $\mathcal{L}^k\le N\varepsilon^2$. Denote $r_k$ its range. The minimum possible Lyapunov for configurations of opinions, given the range $r_k$, is then lower than $N\varepsilon^2$. Then necessarily $r_k\le \varepsilon$, implying that the first occurrence of the event $\{r_k\le\varepsilon\}$ is anterior to the event $\{\mathcal{L}^k\le N\varepsilon^2\}$, i.e. $T_\varepsilon\le T_\varepsilon'$.
\end{proof}

\begin{theorem}
    \label{th:1D finite expected time}
    For a system evolving according to \cref{eq:1D evolution}, for any $\varepsilon>0$,  we have:
    $$\mathbb{E}(T_\varepsilon\mid X_0) \le \frac{3N(N-1)}{2N+1}\ln\left(\frac{\mathcal{L}^0}{N\varepsilon^2}\right)+\frac{3N(N-1)}{2N+1} \le \frac{3}{2}N\ln\left(\frac{\mathcal{L}^0}{N\varepsilon^2}\right) + \frac{3}{2}N.$$
\end{theorem}

\begin{proof}
	We here give an overview of the proof, for which a detailed one can be found in the supplementary material \cref{sm: proof 1D finite expected time}.

	Due to \cref{prop: 1D T<=T'}, it suffices to find an upper bound on the expectation of $T_\varepsilon'$. The idea of the proof is to write this expectation as the sum over $k$ of tail distributions: $\mathbb{P}(T_{\varepsilon}'>k\mid X_0)$. If $\mathcal{L}^k$ is lower or equal to $N\varepsilon^2$, then necessarily $T_\varepsilon'$ is lower or equal to $k$. By contraposition, it thus suffices to find an upper bound on the sum over $k$ of tails of a new distribution: $\mathbb{P}(\mathcal{L}^k>N\varepsilon^2\mid X_0)$. Luckily, we know the expectation of these variables, using \cref{prop: 1D expec lyapunov n+1}, by induction on expectations:
    \begin{align}
        \mathbb{E}(\mathcal{L}^k\mid X_0 ) = \left(1-\frac{2N+1}{3N(N-1)}\right)^k\mathcal{L}^0.
    \end{align}
    
    We can then apply the Markov inequality on each term of the sum to get an upper bound. However, the Markov inequality tends to be of poor quality in the first terms of the summation as it there yields huge unrealistic bounds. We can compensate for this by simply upper-bounding the first terms by $1$. We find that for $k\ge \frac{3N(N-1)}{2N+1}\ln\big(\frac{\mathcal{L}^0}{N\varepsilon^2}\big)$, the Markov inequality provides bounds lower than $1$. Thus we split the summation into two parts, the first $\frac{3N(N-1)}{2N+1}\ln\big(\frac{\mathcal{L}^0}{N\varepsilon^2}\big)$ terms of the sum, that together contribute at most to that amount, and the rest which contributes to at most an infinite geometric series with first term that is quite small. We can show that the second part of the sum can be upper-bounded by $\frac{3N(N-1)}{2N+1}$, which concludes the proof.
\end{proof}

%
%

\begin{corollary}
    \label{cor:1D bound big O N ln N over eps^2}
    For a system evolving according to \cref{eq:1D evolution} and if $I$ is bounded, say $I = \left[a,b\right]$, we have:
    $$\mathbb{E}(T_\varepsilon\mid X_0) 
    \le \frac{3}{2}N\ln\left(\frac{N}{\varepsilon^2}\right) + \frac{3}{2}N\left(\ln\left(\frac{(b-a)^2}{2}\right)+1\right).
    $$
\end{corollary}

\begin{proof}
    This follows immediately from 
    \cref{th:1D finite expected time}
     and using \cref{prop: 1D bounds lyap} we have:
    \begin{equation}
        \mathcal{L}^0\le \frac{N^2}{2}(b-a)^2 
        .
    \end{equation}
\end{proof}

\begin{proposition}
    \label{prop: 1D expect L0 bounded unif iid}
    If $I$ is bounded, say $I = [a,b]$, and if the opinions in $X_0$ have uniform identical independent distributions in $I$, then:
    $$\mathbb{E}(\mathcal{L}^0) = \frac{N(N-1)}{6}(b-a)^2.$$ 
\end{proposition}

\begin{proof}
    The result is straightforward since it uses well known first moments of iid uniform random variables (see \cref{prop: E(x_i) and E(x_i^2)}).
\end{proof}

\begin{theorem}
    \label{th:1D finite expected time expected L0}
    If $I$ is bounded, say $I = [a,b]$, and if the initial opinions in $X_0$ have uniform identical independent distributions in $I$, then:
    $$\mathbb{E}(T_\varepsilon) \le \frac{3}{2}N\ln\left(\frac{N}{\varepsilon^2}\right) + \frac{3}{2}N\left(\ln\left(\frac{(b-a)^2}{6}\right)+1\right).$$
\end{theorem}

\begin{proof}
    The proof is similar to that of \cref{th:1D finite expected time}. To remove the conditioning on $X_0$, we use $\mathbb{E}(T_\varepsilon') = \mathbb{E}\big(\mathbb{E}(T_\varepsilon'\mid X_0)\big)$. We have:
    \begin{align}
        \mathbb{E}(T_\varepsilon') &= \mathbb{E}\big(\mathbb{E}(T_\varepsilon'\mid X_0)\big) \nonumber\\
        &\le \mathbb{E}\left(\sum\limits_{k=0}^\infty \mathbb{P}(\mathcal{L}^k>N\varepsilon^2\mid X_0)\right)  = \sum\limits_{k=0}^\infty \mathbb{E}\big( \mathbb{P}(\mathcal{L}^k>N\varepsilon^2\mid X_0) \big) \label{al:1D inversion E sum in expectation of convergence}\\
        &\le \sum\limits_{k=0}^\infty \mathbb{E}\left(\min\left\{\frac{\mathbb{E}(\mathcal{L}^k\mid X_0)}{N\varepsilon^2},1\right\}\right) \label{al:1D sum of E of min conditioned inexpectation of convergence}\\
        &\le \sum\limits_{k=0}^\infty \min\left\{\frac{\mathbb{E}(\mathbb{E}(\mathcal{L}^k\mid X_0))}{N\varepsilon^2},1\right\} = \sum\limits_{k=0}^\infty \min\left\{\left(1-\frac{2N+1}{3N(N-1)}\right)^k\frac{\mathbb{E}(\mathcal{L}^0)}{N\varepsilon^2},1\right\} \label{al:1D sum of min of E conditioned inexpectation of convergence},
    \end{align}
    where the inversion in \cref{al:1D inversion E sum in expectation of convergence} is achieved by positivity of the terms.
    We can thus replace $\mathcal{L}^0$ by its expectation when removing the conditioning. We then use \cref{prop: 1D expect L0 bounded unif iid} and continue the proof as in \cref{th:1D finite expected time} to get the desired result.
\end{proof}

\begin{theorem}
	\label{th: 1D range bounded by lyap}
	For a system evolving according to \cref{eq:1D evolution}, we have:
	$$ \left(1-\frac{2N+1}{3N(N-1)}\right)^k \frac{2\mathcal{L}^0}{N^2} \le \mathbb{E}(\max\limits_{i\neq j}\left|x_i^k-x_j^k\right|^2\mid X_0) \le \left(1-\frac{2N+1}{3N(N-1)}\right)^k \frac{\mathcal{L}^0}{N}.$$
\end{theorem}

\begin{proof}
	Using \cref{prop: 1D bounds lyap}, if we denote $r_k = \max\limits_{i\neq j}\left|x_i^k-x_j^k\right|$ the range of opinions at step $k$, then:
	\begin{equation}
		\frac{2}{N^2}\mathcal{L}^k \le r_k^2 \le \frac{\mathcal{L}^k}{N}.
	\end{equation}
	
	We get the final result by taking the expectation and applying \cref{prop: 1D expec lyapunov n+1}.
\end{proof}

\begin{theorem}
	If $I$ is bounded, say $I = [a,b]$, and if the opinions in $X_0$ have uniform identical independent distributions in $I$, then if $r_k$ is the range at time step $k$:
	$$\frac{(N-1)\left(1-\frac{2N+1}{3N(N-1)}\right)^k}{3N}(b-a)^2 \le \mathbb{E}(r_k^2\mid X_0) \le \frac{(N-1)\left(1-\frac{2N+1}{3N(N-1)}\right)^k}{6}(b-a)^2.$$
\end{theorem}

\begin{proof}
	The result immediately follows \cref{th: 1D range bounded by lyap,prop: 1D expect L0 bounded unif iid}.
\end{proof}

\begin{theorem}
    \label{th: 1D conv to single point}
    A system evolving according to \cref{eq:1D evolution} converges to a single point $x_\infty\in I$ almost surely.
\end{theorem}

\begin{proof}
    The result immediately follows from \cref{th:1D finite expected time,prop: 1D max dist non increasing}. Note that the limit point $x_\infty$ is random in $I$. 
\end{proof}

Denote for conciseness $\bar{X}_k = \frac{1}{N}\sum\limits_{i=1}^N x_i^k\in\mathbb{R}$ the average opinion at step $k$.

\begin{proposition}
    \label{prop: 1D expec X_bar}
    For all $k\in\mathbb{N}$, we have $\mathbb{E}(\bar{X}_k\mid X_0) = \bar{X}_0$.    
\end{proposition}

\begin{proof}
	The proof is straightforward and uses the well-known first moment of uniform random variables (see supplementary material \cref{prop: E(x_i) and E(x_i^2)}). See supplementary material \cref{sm: proof 1D expec X_bar} for a detailed proof.
\end{proof}

Denote $1_N = (1,\cdots,1)^T$ the vector of 1$(\beta_2,\cdots,\beta_N)\in\mathbb{R}^{N-1}$ an arbitrary orthonormal basis of the space orthogonal to the one dimensional space generated by $1_N$. Define columnwise the following unitary matrix $U = \begin{pmatrix} \frac{1}{\sqrt{N}}1_N & \beta_2 & \cdots & \beta_N  \end{pmatrix}$. Let $\diag(\lambda_1,\cdots,\lambda_N)$ be the diagonal matrix with entries $\lambda_1,\cdots,\lambda_N$.

\begin{proposition}
    \label{prop: 1D expect X_n}
    For all $k\in\mathbb{N}$, we have:
    $$\mathbb{E}(X_k\mid X_0) = U\diag\Bigg(1,\left(1-\frac{1}{N-1}\right)^k,\cdots,\left(1-\frac{1}{N-1}\right)^k\Bigg)U^T X_0.$$ 
\end{proposition}

\begin{proof}
	The proof is also straightforward by working conditionally to the choice of pair $(i,j)$ for update. We find that:
	\begin{equation}
		\mathbb{E}(X_{k+1}\mid X_k) = U\diag\Bigg(1,\left(1-\frac{1}{N-1}\right),\cdots,\left(1-\frac{1}{N-1}\right)\Bigg)U^T X_k,
	\end{equation}
	which gives the final result by induction. A detailed proof is given in the supplementary material in \cref{sm: proof 1D expect X_n}.
\end{proof}

\begin{theorem}
    \label{th: 1D Expec limit}
    The limit point for a system evolving according to \cref{eq:1D evolution} has the expectation:
    $$\mathbb{E}(x_\infty\mid X_0) = \bar{X}_0 =  \frac{1}{N}\sum\limits_{i=1}^N x_i^0.$$
\end{theorem}

\begin{proof}
    Using \cref{th: 1D conv to single point}, all opinions converge almost surely to the same finite but random value, and since the opinions are all bounded by the initial interval $I_0 = [\min X_0,\max X_0]$, we have by bounded convergence that:
    \begin{equation}
        \label{eq: 1D limit expec xi AND 1D invert limit expec x bar}
        \begin{cases}
            \lim\limits_{k\xrightarrow[]{}\infty} \mathbb{E}(x_i^k\mid X_0) = \mathbb{E}(x_\infty\mid X_0), \\
            \lim\limits_{k\xrightarrow[]{}\infty} \mathbb{E}(\bar{X_k}\mid X_0) = \mathbb{E}(x_\infty \mid X_0).
        \end{cases}
    \end{equation}
    We then conclude using \cref{eq: 1D limit expec xi AND 1D invert limit expec x bar} and either of \cref{prop: 1D expect X_n,prop: 1D expec X_bar}..
\end{proof}

Note the following important remark: while the proof of \cref{th: 1D Expec limit} using \cref{prop: 1D expect X_n}
gives that the convergence of the expectations of each opinion is exponentially fast, it does not provide any guarantee on the speed of the convergence of sequences $(X_0,X_1,X_2,\cdots)$. This is given by \cref{th:1D finite expected time}.

We can compare our result with those from the gossip literature. Note that in gossiping, $x_\infty$ is deterministic and equals to $\frac{1}{N}\sum_{i=1}^N x_i^0$  almost surely and that the communication graph may be assumed different from the complete graph, leading to different convergence times.

\begin{definition}
	For any $\varepsilon>0$, we denote $T_{gossip}(\varepsilon)$ the ``$\varepsilon$-averaging time'', which is the deterministic quantity studied in the gossip algorithms' literature and used there as the convergence time, defined as:
	$$ T_{gossip}(\varepsilon) = \sup\limits_{X_0\in I^N} \inf\limits_{k\in\mathbb{N}} \left\{k \mid   \mathbb{P}\left(\frac{\lVert X_k -  x_\infty \begin{pmatrix} 1,\cdots,1\end{pmatrix}^T \rVert_2}{\lVert X_0\rVert_2} \ge \varepsilon \mid X_0\right) \le \varepsilon  \right\}$$
\end{definition}

\begin{theorem}
	\label{th: 1D gossip}
	Assume $I$ is bounded, say $I = \left[a,b\right]$, then:
	$$ T_{gossip}(\varepsilon) \le  \frac{-3\ln\varepsilon +\ln\left(2(N-1)(1-q_{a,b})\right)}{-\ln\left( 1-\frac{2N+1}{3N(N-1)}\right)} \le \frac{3}{2}N\ln\left(\frac{N}{\varepsilon^3}\right) +  \frac{3}{2}N\ln\left(2\left(1-q_{a,b}\right)\right) $$
	where $q_{a,b} = \frac{a^2}{b^2}\mathbbm{1}_{a\ge0} + \frac{b^2}{a^2}\mathbbm{1}_{b\le 0}$.
\end{theorem}

\begin{proof}
	We here give an overview of the proof, for which a detailed one can be found in the supplementary material \cref{sm: proof 1D gossip}.
	
	Using \cref{prop: 1D bounds lyap}, we can link the squared range to the Lyapunov quantity. Thus the tail distribution in the definition of $T_{gossip}(\varepsilon)$ can be upper-bounded by a tail distribution of $\mathcal{L}^k$. We can then apply the Markov inequality to get an upper bound of this tail distribution. Unlike us, the Gossip literature is solely interested in one tail distribution: the one that passes the $\varepsilon$ threshold. It thus suffices to find for which $k$ the bound given by the Markov inequality is smaller than $\varepsilon$.
\end{proof}

Our result is similar to those in the gossip literature. However the quantity we study is slightly different, and we need to add the result of the Markov inequality at all levels and not just at a specific level depending on $\varepsilon$. However, if unconventionally we change the definition of $T_{gossip}(\varepsilon)$ by using the squared 2-norms instead of the plain 2-norms, then the same calculations would give an upper bound with dominant term $\frac{3}{2}N\ln\left(\frac{N}{\varepsilon^2}\right)$. This suggests that the unconventional gossip convergence time is similar to our expected convergence time. This is a general result that is due to the fact that the quantity of interest, for us and for gossiping, is $\mathcal{L}^n$, and that it is an exponentially decreasing positive super-martingale. We can prove that if $(Y_k)$ is a positive exponentially decreasing super-martingale, i.e. there is $\alpha>0$ such that $\mathbb{E}(Y_{k+1}\mid Y_k) = (1-\alpha) Y_k$, and if $Y_k$ converges to $0$ almost surely, then if we slightly change the convergence definition of $T_\varepsilon$ to $T_{cv}(\varepsilon)$ by normalising by $Y_0$, i.e. looking at the threshold $\frac{Y_k}{Y_0}\le \varepsilon$, then $\mathbb{E}(T_{cv}(\varepsilon)\mid Y_0) \le \frac{-\ln\varepsilon}{-\ln(1-\alpha)} + \frac{1}{\alpha}$ and $T_{gossip}(\varepsilon) \le \frac{-\ln\varepsilon}{-\ln(1-\alpha)}$. Thus both quantities of interested have similar upper bounds. In our case, $Y_k$ would be $\mathcal{L}^k$, up to multiplicative factors, different in each case, in order to get a comparison with the same threshold $\varepsilon$. A proof of this similarity between expected convergence time and gossiping is given as supplementary material in \cref{sm: comparison expect cv time and gossip}.

\section{Unconstrained \texorpdfstring{\boldmath$D$}{D}-dimensional case}

We can generalise the previous one-dimensional model to a similar $D$-dimensional one. Now, we assume that people's opinions depend on several parameters. However, we will say that the space of opinions is ``unconstrained'' in the sense that the domain for opinions is a convex set of $\mathbb{R}^D$. This model suits well cases when extreme opinions correspond to at least one ``extreme'' parameter, i.e. very big or very small, and any non extreme point is a feasible opinion. Once again, the space of opinions can either be bounded or unbounded. 

Mathematically, we define the space of opinions to be a convex set $C\subset \mathbb{R}^D$ of a $D$-dimensional vector space. Our system comprises of $N\ge 2$ indistinguishable agents, each with their own opinion: $x_i = (x_{i,1},\cdots,x_{i,D})^T\in C$ for $i\in\{1,\cdots,N\}$. The initial distribution of the opinions is a given $x_i^0\in C$. Opinion dynamics are modelled in discrete time, conditionally to the state of opinions at the previous time step $X_k = (x_1
^k,\cdots,x_N^k)^T\in\mathbb{R}^{N\times D}$. The evolution law at time step $k+1$ for all $i\in \{1,\cdots,N\}$ is:
\begin{equation}
    \label{eq:nD evolution}
    x_i^{k+1} = \mathbbm{1}_{i\notin\{A_{k+1},B_{k+1}\}}x_i^k + \mathbbm{1}_{i=A_{k+1}}U_{1}^{k+1}+\mathbbm{1}_{i=B_{k+1}}U_{2}^{k+1},
\end{equation}
where $(A_{k+1},B_{k+1})$ is a random uniform sampling of two indices of $\{1,\cdots,N\}$ without replacement independent of the past, and where conditionally to $X_k$, $A_{k+1}$, and $B_{k+1}$, $U_1^{k+1}$ and $U_2^{k+1}$ are independent random uniform variables in the $1$-dimensional interval $\{(1-\lambda)x_{A_{k+1}}^k+\lambda x_{B_{k+1}}^k\mid \lambda\in[0,1]\}$. Concretely, at each time step, two random agents $A_{k+1}$ and $B_{k+1}$ are selected and they then independently and uniformly resample their opinion in the interval between both previous opinions. See \cref{fig: ND evolution convex} for an example when $D=2$.

\begin{figure}[tbhp]
  \centering
    \includegraphics[width=\textwidth]{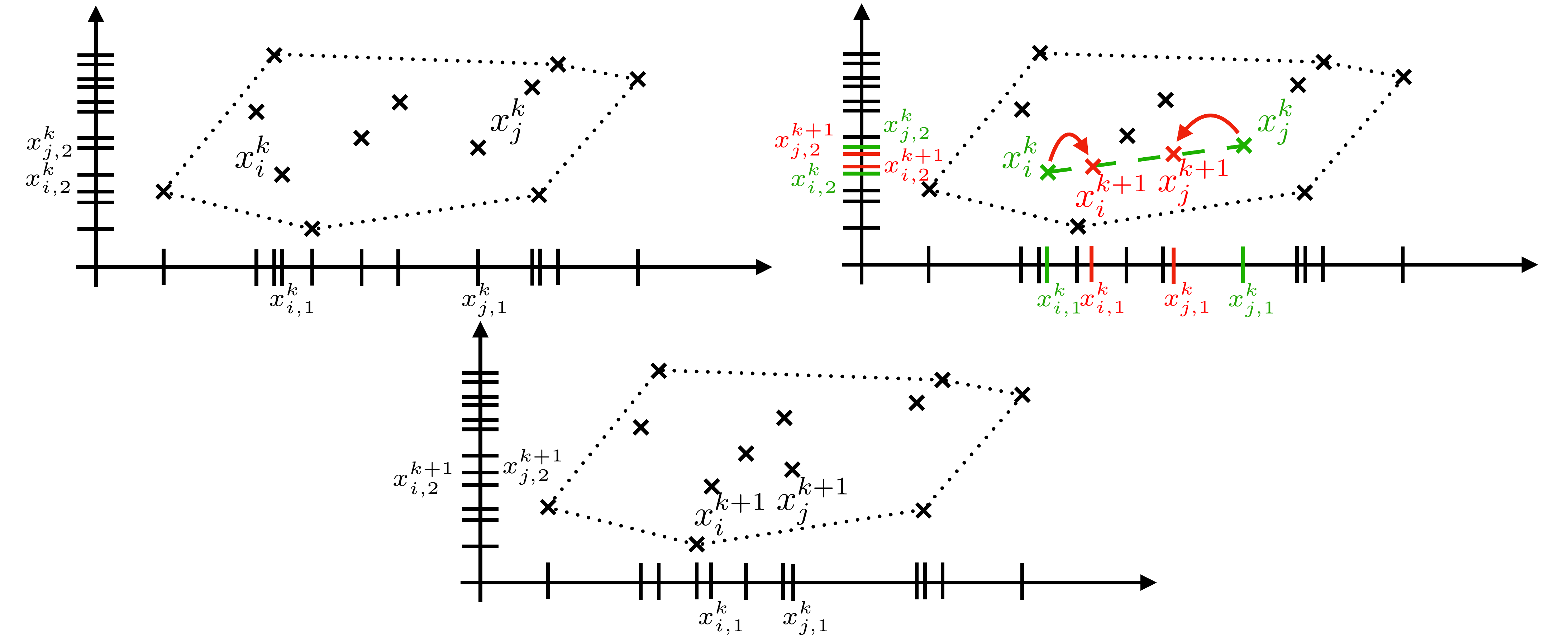}
    \caption{One-step opinion evolution in the unconstrained $2$-dimensional model.}
    \label{fig: ND evolution convex}
\end{figure}

We will conserve the same notations as in the $1$-dimensional case for simplicity unless explicitly mentioned otherwise. 

\begin{definition}
    Denote $\mathcal{L}_d^k$, for $d\in\{1,\cdots,D\}$ and $\mathcal{L}_T^k$ the studied quantities:
    \begin{align}
        \mathcal{L}_d^k &= \sum\limits_{i\neq j}(x_{i,d}^k-x_{j,d}^k)^2 \\
        \mathcal{L}_T^k &= \sum\limits_{d=1}^D \mathcal{L}_d^k =  \sum\limits_{i\neq j}\lVert x_i^k-x_j^k\rVert_2^2.
    \end{align}
\end{definition}

\begin{proposition}
    \label{prop: ND bounds lyap}
    For all $k\in\mathbb{N}$, we have:
    \begin{equation}
        N\max\limits_{i\neq j}\lVert x_i^k - x_j^k \rVert_2^2 \le \mathcal{L}_T^k \le \frac{N^3}{2}\max\limits_{i\neq j}\lVert x_i^k - x_j^k \rVert_2^2,
    \end{equation}
    where only the lower bound is tight.
\end{proposition}

\begin{proof}
	We here give an overview of the proof, for which a detailed one can be found in the supplementary material \cref{sm: proof ND bounds lyap}.

	Each coordinate of the system, i.e. each column of $X_k$, follows the $1D$ motion rule from \cref{eq:1D evolution}. We can sum up those bounds and then use classic norm equivalence bounds to get the desired result. The lower bound is tight as pair-wise distances are smaller along orthogonal projections, in particular along the line between the two most distant points, where the problem then comes down to the one dimensional one. On the other hand, the upper bound is highly pessimistic as simultaneous maximisation of the sum of squared differences along each dimension is not possible in an intersection of $2$-balls in dimensions $D\ge 2$.
\end{proof}

\begin{definition}
    For any $\varepsilon > 0$, we denote $T_\varepsilon$ the stopping time with respect to the natural filtration induced by the $\left(X_k\right)$ sequence defined as:
    $$ T_\varepsilon = \min\left\{ k\in\mathbb{N} \mid \max\limits_{i\neq j} \lVert x_i^k-x_j^k \rVert_2 \le \varepsilon\right\}. $$
\end{definition}

\begin{definition}
    For any $\varepsilon > 0$, we denote $T_\varepsilon'$ the stopping time with respect to the natural filtration induced by the $\left(X_k\right)$ sequence defined as:
    $$ T_\varepsilon' = \min\left\{ k\in\mathbb{N} \mid \mathcal{L}_T^k \le N\varepsilon^2\right\}. $$
\end{definition}

\begin{proposition}
    \label{prop: ND T<=T'}
    For all $\varepsilon>0$, $T_\varepsilon\le T_\varepsilon'$.
\end{proposition}

\begin{proof}
    The result immediately follows from \cref{prop: ND bounds lyap}. If $\mathcal{L}_T^k\le N\varepsilon^2$, then the range of all points is smaller than $\varepsilon$ and thus we have reached convergence. Therefore the first occurrence of the event $\{\mathcal{L}_T^k\le N\varepsilon^2\}$ happens at the same time or later than the first occurrence of $\{\max\limits_{i\neq j}\lVert x_i^k - x_j^k \rVert_2 \le \varepsilon\}$, i.e. of convergence.
\end{proof}

\begin{theorem}
    \label{th:nD finite expected time}
    For a system evolving according to \cref{eq:nD evolution}, for any $\varepsilon>0$, we have:
    $$\mathbb{E}(T_\varepsilon\mid X_0) \le \frac{3N(N-1)}{2N+1} \ln\Bigg(\frac{\sum\limits_{d=1}^D\mathcal{L}_d^0}{N\varepsilon^2}\Bigg) +\frac{3N(N-1)}{2N+1} \le \frac{3}{2}N \ln\Bigg(\frac{\sum\limits_{d=1}^D\mathcal{L}_d^0}{N\varepsilon^2}\Bigg) + \frac{3}{2}N.$$
\end{theorem}

\begin{proof}
	We here give an overview of the proof, for which a detailed one can be found in the supplementary material \cref{sm: proof nD finite expected time}.

	The proof is in essence similar to the one in the one dimensional case \cref{th:nD finite expected time}. Using \cref{prop: ND T<=T'}, we have that, by contraposition, if we have not reached convergence by step $k$, then $\mathcal{L}_T^k > N\varepsilon^2$. This implies that the probability of the event $\{\mathcal{L}_T^k > N\varepsilon^2\}$ is larger than of the event $T_\varepsilon > k$. It thus suffices to study a sum over $k$ of distribution tails: $\mathbb{P}(\mathcal{L}_T^k>N\varepsilon^2\mid X_0)$. We can bound these terms with the Markov inequality. Using the linearity of the expectation, we get a bound proportional to the sum of the expectations along each dimension. As the evolution along each dimensional is purely equivalent to the one-dimensional one, we can apply \cref{prop: 1D expec lyapunov n+1} for each dimension and sum them, regardless of any dependency issue between the dimensions. Once again, the Markov inequality provides unrealistic bounds for small $k$ and we instead use the universal $1$ bound for those terms. We find that for $k$ larger than $\frac{3N(N-1)}{2N+1}\ln\big(\nicefrac{\sum\limits_{d=1}^D\mathcal{L}_d^0}{N\varepsilon^2}\big)$ the Markov inequality yields a bound lower than $1$. We thus get that the expectation is bounded by this amount plus a geometric series with very small initial value and we can prove that the series is smaller than $\frac{3N(N-1)}{2N+1}$.
\end{proof}

%

\begin{corollary}
    \label{cor:nD bound big O DN ln DN over eps^2}
    For a system evolving according to \cref{eq:nD evolution} and if $C$ is a bounded $D$-dimensional cube, say $C = \left[a,b\right]^D$, we have:
    $$\mathbb{E}(T_\varepsilon\mid X_0) 
    \le \frac{3}{2}N\ln\left(\frac{DN}{\varepsilon^2}\right) +\frac{3}{2}N\left(\ln\left(\frac{(b-a)^2}{2}\right)+1\right). $$
\end{corollary}

\begin{proof}
    This follows immediately from 
    \cref{th:nD finite expected time}
     and using \cref{prop: 1D bounds lyap} for all $k$:
    \begin{equation}
        \mathcal{L}_k^0\le \frac{N^2}{2}(b-a)^2 = O_{N\xrightarrow[]{}\infty}(N^2).
    \end{equation}
\end{proof}

\begin{theorem}
    \label{th:nD finite expected time expected bounded}
    If $C$ is a bounded $D$-dimensional cube, say $C = \left[a,b\right]^D$, and if the opinions in $X_0$ have uniform identical independent distributions in $C$, then for a system evolving according to \cref{eq:nD evolution}, we have:
    $$\mathbb{E}(T_\varepsilon\mid X_0) \le \frac{3}{2}N \ln\left(\frac{DN}{\varepsilon^2}\right) +\frac{3}{2}N\left(\ln\left(\frac{(b-a)^2}{6}\right)+1\right).$$
\end{theorem}

\begin{proof}
    Similarly to the proof of \cref{th:1D finite expected time expected L0}, we redo the proof of \cref{th:nD finite expected time} and take the expectation of expectations to remove the conditioning on $X_0$ and find  a new threshold for the comparison of the probabilities given by the Markov inequality to one. We will end up comparing, for summation:
    \begin{equation}
        \frac{1}{N\varepsilon^2}\left(1-\frac{2N+1}{3N(N-1)}\right)^k \sum\limits_{d=1}^D\mathbb{E}(\mathcal{L}_d^0) \le 1.
    \end{equation}
    
    For this reason, we find that we can replace each $\mathcal{L}_d^0$ by its expectation to get the desired bound. Since $X_0$ follows a uniform distribution in a cube, its projections along the dimension of the cube also follow a uniform distribution. Thus each column of $X_0$ follows a uniform distribution in $\left[a,b\right]$. Therefore the expectation of $\mathcal{L}_d^0$ is given by \cref{prop: 1D expect L0 bounded unif iid}. 
\end{proof}

\begin{theorem}
	\label{th: nD range bounded by lyap}
	For a system evolving according to \cref{eq:nD evolution}, we have:
	$$ \frac{2 \left(1-\frac{2N+1}{3N(N-1)}\right)^k}{N^3}\sum\limits_{d=1}^D\mathcal{L}_d^0 \le \mathbb{E}(\max\limits_{i\neq j}\lVert x_i^k-x_j^k\rVert_2^2\mid X_0) \le \frac{\left(1-\frac{2N+1}{3N(N-1)}\right)^k}{N}\sum\limits_{d=1}^D\mathcal{L}_d^0.$$
\end{theorem}

\begin{proof}
	This result is a direct consequence of taking the expectation on \cref{prop: ND bounds lyap} and since each column of $X_k$ follows the one-dimensional motion rule from \cref{eq:1D evolution} the expectation of $\mathcal{L}_d^k$ is given by \cref{prop: 1D expec lyapunov n+1}.
\end{proof}

\begin{theorem}
    If $C$ is a bounded $D$ dimensional cube, say $C = \left[a,b\right]^D$, and if the opinions in $X_0$ have uniform identical independent distributions in $C$, then for a system evolving according to \cref{eq:nD evolution}, if we denote $r_k$ the range at time step $k$, and $r_C = \left|b-a\right|$ the diameter of $C$ in infinite norm, we have:
	$$\frac{N-1}{3N^2}\left(1-\frac{2N+1}{3N(N-1)}\right)^k Dr_C^2 \le \mathbb{E}(r_k^2\mid X_0) \le \frac{N-1}{6}\left(1-\frac{2N+1}{3N(N-1)}\right)^k Dr_C^2.$$
\end{theorem}

\begin{proof}
	As $r_k = \max\limits_{i\neq j}\lVert x_i^k-x_j^k\rVert_2^2$, the result immediately follows \cref{th: nD range bounded by lyap,prop: 1D expect L0 bounded unif iid}.
\end{proof}

\begin{theorem}
    A system evolving according to \cref{eq:nD evolution} converges to a single point $x_\infty\in C$ almost surely.
\end{theorem}

\begin{proof}
    The result immediately follows from \cref{th:nD finite expected time,prop: 1D max dist non increasing}. Note that the limit point $x_\infty$ is random in $C$.
\end{proof}

\begin{theorem}
    \label{th: nD Expec limit}
    The limit point for a system evolving according to \cref{eq:nD evolution} has the expectation $\mathbb{E}(x_\infty\mid X_0) = \bar{X_0} \triangleq \frac{1}{N}\sum\limits_{i=1}^N x_i^0\in\mathbb{R}^D$.
\end{theorem}

\begin{proof}
    The result immediately follows from \cref{th: 1D Expec limit}: each component converges on average to the initial average of that component, thus the opinions converge on average to the initial average opinion.
\end{proof}

\section{Constrained 2-dimensional case}

The limitation of the previous model is the convex opinion space assumption, which is well adapted to situations where ``extreme'' opinions correspond to at least one ``big'' opinion parameter. However, in some cases, it is more accurate to also consider some ``extreme'' cases with neither parameter being ``big''. This happens when the opinion space is no longer convex.

For non convex opinion spaces, it is then necessary to redefine how agents interact. In the convex case, we modelled an interaction along the straight line linking the two states. In a non-convex, yet arc-connected space, a reasonable possibility to model interactions between opinions is to consider a geodesic between the opinions in the opinion space. In this paper we will study a simple case that naturally generalises the previous models: the unit circle, which is interesting from two aspects. First it is non convex in $\mathbb{R}^2$. Secondly, a reparametrization of $S$ using the oriented angle from the $x$-axis in $[0,2\pi)$ leads to a new parameter space $[0,2\pi)$ for $S$ which is convex. However, it fundamentally differs from the previous convex models for the two following reasons. First, as opinions communicate along geodesics of $S$, if $(\theta_1,\theta_2)\in [0,2\pi)^2$, then, depending on the size of $\left|\theta_2-\theta_1\right|$, communication can happen in the $[\min\{\theta_1,\theta_2\},\max\{\theta_1,\theta_2\}]$ interval or in its closed complement in $ [0,2\pi)$ which is equal to $[\max\{\theta_1,\theta_2\},2\pi)\cup[0,\min\{\theta_1,\theta_2\}]$. Thus convexity in the parameter space is not enough to use the previous models as we require convexity in the embedding space. Secondly, if $\theta\in[0,2\pi)$ increases, then as $\theta$ reaches the right boundary of the interval then $\theta$ simultaneously reaches the left boundary as well. This violates the principle that we assumed in the previous cases where a more and more 
``right-wing'' opinion could not simultaneously become more and more ``left-wing''.

Mathematically, we define the space of opinions to be $S\in\mathbb{R}^2$ the unit circle embedded in $2$-dimensional space. Our system comprises of $N\ge 2$ indistinguishable agents, each with their own opinion: $x_i = (x_{i,1},x_{i,2})^T$ following the circle constraint $x_{i,1}^2+x_{i,2}^2=1$. The opinion state is $X_k = (x_1^k,\cdots,x_N^k)^T\in\mathbb{R}^{N\times2}$. It will be useful to consider the equivalent reparametrization by angles $\theta_i\in[0,2\pi)$ with $(x_{i,1},x_{i,2}) = (\cos(\theta_i),\sin(\theta_i))$. The initial distribution of the opinions is a given $\theta_i^0\in [0,2\pi)$. Here the state-evolution dynamics is modelled in discrete time, conditionally to the state of opinions at the previous time step $\Theta_k = (\theta_1^k,\cdots,\theta_N^k)^T\in \mathbb{R}^{N}$. The evolution law at time step $k+1$ for all $i\in \{1,\cdots,N\}$ is now done along geodesics and is:
\begin{equation}
    \label{eq:circle evolution}
    \theta_i^{k+1} = \mathbbm{1}_{i\notin\{A_{k+1},B_{k+1}\}}\theta_i^k + \mathbbm{1}_{i=A_{k+1}}U_1^{k+1} + \mathbbm{1}_{i=B_{k+1}}U_2^{k+1},
\end{equation}
where $(A_{k+1},B_{k+1})$ is a random uniform sampling of two indices of $\{1,\cdots,N\}$ without replacement independent of the past, and where conditionally to $\Theta_k$, $A_{k+1}$, and $B_{k+1}$, $U_1^{k+1}$ and $U_2^{k+1}$ are independent random uniform variables in $G(\theta_{A_{k+1}},\theta_{B_{k+1}})\subset [0,2\pi)$, which is the geodesic circle arc  between opinions $x_{A_{k+1}}$ and $x_{B_{k+1}}$:
\begin{equation}
    G(\theta,\tilde{\theta}) =
    \begin{cases}
        \left[\min\{\theta,\tilde{\theta}\},\max\{\theta,\tilde{\theta}\}\right] &\text{if } \left|\tilde{\theta} - \theta\right|\le \pi\\
        \left[\max\{\theta,\tilde{\theta}\},2\pi\right)\cup\left[0,\min\{\theta,\tilde{\theta}\}\right] &\text{if } \left|\tilde{\theta} - \theta\right|> \pi
    \end{cases}
\end{equation}

Concretely, at each time step, two random agents $A_{k+1}$ and $B_{k+1}$ are selected and they then independently and uniformly resample their opinion on the shortest circle arc between both previous opinions. See \cref{fig: ND evolution circle} for an example.

Note that for the pathological case of two agents at an angular distance of exactly $\pi$, then we chose a deterministic geodesic. This work would be similar if we chose a random geodesic in that case and even if we chose for the two agents to not necessarily choose the same one. This is because this $\pi$ distance configuration almost surely never happens, except for the eventual cases in $\Theta_0$ were angles are initially set to be at such a distance.

\begin{figure}[tbhp]
  \centering
    \includegraphics[width=\textwidth]{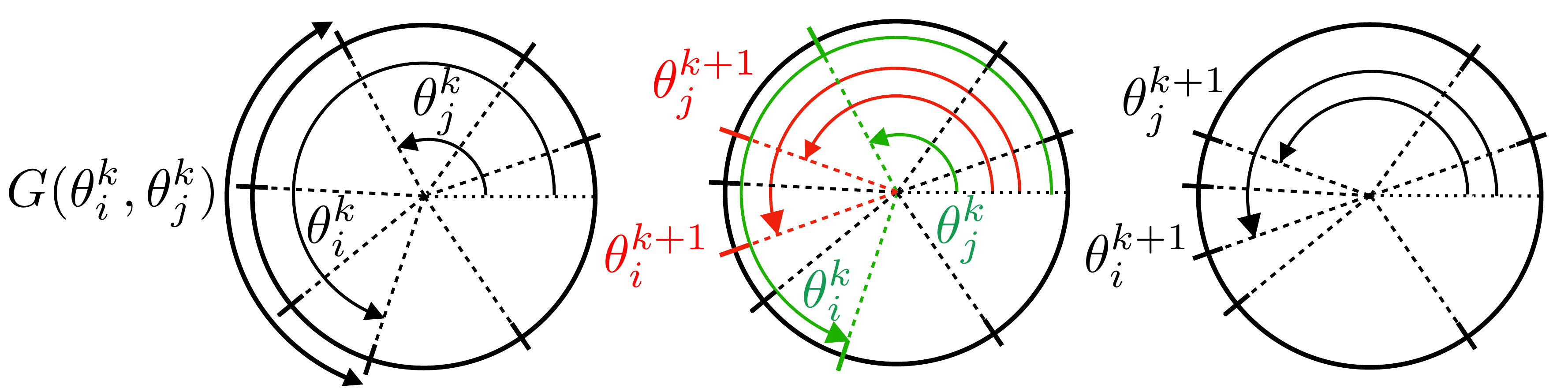}
    \caption{One-step opinion evolution in the constrained $2$-dimensional model.}
    \label{fig: ND evolution circle}
\end{figure}

Similarly to the convex case, we define the following stopping time on the angle parametrization.

\begin{definition}
    \label{def: circle stopping time convergence}
    For any $\varepsilon>0$, we denote $T_\varepsilon$ the stopping time, with respect to the natural filtration induced by the $(\Theta_k)$ sequence, defined as the first time step when all unit vector opinions are within a circle arc of length $\varepsilon$. For $\varepsilon<\frac{2\pi}{3}$, this means:
    $$T_\varepsilon =  \min\Big\{k\in\mathbb{N}\mid \max\limits_{i\neq j} \min\left\{\left|\theta_i^k - \theta_j^k\right| , 2\pi-\left|\theta_i^k - \theta_j^k\right|\right\} \le \varepsilon \Big\}.$$
\end{definition}

\begin{definition}
    \label{def: circle stopping time half disk}
    We denote $T_{\textit{HD}}$ the stopping time, with respect to the natural filtration induced by the $(\Theta_k)$ sequence, of the event that all unit vector opinions are within a half-disk:
    $$T_{\textit{HD}} =  \min\{k\in\mathbb{N}\mid \exists\theta_{\textit{HD}}^k\in\left[0,2\pi\right), \forall i\in\{1,\cdots,N\}, \cos(\theta_i^k-\theta_{HD}^k) > 0\}.$$
\end{definition}

\begin{proposition}
    \label{prop: circle half disk stable}
    For any system evolving according to \cref{eq:circle evolution}, if at time step $k\in \mathbb{N}$ all unit vector opinions are within a half-disk, then for all $k'\ge k$, all unit vector opinions are within a half-disk. 
\end{proposition}

\begin{proof}
    Let $\theta_{\textit{HD}}^k$ be an angle such that for all $i$, $\cos(\theta_i^k-\theta_{\textit{HD}}^k)>0$. Perform the change of angular parametrization around $\theta_{\textit{HD}}^k$ for all future time steps $k'\ge k$:
    \begin{equation}
        \tilde{\theta}_i^{k'} = \theta_i^{k'} - \theta_{\textit{HD}}^k \mod 2\pi.
    \end{equation}
    
    In this new parametrization, all geodesics are contained within the geodesic between the two most extreme opinions, i.e.  $G(\tilde{\theta}_i^k,\tilde{\theta}_j^k) \subset G\big(\min\{\tilde{\theta}_l^k\},\max\{\tilde{\theta}_l^k\}\big)$ for all $(i,j)\in\{1,\cdots,N\}^2$. Due to the rules of motion \cref{eq:circle evolution}, this in turn implies that all angles at the next step are within that same geodesic. By induction, we can claim that for all future time steps $k'\ge k$, $\tilde{\theta}_i^{k'} \in G\big(\max\{\tilde{\theta}_l^k\},\min\{\tilde{\theta}_l^k\}\big)$, which implies that all angles are contained within a half-disk forever.
\end{proof}

\begin{theorem}
    \label{th:circle to half disk finite expec time optimised}
    Any system evolving according to \cref{eq:circle evolution} has unit vector opinions within a half-disk in finite expected time. In particular:
    $$\mathbb{E}(T_{\textit{HD}}\mid \Theta_0) \le \Big(\frac{27}{4}N^2(N-1)^2\Big)^{\floor{\frac{N}{2}}} + 2\floor[\Big]{\frac{N}{2}}.$$
\end{theorem}

\begin{proof}
    The proof consists in showing the existence of a sequence of events with probability lower bounded by a strictly positive constant that drives the system from any configuration that is not contained in a half-disk to a configuration contained in one. To do this, is suffices to find a finite sequence of events leading to the half-disk configuration from any other configuration. By finiteness of the problem, each of these events will have lower bounded probabilities, and since their succession is finite, we have a non zero lower bound for the probability to have such a sequence occur from any given configuration. In turn this gives that almost surely all agents will be located within a half-disk and that the expected time for this to occur is finite. A detailed proof is given as supplementary material in \cref{sm: proof circle to half disk finite expec time optimised}.   
\end{proof}

\begin{definition}
    Let $\mathcal{B}_N^{HD}$ be the set of finite upper bounds of $\mathbb{E}(T_{HD}\mid \Theta_0)$:
    $$ \mathcal{B}_N^{HD} = \left\{B_N^{HD}\in\mathbb{R}\mid \mathbb{E}(T_{HD} \mid \Theta_0)\le B_N^{HD} \right\}. $$
\end{definition}

\begin{theorem}
    \label{th: circle finite expec cv}
    Any system evolving according to \cref{eq:circle evolution}, for any $\varepsilon >0$, has unit vector opinions within a circle arc of angle $\varepsilon$ in finite expected time. In particular, for any $B_N^{HD}\in\mathcal{B}_N^{HD}$:
    $$\mathbb{E}(T_{\varepsilon}\mid \Theta_0) \le B_N^{HD} + \frac{3}{2}N\ln\left(\frac{N}{\varepsilon^2}\right)+ \frac{3}{2}N\left(\ln{\frac{\pi^2}{2}} + 1\right).$$
\end{theorem}

\begin{proof}
    The result immediately follows from
    \cref{th:circle to half disk finite expec time optimised,th:1D finite expected time}. Indeed, it suffices to notice that once all agents are within a half-disk, then the dynamics of the system using \cref{eq:circle evolution} is equivalent to the one-dimensional dynamics \cref{eq:1D evolution} using the angles for the opinions. Note that the ``initial'' one dimensional Lyapunov once we have reached the state where all unit vector opinions are within a half-disk, that is the Lyapunov at time step $T_{\textit{HD}}$, is a random value, however, since the pair-wise angular distance is then less than $\pi$ for any pair of opinions, we have bounded it by:
    \begin{equation}
        \mathcal{L}^{T_{\textit{HD}}} \le \frac{N^2}{2}\pi^2.
    \end{equation}
\end{proof}

\section{Open problems on the constrained 2-dimensional case}
\label{sec: open problems for the circle}

Many issues remain unsolved for the constrained 2-dimensional case. We propose them in this paper as open questions. The main problem was to obtain a better bound than the crude $O((\frac{3\sqrt{3}}{2})^N N^{2N})$ provided in 
\cref{th:circle to half disk finite expec time optimised}
for the expectation of the time for all agents to get within a half-disk. We provide three interesting approaches based on different quantities for which we do not have a final solution. Details can be found in the supplementary materials \cref{sm: Open problems on the constrained 2-dimensional case}.
 
The first approach consists in studying the vector sum of all the unit vector opinions $\mathcal{S}^k = \sum_{i=1}^N x_i^k$. The purpose of studying this vector is that convergence of opinions in $S$ is equivalent to convergence of $\mathcal{S}^k$ in $\mathbb{R}^2$ and of convergence of its $2$-norm to its upper bound, $N$, by finiteness of the problem. Intuitively and experimentally, if $\left\lVert \mathcal{S}^k\right\rVert_2^2$ is ``large'', then there is a ``large'' number of opinions positively oriented with $\mathcal{S}^k$, and furthermore opinions positively oriented with $\mathcal{S}^k$ tend to be updated in a way that further increases the norm of $\mathcal{S}^k$. However $\mathcal{S}^k$ is upper-bounded by $N$ which can only happen for opinions arbitrarily close to each other. Therefore we can simply study the evolution of $\left\lVert \mathcal{S}^k\right\rVert_2^2$, which is an upper-bounded random real quantity and show that it converges to its upper-bound and study its speed of convergence. An other possibility would be to analyse $\langle \mathcal{S}_{k+1},\mathcal{S}_k \rangle$ in order to take into account reinforcement drift in the direction of $S_k$ when its norm is sufficiently large. We propose to introduce the geodesic bisectors $\beta_{i,j}^k$ between each pair of agents $\{i,j\}$ and the half angle $\alpha_{i,j}^k$ of the geodesic circle arc between them. Many interesting properties and formula can be derived, unfortunately we are faced with summations of quantities that are difficult to bound.

The second approach consists in analysing the evolution of the maximal empty angle $\gamma_{\mathrm{max}}^k$, which is the angle of the longest circle arc between two consecutive opinions on the circle. Note that this arc is not necessarily geodesic. Interestingly, there is equivalence between $\gamma_{\mathrm{max}}^k\ge \frac{\pi}{2}$ and all unit vector opinions are within a half-disk. Thus we could study $\gamma_{\mathrm{max}}^k$ as a random walk on $[0,2\pi)$ starting in $[0,\pi)$ and look for the first time it passes the $\pi$ threshold. Ideally, $\gamma_{\mathrm{max}}^k$ would be a sub-martingale which would then give us almost sure convergence and convergence time bounds. We can show that while the opinions are not yet contained within a half-disk, $\gamma_{\mathrm{max}}^k$ is biased to increase, in particular that $ \mathbb{P}(\gamma_{\mathrm{max}}^{k+1} < \gamma_{\mathrm{max}}^k \mid \gamma_{\mathrm{max}}^k) \le \frac{1}{2}\left(1-\frac{1}{N}\right)$. Unfortunately, simply having that the probability of decrease is upper bounded by a value strictly smaller than $\frac{1}{2}$ is not enough, we need to study with more detail the probability distribution of $\gamma_{\mathrm{max}}^k$ for its expectation. The proof provides a reasonable approach to bound the entire distribution of the decrease event of the maximal empty angle. However, analysing the increase is significantly harder and remains an open challenge. 

The third approach consists in designing and analysing Markov chains. We studied a Markov chain with $n+1=\floor{\frac{N}{2}}+1$ states, which is an extension of the naive proof in \cref{th:circle to half disk finite expec time optimised}.  It is essentially a doubly chained graph with probability $c$ of increase and $1-c$ of decrease and the last state is absorbing. On average, reaching the absorbing state takes longer than to reach a half-disk configuration. In \cref{th:circle to half disk finite expec time optimised}, we analyse $n$ successive increases. In reality, we tolerate some decreases in the process. Explicit calculation of the expected time to reach the absorbing state is possible by inverting an almost tridiagonal Toeplitz matrix using the Sherman-Morrison formula and the well-known invert of a tridiagonal Toeplitz matrix \cite{dafonseca2001}. As $c = \frac{4}{27 N^2 (N-1)^2}<\frac{1}{2}$, the expected time is approximately $\left(\frac{1-c}{c}\right)^n$. This yields a bound similar to the one given in \cref{th:circle to half disk finite expec time optimised}. 
The problem is that $c$ was derived using a pessimistic worst case geometry per state. In practice, closer to half-disk configurations, thus with higher state number, the geometry is biased far away from the worst case scenario giving on average significantly higher state increase probabilities. We believe that it should be possible to find an alternative simple Markov chain with higher probabilities for getting to the absorbing state that provides a reasonable upper bound
.

\section{Numerical results}

While the theory provides a guarantee of finite expected time convergence in all previous cases, it also provides explicit bounds, which we can compare to empirical results in numerical simulations.

\subsection{One dimensional case} The chosen domain is the unit interval $I = [0,1]$. The initial opinions in $X_0$ follow an iid uniform distribution in $I$. We tested the grid of configurations defined by the number of agents $N\in\{5,10,100,250,500,750,1000\}$ and convergence threshold $\varepsilon\in\{0.0001,0.0005,0.001,0.005,0.01,0.05,0.1\}$. For each configuration, $n_{\textit{trials}} = 1000$ independent trials were performed. Each trial was stopped when $\mathcal{L}^k\le2\varepsilon^2 < N\varepsilon^2$, which guarantees to have reached $\varepsilon$ convergence. We denote $\hat{T}_\varepsilon$ the natural estimator of $\mathbb{E}(T_\varepsilon)$ by simply taking its empirical average.

A summary of the empirical dependency of the average convergence time on the convergence threshold $\varepsilon$ is done in \cref{fig:1D cv eps}, where we plot $\hat{T}_\varepsilon$ against $\varepsilon$ and against $-\ln\varepsilon$. We find that $\hat{T}_\varepsilon$ has a minus logarithmic dependency on $\varepsilon$ as expected from \cref{th:1D finite expected time expected L0}. Furthermore, the slopes of the curves and their respective bounds from \cref{th:1D finite expected time expected L0} in \cref{fig:1D mean cv minus log eps with bounds full} seem to be approximately the same for high $N$, suggesting that in fact the convergence time is not only upper bounded but also lower bounded by a similar term with approximately the same dominant coefficient: $T_\varepsilon \approx c_N\frac{3}{2} N\ln\left(\frac{N}{\varepsilon^2}\right) + O(1)$ where $O(1)$ represents a function bounded with respect to $\varepsilon$ (but not with respect to $N$), and $c_N\in[0,1]$ a constant depending on $N$ such that $c_N\xrightarrow[N\to\infty]{} 1$ and $c_N\approx 1$ when $N\ge 100$? A further analysis on the dependency on $N$ of $\hat{T}_{\varepsilon}$ gives that empirically $\hat{T}_{\varepsilon} \approx -3 c_N  N\ln\varepsilon + 0.89 N\ln{N}  -2.3 N + 5.8$. This is done in the supplementary material \cref{sm: 1D emp reg dep on N} as we are primarily interested in the $\varepsilon$ dependency in this paper.

\begin{figure}[tbhp]
    \centering
    \subfloat[]{\label{fig:1D mean cv eps with bounds full}\includegraphics[width=0.49\textwidth]{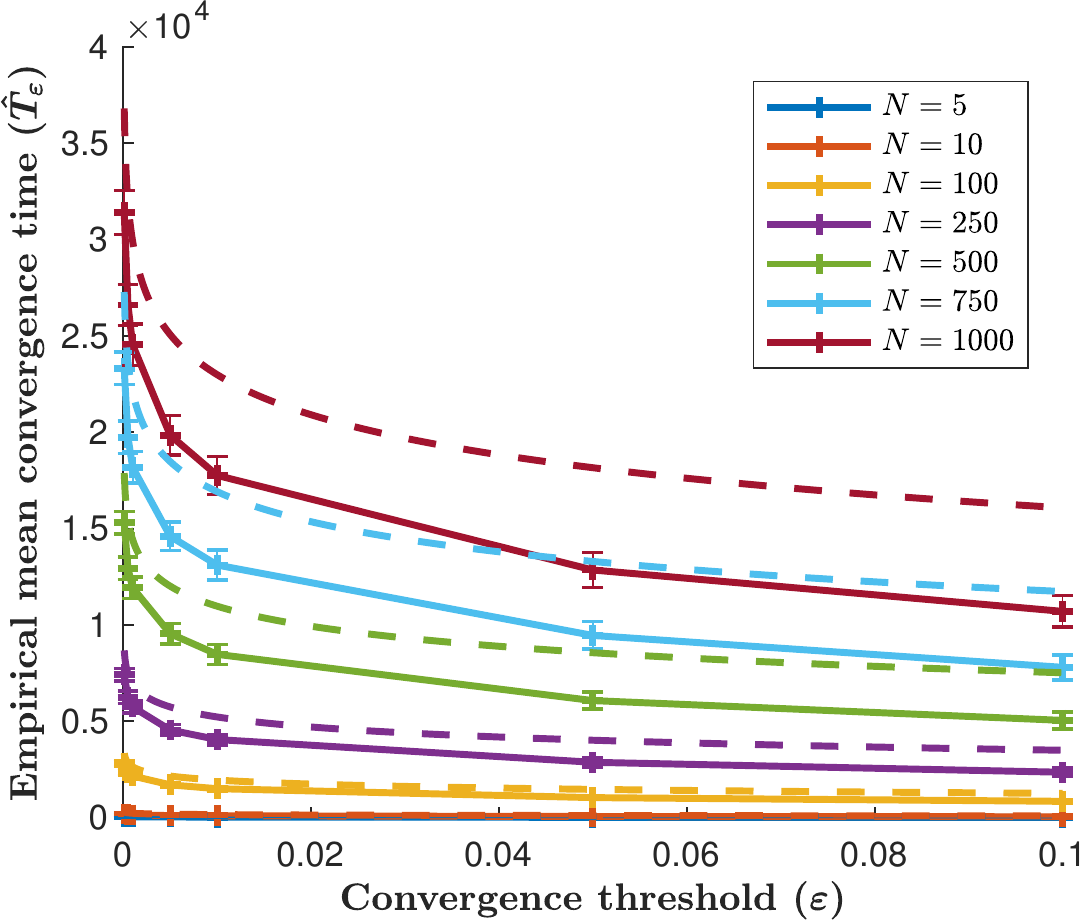}}
    \subfloat[]{\label{fig:1D mean cv minus log eps with bounds full}\includegraphics[width=0.49\textwidth]{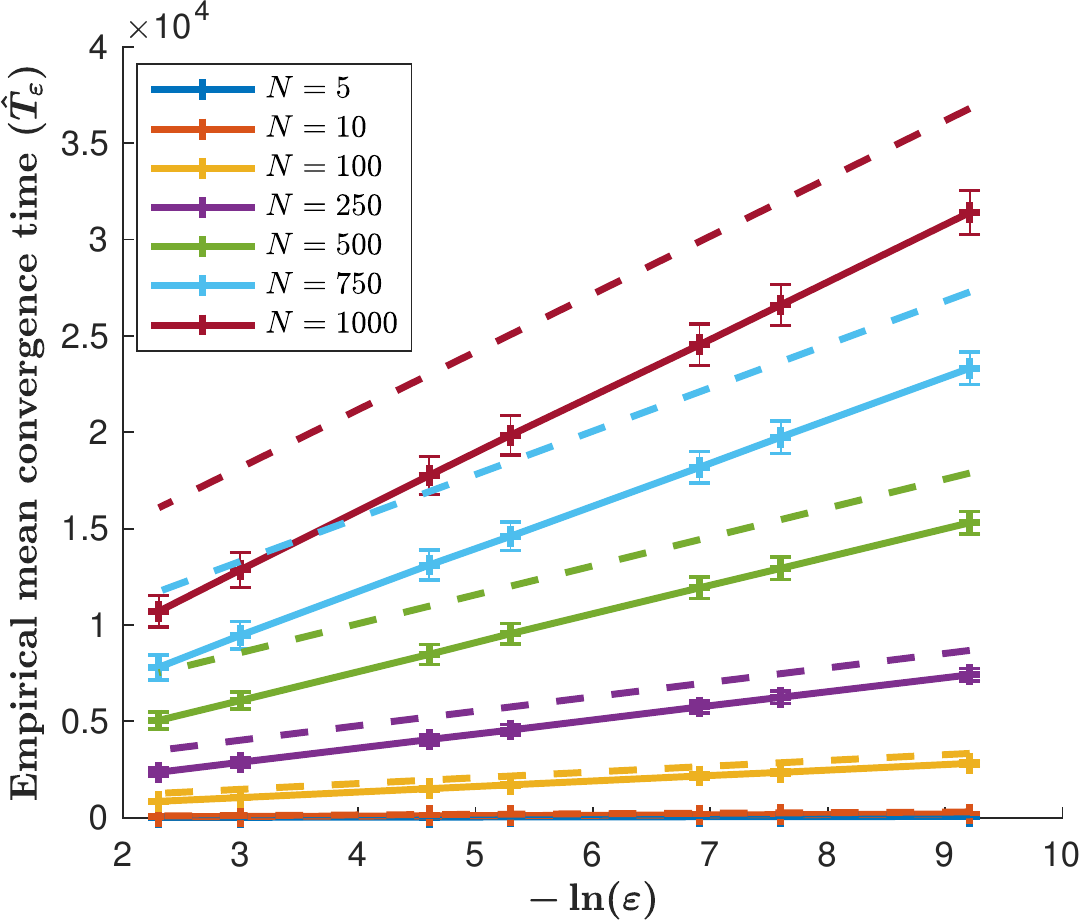}}\\ 
    \caption{One dimensional evolution: dependency of the empirical mean convergence time on the convergence threshold $\varepsilon$. Left: $\varepsilon$ abscissa. Right: $-\ln \varepsilon$ abscissa. The plain curves correspond to the empirical results whereas the dashed ones correspond to the theoretical bounds. We superimpose on the empirical curves the classic unbiased estimator of the standard deviation of each data point.}
    \label{fig:1D cv eps}
\end{figure}

\subsection{Unconstrained \texorpdfstring{\boldmath$D$}{D}-dimensional case} The chosen domain is the unit cube $C = [0,1]^D$, where the dimension $D$ ranges in $\{2,3,4\}$. The initial opinions in $X_0$ follow an iid uniform distribution in $C$. We tested the grid of configurations defined by the number of agents $N\in\{5,10,50,100,250\}$ and convergence threshold $\varepsilon\in\{0.0005,0.001,0.005,0.01,0.05,0.1\}$. For each configuration, $n_{\textit{trials}} = 1000$ independent trials were performed. Each trial was stopped when $\mathcal{L}_d^k\le2\varepsilon^2 < N\varepsilon^2$ for all dimensions $d$. 
A similar estimator was used to the 1-Dimensional case.

A summary of the empirical dependency of the average convergence time on the convergence threshold $\varepsilon$ is done in \cref{fig:ND cv eps}. Once again, we find that $\hat{T}_\varepsilon$ has a minus logarithmic dependency on $\varepsilon$, and the slopes approximately correspond to those derived in the upper bound. On the other hand, for the tested values of $D$, the displacement between the true convergence time and the bounds seem to be the same. Furthermore, we see a slight increase in convergence time with respect to $D$. However it would require extensive trials with high $D$ to be able to claim that the dependency is indeed logarithmic, which would be computationally too expensive for our purposes. These three observations lead us to generalise naturally the conjecture made in the one dimensional case: $T_\varepsilon \approx c_N\frac{3}{2}N\ln\left(\frac{DN}{\varepsilon^2}\right) + O(1)$ where $O(1)$ represents a function bounded with respect to $\varepsilon$ (and perhaps also with respect to $D$ but not with respect to $N$)? See the supplementary material \cref{sm: ND emp reg dep on N} for an analysis of the dependency in $N$ and in particular for the confirmation of the presence of $c_N$.

\begin{figure}[tbhp]
    \centering
    \subfloat{\label{fig:2D mean cv eps}\includegraphics[width=0.33\textwidth]{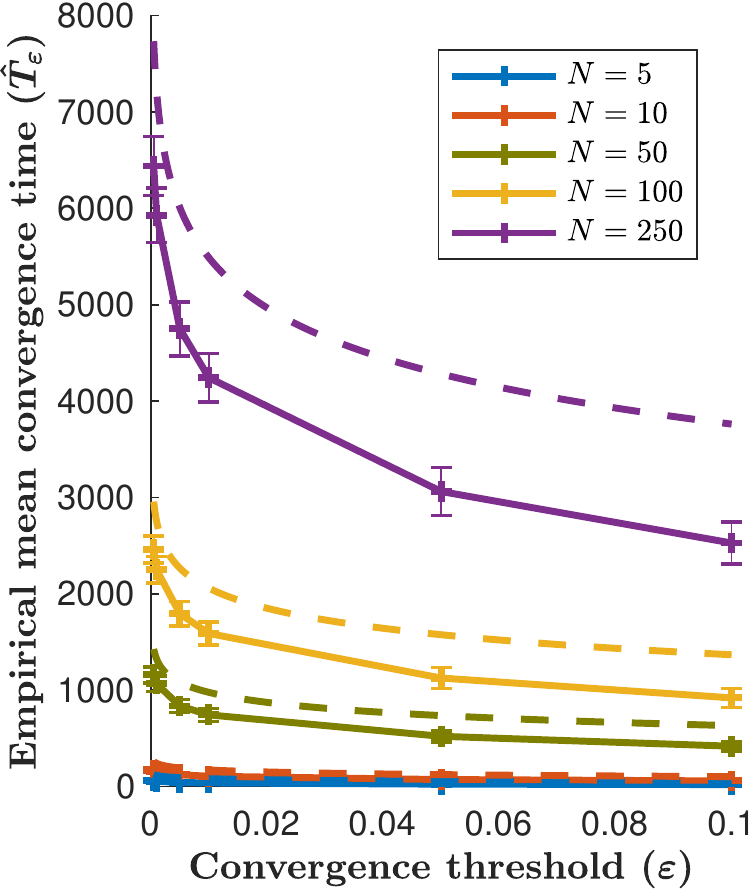}}
    \subfloat{\label{fig:3D mean cv eps}\includegraphics[width=0.33\textwidth]{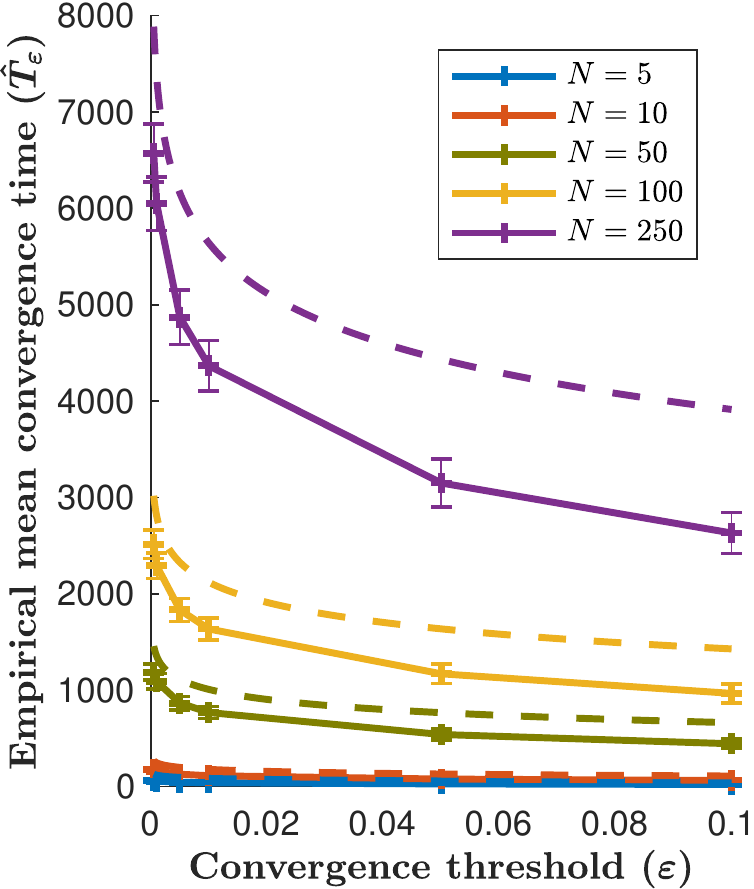}}
    \subfloat{\label{fig:4D mean cv eps}\includegraphics[width=0.33\textwidth]{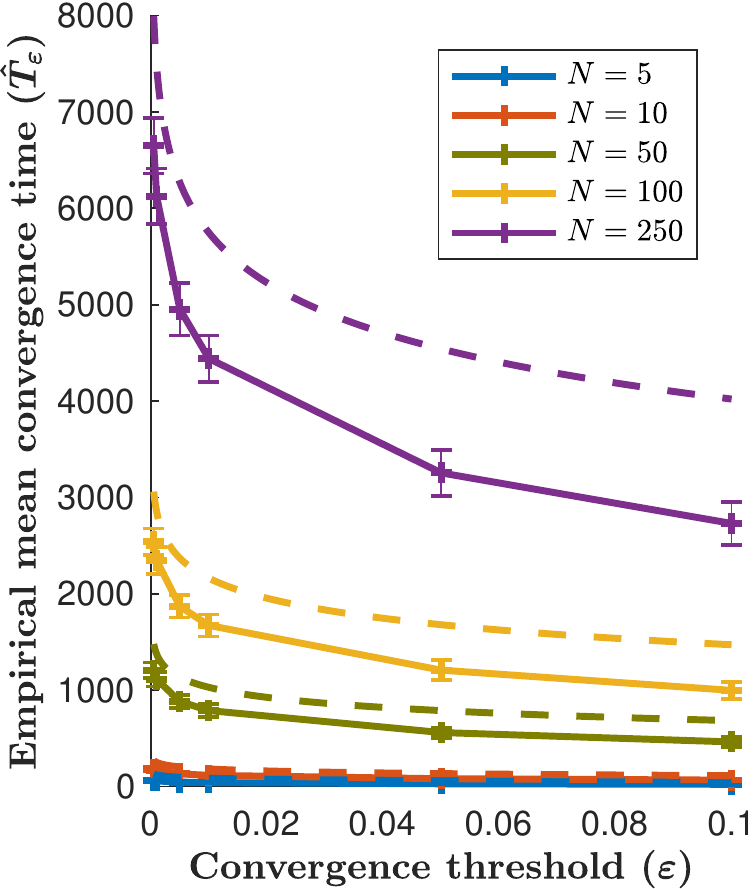}}\\
    \subfloat{\label{fig:2D mean cv eps log}\includegraphics[width=0.33\textwidth]{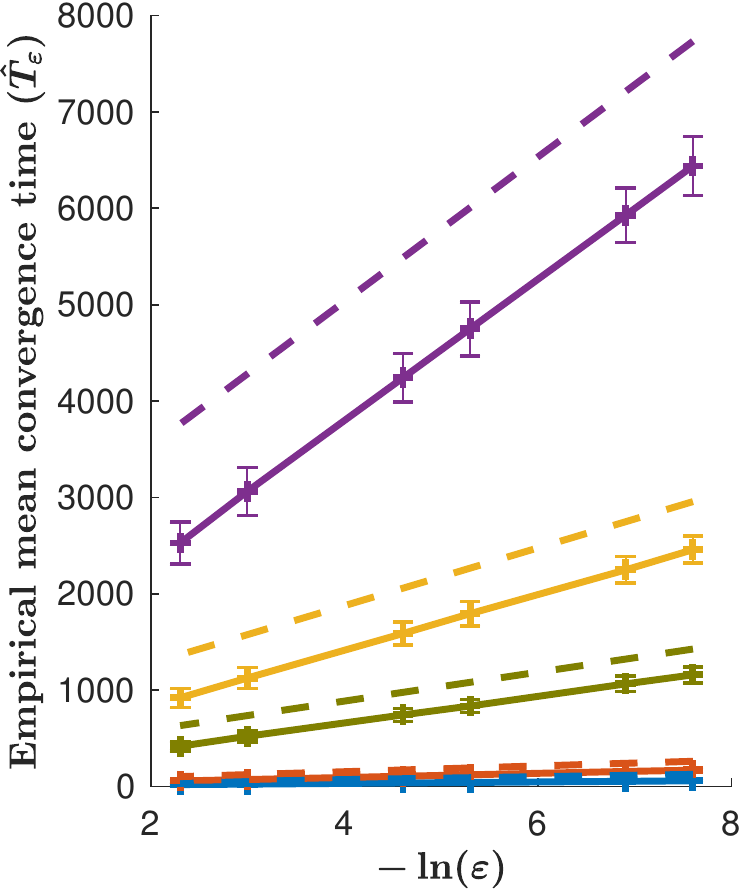}}
    \subfloat{\label{fig:3D mean cv eps log}\includegraphics[width=0.33\textwidth]{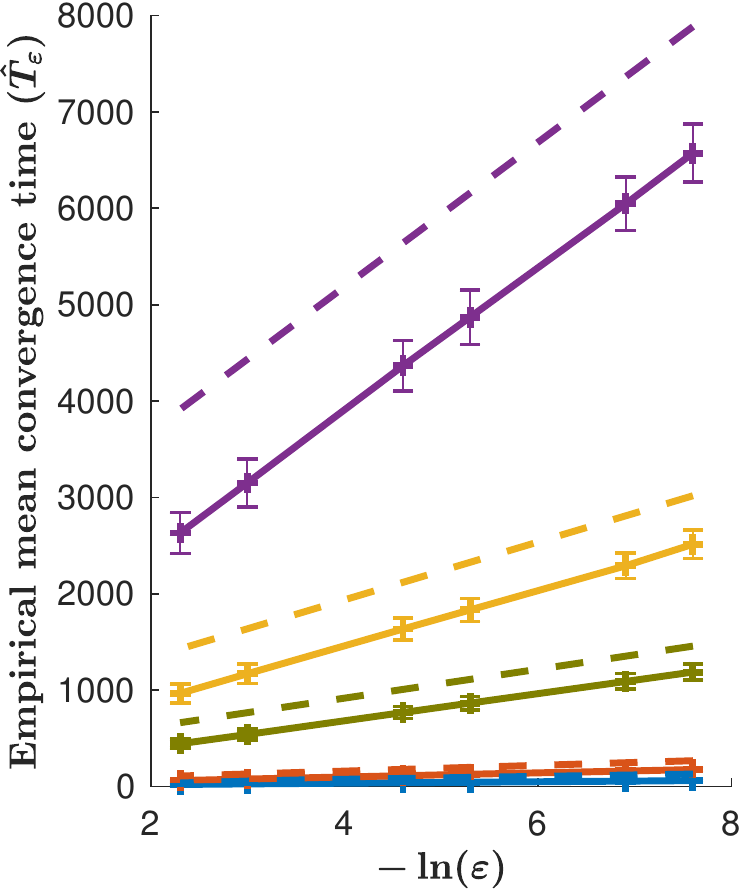}}
    \subfloat{\label{fig:4D mean cv eps log}\includegraphics[width=0.33\textwidth]{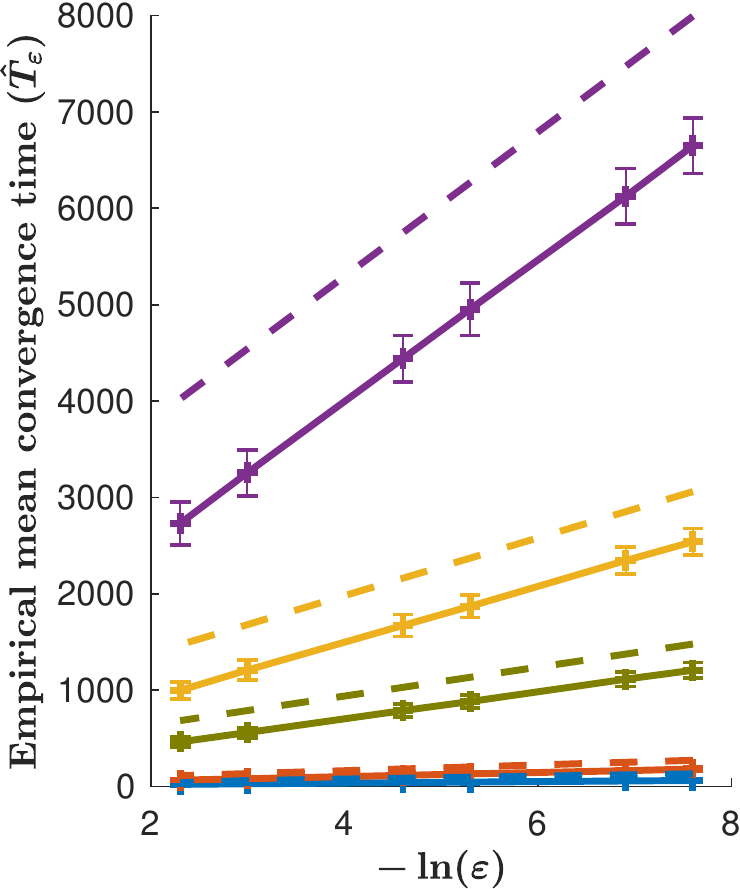}} \\ 
    \caption{$D$-dimensional evolution: dependency of the empirical mean convergence time on the convergence threshold $\varepsilon$ in the $2$, $3$, and $4$ dimensional cases from left to right. Top: $\varepsilon$ abscissa. Bottom: $-\ln \varepsilon$ abscissa. The plain curves correspond to the empirical results whereas the dashed ones correspond to the theoretical bounds. We superimpose on the empirical curves the traditional unbiased estimator of the standard deviation of each data point.}
    \label{fig:ND cv eps}
\end{figure}

\subsection{Constrained 2-dimensional case} The initial opinions in $\Theta_0$ follow an iid uniform distribution in $\left[0,2\pi\right)$. We tested the grid of configurations defined by the number of agents $N\in\{5,10,100,250,500,750,1000\}$ and convergence threshold $\varepsilon\in\{0.0001,0.0005,0.001,0.005,0.01,0.05,0.1\}$. For each configuration, $n_{\textit{trials}} = 1000$ independent trials were performed.  Each trial was stopped as soon as $\varepsilon$ convergence was reached, i.e. $\gamma_{\mathrm{max}}^k\ge2\pi-\varepsilon$.  A similar estimator was used to the convex case. 

A summary of the empirical dependency of the average convergence time on the convergence threshold $\varepsilon$ is done in \cref{fig:circle mean cv eps,fig:circle mean cv minus log eps}. As in the convex case, we find that $\hat{T}_\varepsilon$ has a minus logarithmic dependency on $\varepsilon$ as predicted in \cref{th: circle finite expec cv} with similar slope. However, the bound is many orders of magnitude larger than our estimator even for large $\varepsilon$. This is due to our poor bound $B_N^{HD}$ deriving from a Borel-Cantelli like idea when studying $T_{HD}$. Since $T_{HD}$ is independent of $\varepsilon$ as soon as $\varepsilon\le \pi$, the dependency on $\varepsilon$ is naturally inherited from the one dimensional case, as the angle of the opinions follow the $1D$ case evolution when all unit vector opinions are within a half-disk. We can therefore extend the conjecture to the circle case: $T_\varepsilon \approx c_N\frac{3}{2}N\ln\left(\frac{N}{\varepsilon^2}\right) + O(1)$ where $O(1)$ represents a function bounded with respect to $\varepsilon$ (but not with respect to $N$)?

The key part in the circle evolution, and the hardest one to analyse, is the transitory regime when not all agents are within a half-disk, i.e. $k<T_{HD}$. To better understand the behaviour of the systems in this regime, a summary of the empirical dependency of the average stopping time to a half-disk configuration on the number of agents $N$ is done in \cref{fig:circle mean hd N,fig:circle mean hd N log N}. We find that $\hat{T}_{HD}$ depends quasi-linearly on $N$, in fact a linear regression gives that $\hat{T}_{HD}\approx 0.92 N\ln{N} + 100$, to be compared with the $O((\frac{3\sqrt{3}}{2})^N N^{2N})$ bound from \cref{th:circle to half disk finite expec time optimised}, which is many orders of magnitude larger than our estimator even for the smallest number of agents. Further work is needed to find a better theoretical $B_N^{HD}\in \mathcal{B}_N^{HD}$ that should be a $O(N\log N)$.

\begin{figure}[tbhp]
    \centering
    \subfloat[]{\label{fig:circle mean cv eps}\includegraphics[width=0.25\textwidth]{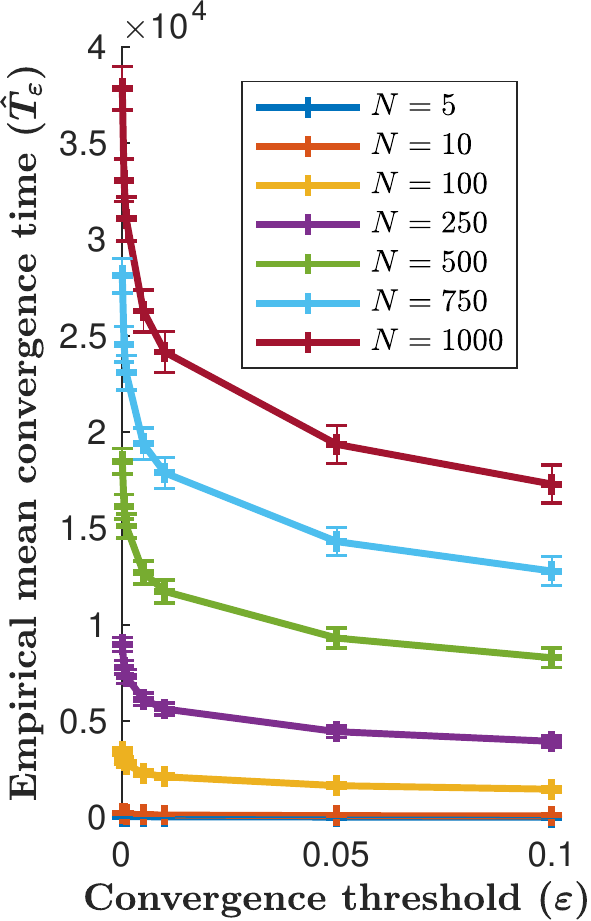}}
    \subfloat[]{\label{fig:circle mean cv minus log eps}\includegraphics[width=0.25\textwidth]{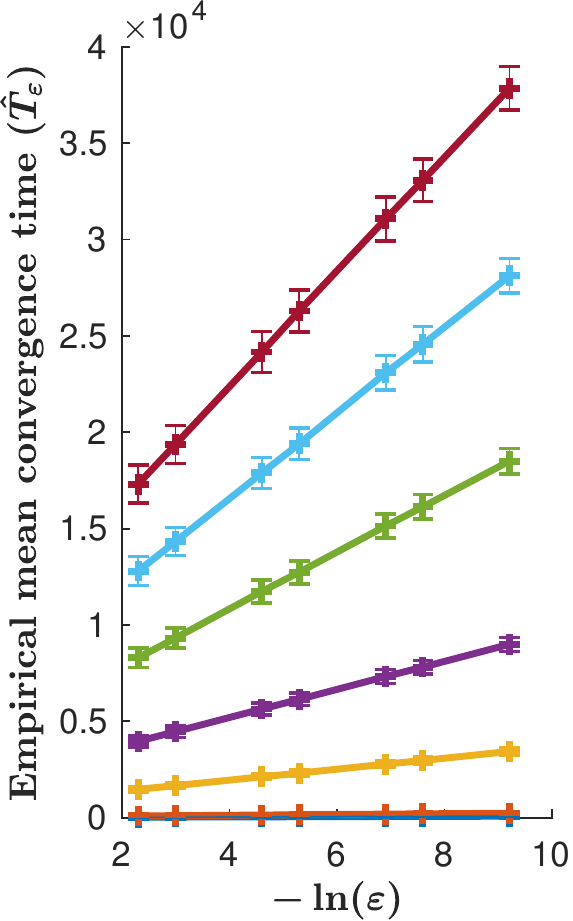}}
    \subfloat[]{\label{fig:circle mean hd N}\includegraphics[width=0.25\textwidth]{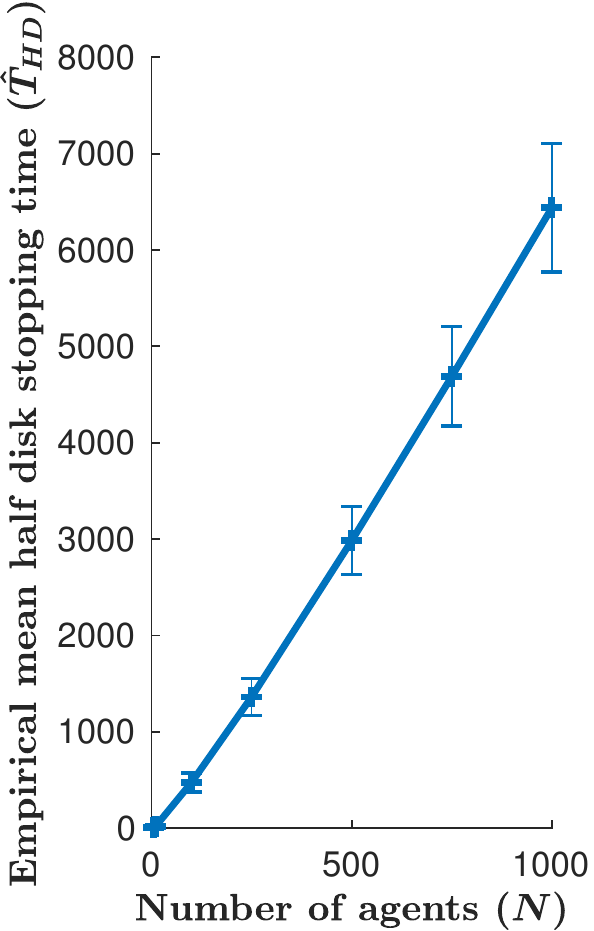}}
    \subfloat[]{\label{fig:circle mean hd N log N}\includegraphics[width=0.25\textwidth]{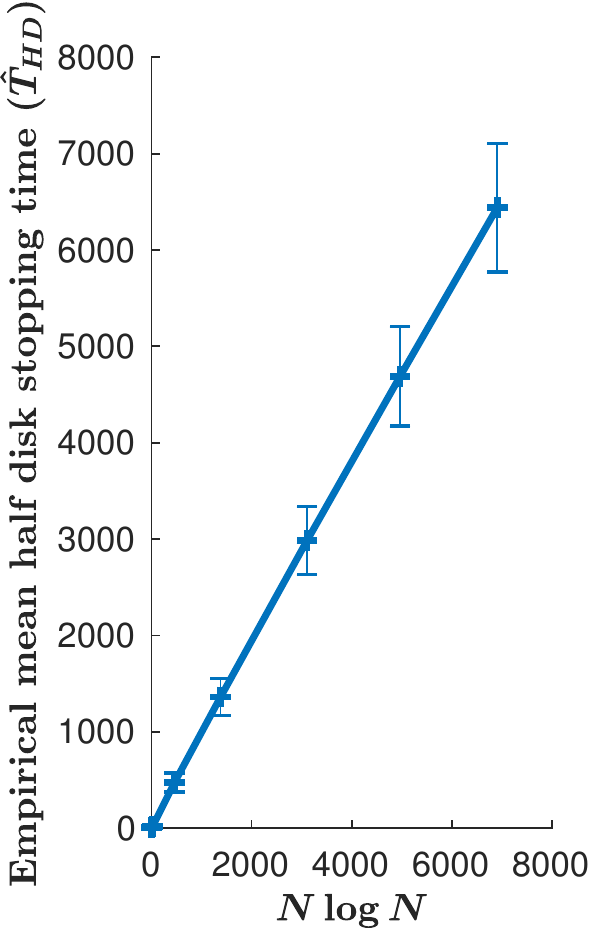}}\\
    \caption{Circle evolution: the two plots on the left represent the dependency of the empirical mean convergence time on the convergence threshold $\varepsilon$ while the two on the right display the dependency of the empirical mean half-disk stopping time $\hat{T}_{HD}$ on the number of agents $N$. Left: Empirical results with abscissa $\varepsilon$. Middle left: Empirical results with abscissa $-\ln\varepsilon$. Middle right: Empirical results with abscissa $N$. Right: Empirical results with abscissa $N\ln N$. We superimpose the traditional unbiased estimator of the standard deviation of each data point.}
    \label{fig:circle cv eps and circle hd N}
\end{figure}

A further analysis on the dependency on $N$ of $\hat{T}_{\varepsilon}$ gives that empirically we have $\hat{T}_{\varepsilon} \approx -3 c_N  N\ln\varepsilon + 0.93N\ln{N}+4.1N-18$. This is done in the supplementary material \cref{sm: circle emp reg dep on N} as we are primarily interested in the $\varepsilon$ dependency in this paper.

We also plot examples of evolutions of $\mathcal{S}^k$, the vector sum of all unit vectors, in single trials for various number of agents in \cref{fig:circle S evolution}. It seems that $\mathcal{S}^k$ is initially random around 0 and then after a small threshold distance drifts in its current direction, suggesting that  $\left\lVert\mathcal{S}^k\right\rVert_2^2$ or $\left\langle\mathcal{S}^{k+1},\mathcal{S}^k\right\rangle$ would be interesting quantities to analyse.

\begin{figure}[tbhp]
    \centering
        \subfloat[$N = 5$]{\label{fig:circle S evolution N 5}\includegraphics[width=0.25\textwidth]{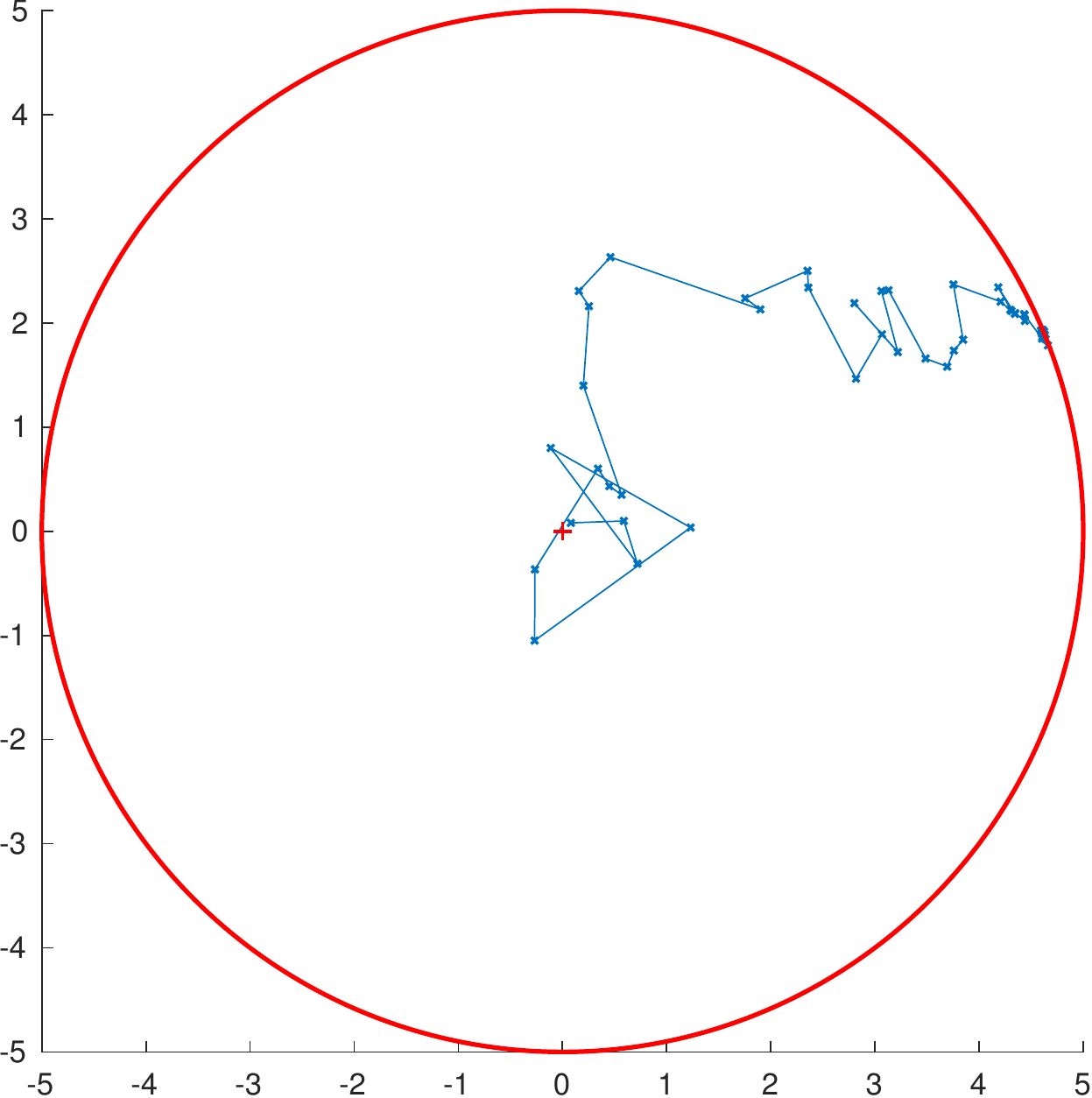}}
    \subfloat[$N = 50$]{\label{fig:circle S evolution N 50}\includegraphics[width=0.25\textwidth]{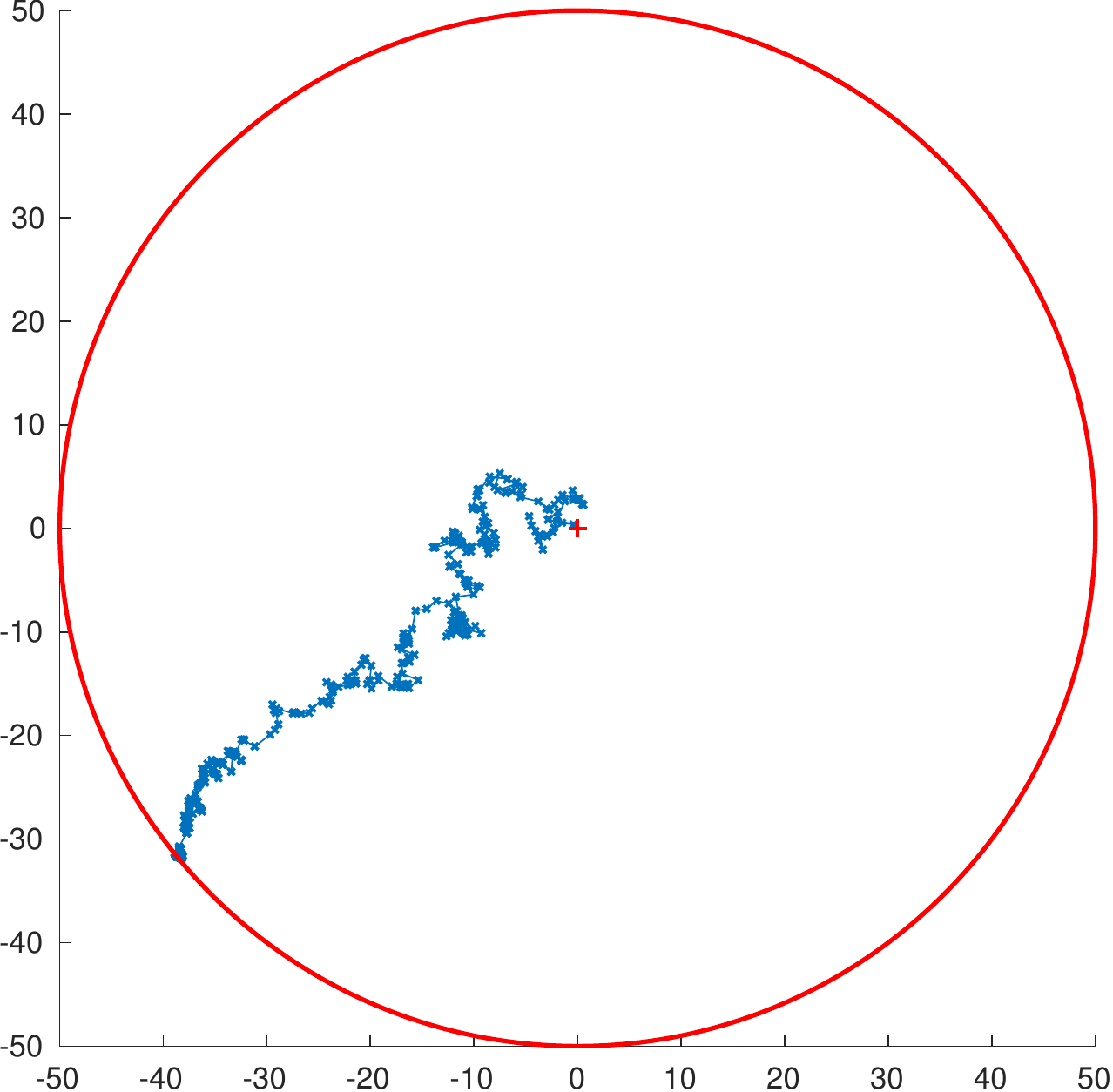}}
    \subfloat[$N = 100$]{\label{fig:circle S evolution N 100}\includegraphics[width=0.25\textwidth]{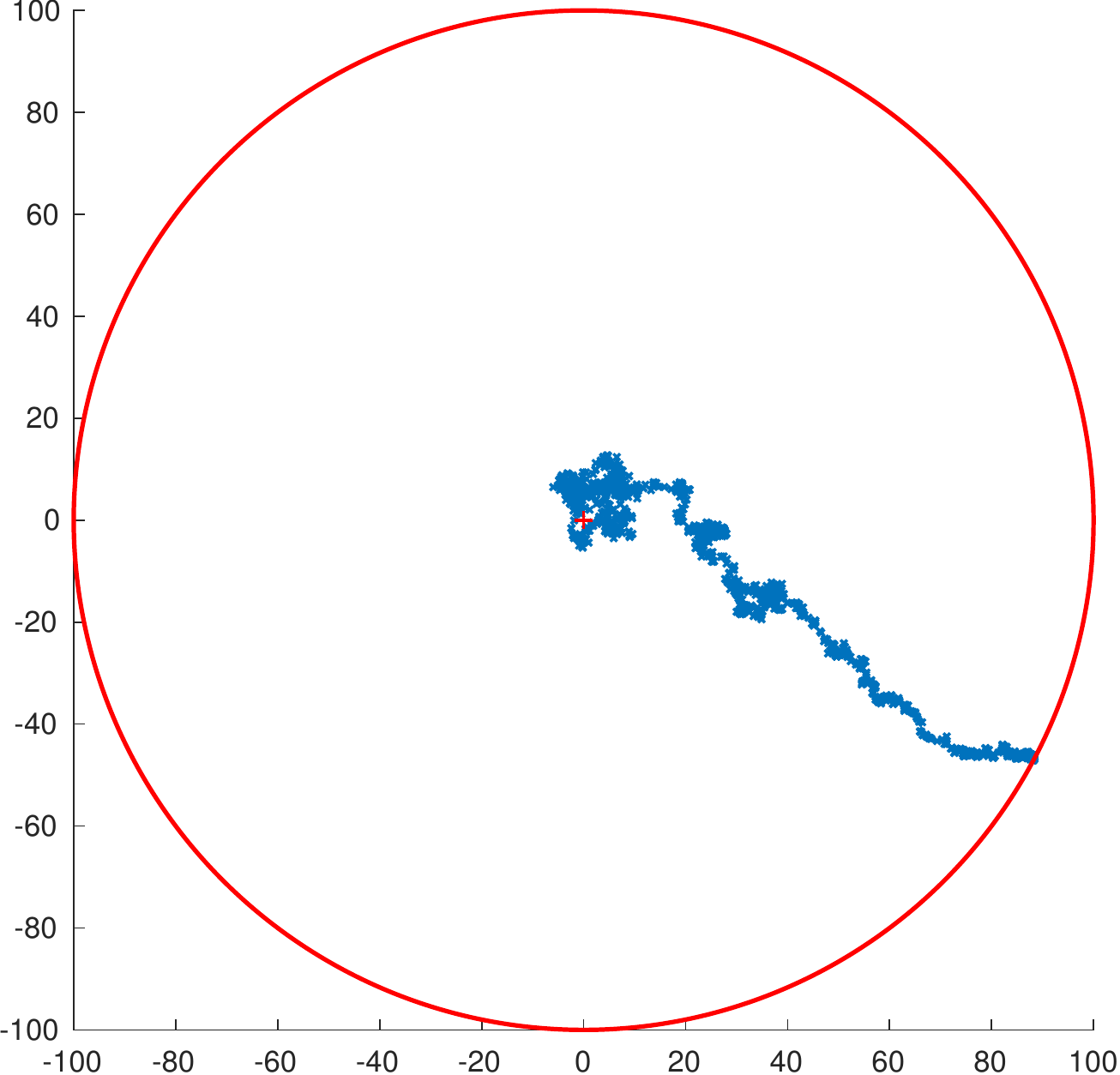}}
    \subfloat[$N = 1000$]{\label{fig:circle S evolution N 1000}\includegraphics[width=0.25\textwidth]{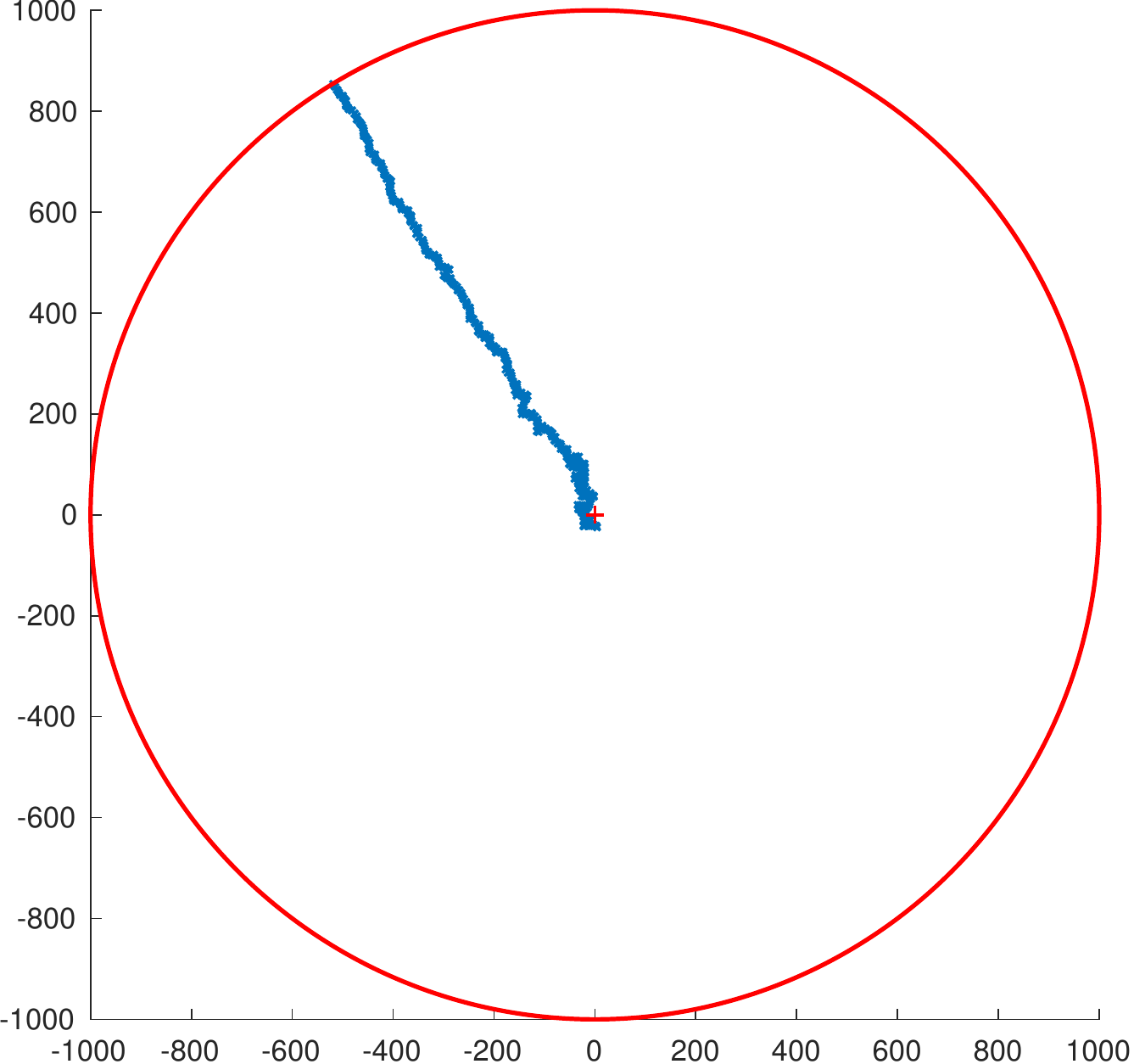}}\\
    \caption{Circle evolution: evolution of the average of unit vector opinions $\mathcal{S}^k$ in a single trial for various number of agents with random uniform initial distribution on the circle. The red circle corresponds to $\left\lVert \mathcal{S}^k\right\rVert_2 = N$. We stopped the evolution after $1000$ steps for $N \in \left\{5,50\right\}$, $2000$ steps for $N = 100$, and $30000$ steps for $N = 1000$.}
    \label{fig:circle S evolution}
\end{figure}

\section{Conclusion}

We analysed in detail models of doubly stochastic pairwise interactions for $N$ agents with states described by a single real value, by a $D$-dimensional real vector, or by a constrained unit vector on the circle. The evolution in time of the states of the $N$-agent system was found to exhibit convergence to $\varepsilon$-agreement in finite expected time, and we provide upper bounds on the expected time that are tight in the case of unconstrained states. However, for unit vector states the dependence on $N$ in the upper bound is quite far from the empirical results. This is due to the difficulty in proving a fast gathering of unit vectors into a half-plane as a result of the assumed doubly stochastic pariwise interactions. This problem is challenging but we hope to address it in the near future, along the lines outlined in \cref{sec: open problems for the circle}.



\bibliographystyle{siamplain}
\bibliography{references}

\clearpage

\begin{center}
\textbf{\large SUPPLEMENTARY MATERIALS: 
DOUBLY STOCHASTIC PAIRWISE INTERACTIONS FOR AGREEMENT AND ALIGNMENT }
\end{center}
\setcounter{equation}{0}
\setcounter{figure}{0}
\setcounter{table}{0}
\setcounter{page}{1}
\setcounter{section}{0}
\makeatletter
\renewcommand{\theequation}{S\arabic{equation}}
\renewcommand{\thefigure}{S\arabic{figure}}

\section{1-Dimensional case}

\subsection{Preliminary definitions and properties in the 1-Dimensional case}

\begin{definition}
    \label{def: E_ij conditional choosing}
    By abuse of notation, for a random variable $Y$, at step $k+1$, we will write the conditioning with respect to choosing the agents $(i,j)$ for evolution as 
    \begin{equation}
        \mathbb{E}_{i,j}(Y) = \mathbb{E}\big(Y\mid (A_{k+1},B_{k+1}) = (i,j)\big).
    \end{equation}
\end{definition}

\begin{proposition}
    \label{prop: E(x_i) and E(x_i^2)}
    Conditionally to $X_k$ and choosing $(i,j)$, we have:
    \begin{align}
        \label{eq: 1D expec cond i-j xi next}
        \mathbb{E}_{i,j}(x_i^{k+1}\mid X_k) &= \mathbb{E}_{i,j}(x_j^{k+1}\mid X_k) = \frac{x_i+x_j}{2} \\
        \mathbb{E}_{i,j}((x_i^{k+1})^2 \mid X_k) &= \frac{1}{3}\big((x_i^k)^2 +x_i^k x_j^k + (x_j^k)^2\big).
    \end{align}
\end{proposition}

\begin{proof}
    The results are straightforward due to the fact that conditionally to $X_k$ and choosing the pair $(i,j)$, the updated random variables $U_1^{k+1}$ and $U_2^{k+1}$ are uniform random variables.
\end{proof}

\subsection{Proof of \cref{prop: 1D expec lyapunov n+1}}
\label{sm: proof 1D expec lyapunov n+1}
    For conciseness, denote $\mathcal{L}_{i,j}^k$ and $\mathcal{L}^k$ the studied quantities:
    \begin{align}
        \mathcal{L}_{i,j}^k &= (x_i^k-x_j^k)^2 \\
        \mathcal{L}^k &= \sum\limits_{i\neq j}(x_i^k-x_j^k)^2 = \sum\limits_{i\neq j}\mathcal{L}_{i,j}^k.
    \end{align}

    We then have, by conditioning on the chosen pair $(m,l)$ for update:
    \begin{align}
        \mathbb{E}(\mathcal{L}^{k+1}\mid X_k) &= \sum\limits_{m\neq l}\mathbb{E}_{m,l}(\mathcal{L}^{k+1}\mid X_k)\mathbb{P}\big((A_{k+1},B_{k+1}) = (m,l)\big) \nonumber\\
        &= \frac{1}{N(N-1)}\sum\limits_{m\neq l}\mathbb{E}_{m,l}(\mathcal{L}^{k+1}\mid X_k) \nonumber\\
        &= \frac{1}{N(N-1)}\sum\limits_{m\neq l}\Bigg[\sum\limits_{\substack{i=m\\j\notin \{m,l\}}}\mathbb{E}_{m,l}(\mathcal{L}_{i,j}^{k+1}\mid X_k) + \sum\limits_{\substack{i=l\\j\notin \{m,l\}}}\mathbb{E}_{m,l}(\mathcal{L}_{i,j}^{k+1}\mid X_k) \nonumber\\
        &\hspace{2em}+ \sum\limits_{\substack{i=m\\j=l}}\mathbb{E}_{m,l}(\mathcal{L}_{i,j}^{k+1}\mid X_k) + \sum\limits_{\substack{i=l\\j=m}}\mathbb{E}_{m,l}(\mathcal{L}_{i,j}^{k+1}\mid X_k) \nonumber\\
        &\hspace{2em}+ \sum\limits_{\substack{j=m\\i\notin \{m,l\}}}\mathbb{E}_{m,l}(\mathcal{L}_{i,j}^{k+1}\mid X_k) + \sum\limits_{\substack{j=l\\i\notin \{m,l\}}}\mathbb{E}_{m,l}(\mathcal{L}_{i,j}^{k+1}\mid X_k) \nonumber\\
        &\hspace{2em}+
        \sum\limits_{\{i,j\}\cap \{m,l\} = \emptyset}\mathbb{E}_{m,l}(\mathcal{L}_{i,j}^{k+1}\mid X_k)
        \Bigg].
    \end{align}
    
    By expanding the squares and using \cref{prop: E(x_i) and E(x_i^2)}, we get:
    \begin{align}
        \label{eq: 1D expec not yet Ln explicitly appear}
        \mathbb{E}(\mathcal{L}^{k+1}\mid X_k) &= \frac{1}{N(N-1)}\sum\limits_{m\neq l}\Bigg[ \sum\limits_{\{i,j\}\cap \{m,l\} = \emptyset}(x_i^k-x_j^k)^2 \nonumber\\
        &\hspace{2em}+ 2\Big(\sum\limits_{\substack{i=m \\ j\notin \{m,l\}}}\big(\frac{1}{3}((x_m^k)^2+x_m^k x_l^k +(x_l^k)^2)-2\frac{x_m^k+x_l^k}{2}x_j^k + (x_j^k)^2\big) \Big) \nonumber\\
        &\hspace{2em}+ 2\Big(\sum\limits_{\substack{i=l \\ j\notin \{m,l\}}}\big(\frac{1}{3}((x_m^k)^2+x_m^k x_l^k +(x_l^k)^2)-2\frac{x_m^k+x_l^k}{2}x_j^k + (x_j^k)^2\big) \Big) \nonumber\\
        &\hspace{2em}+ 2\Big(\frac{2}{3}( (x_m^k)^2 +x_m^k x_l^k + (x_l^k)^2) - 2\big(\frac{x_m^k+x_l^k}{2}\big)^2\Big) \Bigg].
    \end{align}
    
    By making $\mathcal{L}^k$ explicitly appear, we have:
    \begin{align}
        \label{eq: 1D expec Ln explicitly appear}
        \sum\limits_{\{i,j\}\cap \{m,l\} = \emptyset}\mathcal{L}_{i,j}^k &= \sum\limits_{i\neq j}\mathcal{L}_{i,j}^k - \Big(\sum\limits_{\substack{i=m\\j\notin \{m,l\}}}\mathcal{L}_{i,j}^k \Big) - \mathcal{L}_{m,l}^k - \Big(\sum\limits_{\substack{i=l\\j\notin \{m,l\}}}\mathcal{L}_{i,j}^k \Big) \nonumber \\
        &\hspace{2em}- \mathcal{L}_{l,m}^k - \Big(\sum\limits_{\substack{j=m\\i\notin\{m,l\}}}\mathcal{L}_{i,j}^k\Big)-\Big(\sum\limits_{\substack{j=l\\i\notin\{m,l\}}}\mathcal{L}_{i,j}^k\Big) \nonumber\\
        &= \mathcal{L}^k -2(x_m^k - x_l^k)^2 - 2\sum\limits_{i\notin \{m,l\}}(x_m^k-x_i^k)^2-2\sum\limits_{i\notin \{m,l\}}(x_l^k-x_i^k)^2.
    \end{align}
    
    By plugging \cref{eq: 1D expec Ln explicitly appear} into \cref{eq: 1D expec not yet Ln explicitly appear} and expanding the squares, we get:
    \begin{align}
        \mathbb{E}(\mathcal{L}^{k+1}\mid X_k) &= \mathcal{L}^k + \frac{1}{N(N-1)}\sum\limits_{m\neq l}\Bigg[ -2(x_m^k)^2-2(x_l^k)^2+4x_m^k x_l^k -2(N-2)(x_m^k)^2 \nonumber\\
        &\hspace{2em}+4\sum\limits_{i\notin\{m,l\}}x_i^k x_m^k-2\sum\limits_{i\notin\{m,l\}}(x_i^k)^2 -2(N-2)(x_l^k)^2\nonumber 
        +4\sum\limits_{i\notin\{m,l\}}x_i^k x_l^k \\
        &\hspace{2em}-2\sum\limits_{i\notin\{m,l\}}(x_i^k)^2 +\frac{4}{3}(N-2)((x_m^k)^2+x_m^k x_l^k +(x_l^k)^2)\nonumber \\
        &\hspace{2em} - 4\sum\limits_{i\notin\{m,l\}}x_i^k x_m^k - 4\sum\limits_{i\notin\{m,l\}}x_i^k x_l^k + 4 \sum\limits_{i\notin\{m,l\}}(x_i^k)^2 +\frac{1}{3}(x_m^k)^2\nonumber \\
        &\hspace{2em} + \frac{1}{3}(x_l^k)^2 -\frac{2}{3}x_m^k x_l^k \Bigg].
    \end{align}
    
    The terms in this expression depending on the summation over $i\notin\{m,l\}$ cancel out one another. Then by regrouping terms:
    \begin{align}
        \mathbb{E}(\mathcal{L}^{k+1}\mid X_k) &= \mathcal{L}^k + \frac{1}{N(N-1)}\sum\limits_{m\neq l}\Bigg[ \big(-2-2(N-2)+\frac{4}{3}(N-2)+\frac{1}{3}\big)(x_k^k)^2 \nonumber\\
        &\hspace{2em}+ \big(-2-2(N-2)+\frac{4}{3}(N-2)+\frac{1}{3}\big)(x_l^k)^2 \nonumber \\
        &\hspace{2em}+ \big(4+\frac{4}{3}(N-2)-\frac{2}{3}\big)x_m^k x_l^k \Bigg] \nonumber\\
        &= \mathcal{L}^k + \frac{1}{N(N-1)}\sum\limits_{m\neq l}\Bigg[ \big(-\frac{1}{3}-\frac{2}{3}N \big)(x_m^k)^2 +\big(-\frac{1}{3}-\frac{2}{3}N \big)(x_l^k)^2 \nonumber\\
        &\hspace{2em} +\big(\frac{4}{3}N+\frac{2}{3}\big)x_m^k x_l^k \Bigg] \nonumber \\
        &= \Big(1-\frac{2N+1}{3N(N-1)}\Big)\mathcal{L}^k.
    \end{align}
$\hfill\square$

\subsection{Proof of \cref{prop: 1D bounds lyap}}
\label{sm: proof 1D bounds lyap}
	In order to use the result from \cite{popoviciu1935equations}
	for the upper bound and \cite{nagy1918algebraische} for the lower bound, it suffices to notice that up to normalisation and a constant factor, the Lyapunov sum of square differences of $N$ points $x_1,\cdots, x_N$ is the (biased) empirical variance of the points:
    \begin{align}
        \sum\limits_{i\neq j} (x_i - x_j)^2 &= 2\sum\limits_{i=1}^N Nx_i^2 - 2\left(\sum\limits_{i=1}^N x_i\right)^2 = 2N\left[\sum\limits_{i=1}^N x_i^2 - \frac{1}{N}\left(\sum\limits_{i=1}^N x_i\right)^2 \right]\nonumber \\
        &= 2N \left[\sum\limits_{i=1}^N x_i^2 - \frac{2}{N}\sum\limits_{i=1}^N x_i \sum\limits_{j=1}^N x_j  + N\left(\frac{1}{N} \sum\limits_{j=1}^N x_j\right)^2\right] \nonumber \\
        &= 2N \sum\limits_{i=1}^N \left(x_i - \frac{1}{N}\sum\limits_{j=1}^N x_j\right)^2.
    \end{align}
    
    We rewrite here a proof, without prior knowledge of \cite{popoviciu1935equations,nagy1918algebraische}, based on functional optimization of the Lyapunov given the range of the opinions. Thus, let $\mathcal{L}$ be the function:
    \begin{equation}
        \mathcal{L}(x_1,\cdots,x_N) = \sum\limits_{i\neq j} (x_i-x_j)^2.
    \end{equation}
    
    Note that $\mathcal{L}^k = \mathcal{L}(x_1^k,\cdots,x_N^k)$.
    Let $r>0$ be an arbitrary number and consider $\mathcal{X}_r$ the set of all possible opinion vectors such that their range is equal to $r$:
    \begin{equation}
        \mathcal{X}_r = \{X\in\mathbb{R}^k\mid \max\limits_{i\neq j} \left|x_i-x_j\right| = r\}.
    \end{equation}
    
    We now wish to minimise $\mathcal{L}$ on $\mathcal{X}_r$. Without loss of generality, by performing a permutation of indices, and by translating all opinions by the same displacement, we can assume that $x_1 = \min\limits_i x_i$ and $x_N = \min\limits_i x_N$ are fixed at values $a<b$ distant from $r$. Let $\tilde{\mathcal{L}}$ be the sub function:
    \begin{equation}
        \tilde{\mathcal{L}}(x_2,\cdots,x_{N-1}) = \mathcal{L}(x_1,\cdots,x_N).
    \end{equation}
    
    We have for all indices $i\in\{2,\cdots,N-1\}$:
    \begin{equation}
        \frac{\partial \tilde{\mathcal{L}}}{\partial x_i} = 2\sum\limits_{\substack{j=1\\j\neq i}}^N 2(x_i-x_j) = 4(N-1)x_i - 4\sum\limits_{\substack{j=1\\j\neq i}}^N x_j = 4Nx_i - 4\sum\limits_{j=1}^N x_j,
    \end{equation}
    thus:
    \begin{equation}
        \frac{\partial \tilde{\mathcal{L}}}{\partial x_i} = 0 \iff x_i = \frac{1}{N}\sum\limits_{j=1}^N x_j.
    \end{equation}
    
    Thus at the critical point of $\tilde{\mathcal{L}}$ we have $x_2 = \cdots = x_{N-1}$ and then:
    \begin{equation}
        x_i = \frac{1}{N}(x_1 + (N-2)x_i + x_N),
    \end{equation}
    which leads to:
    \begin{equation*}
        x_i = \frac{x_1+x_N}{2}.
    \end{equation*}
    
    Furthermore, by convexity of the squared function, this provides a global minimum of $\tilde{\mathcal{L}}$ in each of its coordinate directions. Therefore, taking all $x_i$ to be the midpoint of $x_1$ and $x_N$ yields the global minimum for $\tilde{\mathcal{L}}$. By denoting $x_i^{\min}$ these optimal opinions and the $\mathcal{L}_{\min}$ the minimal Lyapunov in $\mathcal{X}_r$, we have:
    \begin{align}
        \mathcal{L}_{\min} &= \sum\limits_{i\neq j} (x_i^{\min} - x_j^{\min})^2\nonumber\\
        &= 2\sum\limits_{i=2}^{N-1}(x_1-x_i^{\min})^2+2\sum\limits_{i=2}^{N-1} (x_N-x_i^{\min})^2+2(x_1-x_N)^2\nonumber\\
        &\hspace{2em}+\sum\limits_{\substack{i\neq j\\(i,j)\in\{2,\cdots,N-1\}^2}}(x_i^{\min}-x_j^{\min})\nonumber\\
        &= 2(N-2)\left(\frac{x_1-x_N}{2}\right)^2 + 2(N-2)\left(\frac{x_N-x_1}{2}\right)^2+2(x_1-x_N)^2\nonumber\\
        &= N(x_N-x_1)^2 = Nr^2.
    \end{align}
    
    As for the maximization, there are no other zeros of the gradient, thus the maximum will be reached on the border of the optimisation domain. We thus use Lagrangian optimization. We continue to assume $x_1 = a$ and $x_N = b$. Since the constraints are $a-x_i\le 0$ and $x_i-b\le 0$ for all $i\in\{2,\cdots,N-1\}$, we define the Lagrangian $L(x,\lambda,\mu)$ to be:
    \begin{equation}
        L(x,\lambda,\mu) = \sum\limits_{i\neq j} (x_i-x_j)^2 + \sum\limits_{i=2}^{N-1}\lambda_i(a-x_i)+\mu_i(x_i-b).
    \end{equation}
    
    Computing the $x$ gradient of $L$, we get for all $i\in\{2,\cdots,N-1\}$:
    \begin{align}
    	\frac{\partial L}{\partial x_i} &= 2\sum\limits_{i\neq j}2(x_i-x_j) - \lambda_i+ \mu_i = 4(N-1)x_i-4\sum\limits_{i\neq j}x_j - \lambda_i+\mu_i \nonumber\\
	&= 4Nx_i-4\sum\limits_{j=1}^Nx_j-\lambda_i+\mu_i.
    \end{align}
    
    We then have:
    \begin{equation}
    	\label{eq: 1D Lagrangian grad 0}
    	\frac{\partial L}{\partial x_i} = 0 \iff x_i = \frac{1}{N}\sum\limits_{j=1}^N x_j + \frac{1}{4N}(\lambda_i - \mu_i).
    \end{equation}
    
    At the maximum configuration, for all $i$, for each constraint, it is either saturated or its Lagrangian multiplier is 0. Thus for all $i\in\{2,\cdots,N-1\}$, we have:
    \begin{equation}
    	\label{eq: 1D Lagrangian alternatives}
    	\left( \begin{cases}  x_i = a \\ \mu_i = 0\end{cases} \text{or}\quad \lambda_i = 0\right) \;\text{and}\; \left(  \begin{cases}  x_i = b \\ \lambda_i = 0\end{cases} \text{or}\quad \mu_i = 0\right).
    \end{equation}
    
    Denote $I_a = \{i\mid x_i=a\}$, $I_b = \{i\mid x_i = b\}$, and $I_c = \{i\mid x_i\notin \{a,b\}\}$.
    
    Assume $I_c\neq \emptyset$. Then for all $i\in I_c$, \cref{eq: 1D Lagrangian alternatives} gives that $\lambda_i = \mu_i = 0$, and then \cref{eq: 1D Lagrangian grad 0} gives:
    \begin{equation}
    	\label{eq: 1D Lagrangian max i in C implies average}
    	x_i = \frac{1}{N}\sum\limits_{j}x_j.
    \end{equation}
    
    Denote $\mathcal{L}_c$ the subfunction of $\mathcal{L}$ using only the indices in $I_c$:
    \begin{equation}
    	\mathcal{L}_c((x_i)_{i\in I_c}) = \mathcal{L}(x),
    \end{equation}
    
    By convexity along each coordinate direction of $\mathcal{L}_c$, \cref{eq: 1D Lagrangian max i in C implies average} gives a local minimum along each dimension and thus a global minimum of $\mathcal{L}_c$. This implies that $x$ cannot be the global maximum of the non constant function $\mathcal{L}$. Thus necessarily for the optimal $x$ maximizing $\mathcal{L}$, we have $I_c = \emptyset$.
    
    Then \cref{eq: 1D Lagrangian alternatives} becomes:
    \begin{equation}
    	\begin{cases}  
		\forall i\in I_a,\, \mu_i = 0 \\
		\forall i\in I_b,\, \lambda_i = 0
	\end{cases}.
    \end{equation}
    
    If $i\in I_a$, then the constraint $a-x_i \neq 0$ is trivially saturated. For this to be possible, we need, according to \cref{eq: 1D Lagrangian grad 0}:
    \begin{equation}
    	a = x_i = \frac{1}{N}\sum\limits_{j=1}^N x_j + \frac{1}{4N}\lambda_i = \frac{\left|I_a\right|}{N}a + \frac{\left|I_b\right|}{N}b + \frac{1}{4N}\lambda_i.
    \end{equation}
    
    Thus:
    \begin{equation}
    	\lambda_i = 4N \Bigg( \left(1-\frac{\left|I_a\right|}{N}\right)a -\frac{\left|I_b\right|}{N}b\Bigg).
    \end{equation}
    
    Similarly, for $i\in I_b$, constraint saturation is possible, according to \cref{eq: 1D Lagrangian grad 0}, if:
    \begin{equation}
    	b = x_i = \frac{1}{N}\sum\limits_{j=1}^N x_j - \frac{1}{4N}\mu_i = \frac{\left|I_a\right|}{N}a + \frac{\left|I_b\right|}{N}b - \frac{1}{4N}\mu_i.
    \end{equation}
    
    Thus:
    \begin{equation}
    	\mu_i = 4N \Bigg( \frac{\left|I_a\right|}{N}a -  \left(1-\frac{\left|I_b\right|}{N}\right)b\Bigg).
    \end{equation}
    
    Then:
    \begin{align}
    	\mathcal{L}(x) &= \sum\limits_{i\neq j}(x_i-x_j)^2 \nonumber\\
	&= \sum\limits_{\substack{i\neq j \\ i\in I_a \\ j\in I_a}}(x_i-x_j)^2 + \sum\limits_{\substack{i\neq j \\ i\in I_a\\ j\in I_b}}(x_i-x_j)^2 + \sum\limits_{\substack{i\neq j \\ i\in I_b\\ j\in I_a}}(x_i-x_j)^2 + \sum\limits_{\substack{i\neq j \\ i\in I_b \\ j\in I_b}}(x_i-x_j)^2 \nonumber \\
	&= 2\left|I_a\right| \left| I_b\right| (b-a)^2.
    \end{align}
    
    For $\mathcal{L}$ to be maximal, we then need $\left|I_a\right| \left| I_b\right|$ to be maximal, under the constraint $\left|I_a\right| + \left| I_b\right| = N$. This happens when:
    \begin{equation}
    	\begin{cases}
		\left|I_a\right|  = \floor{\frac{N}{2}} \\
		\left| I_b\right| = \ceil{\frac{N}{2}}
	\end{cases}
	\text{or}\quad
	\begin{cases}
		\left|I_a\right|  = \ceil{\frac{N}{2}} \\
		\left| I_b\right| = \floor{\frac{N}{2}}
	\end{cases}
    \end{equation}
    leading to the maximal possible value of the Lyapunov given the range $b-a = r$:
    \begin{equation}
    	\mathcal{L}_{\mathrm{max}} = 2 \floor{\frac{N}{2}} \ceil{\frac{N}{2}}r^2 \le \frac{N^2}{2}r^2.
    \end{equation}
$\hfill\square$

\subsection{Proof of \cref{th:1D finite expected time}}
\label{sm: proof 1D finite expected time}
Using \cref{prop: 1D expec lyapunov n+1}, we have by induction on expectations:
    \begin{align}
        \label{eq: 1D Lyapunov function of 0}
        \mathbb{E}\Big(\sum\limits_{i\neq j}(x_i^k-x_j^k)^2\mid X_0)\Big) &= \Big(1-\frac{2N+1}{3N(N-1)}\Big)^k\sum\limits_{i\neq j}(x_i^0-x_j^0)^2 \nonumber\\
        &= \Big(1-\frac{2N+1}{3N(N-1)}\Big)^k\mathcal{L}^0.
    \end{align}
    
    Consider any $\varepsilon>0$. Since $T_\varepsilon'$ is a stopping time, for any time step $n\in\mathbb{N}$, we have the inclusion of random events:
    \begin{equation}
        \{\mathcal{L}^k\le N\varepsilon^2\} \subset \{T_\varepsilon'\le k\}.
    \end{equation}
    
    This gives:
    \begin{align}
        \mathbb{P}(T_\varepsilon'\le k\mid X_0) &\ge \mathbb{P}(\mathcal{L}^k\le N\varepsilon^2\mid X_0) \\
        &\ge 1 - \mathbb{P}(\mathcal{L}^k>N\varepsilon^2 \mid X_0).
    \end{align}
    
    Using the Markov inequality and \cref{eq: 1D Lyapunov function of 0}, we can write:
    \begin{equation}
        \label{eq: 1D markov}
        \mathbb{P}(\mathcal{L}^k > N\varepsilon^2\mid X_0) \le \frac{1}{N\varepsilon^2}\mathbb{E}(\mathcal{L}^k\mid X_0) = \frac{1}{N\varepsilon^2}\Big(1-\frac{2N+1}{3N(N-1)}\Big)^k\mathcal{L}^0.
    \end{equation}
    
    For small $k$, the Markov inequality \cref{eq: 1D markov} provides huge unrealistic bounds. Since probabilities are bounded by 1, we have:
    \begin{equation}
        \mathbb{P}(\mathcal{L}^k > N\varepsilon^2\mid X_0) \le \min\left\{\frac{1}{N\varepsilon^2}\Big(1-\frac{2N+1}{3N(N-1)}\Big)^k\mathcal{L}^0 ,1\right\}.
    \end{equation}
    
    We thus have:
    \begin{equation}
        \mathbb{P}(T_\varepsilon'\le k\mid X_0) \ge 1-\min\left\{\frac{1}{N\varepsilon^2}\Big(1-\frac{2N+1}{3N(N-1)}\Big)^k\mathcal{L}^0 ,1\right\}.
    \end{equation}
    
    For formula simplicity, we use concavity of the logarithm to write:
    \begin{equation}
        \frac{1}{N\varepsilon^2}\left(1-\frac{2N+1}{3N(N-1)}\right)^k \mathcal{L}^0 \le \frac{1}{N\varepsilon^2}e^{-k\frac{2N+1}{3N(N-1)}}\mathcal{L}^0.
    \end{equation}
    
    We can then find sufficiently large $k$ for the Markov inequality:
    \begin{equation}
        \frac{1}{N\varepsilon^2}e^{-k\frac{2N+1}{3N(N-1)}}\mathcal{L}^0 \le 1 \iff k \ge \frac{3N(N-1)}{2N+1}\ln\left(\frac{\mathcal{L}^0}{N\varepsilon^2}\right).
    \end{equation}
    
    We can then bound the expectation of $T_\varepsilon'$ as follows:
    \begin{align}
        \mathbb{E}(T_\varepsilon'\mid X_0) &= \sum\limits_{k=0}^\infty \mathbb{P}(T_\varepsilon' > k) \nonumber\\ 
        &= \sum\limits_{k=0}^\infty \Big(1-\mathbb{P}(T_\varepsilon'\le k)\Big) \nonumber\\
        &\le \frac{3N(N-1)}{2N+1}\ln\left(\frac{\mathcal{L}^0}{N\varepsilon^2}\right) + \frac{\mathcal{L}^0}{N\varepsilon^2}\sum\limits_{k=\ceil{\frac{3N(N-1)}{2N+1}\ln\left(\frac{\mathcal{L}^0}{N\varepsilon^2}\right)}}^\infty \Big(1-\frac{2N+1}{3N(N-1)}\Big)^k \nonumber\\
        &\le \frac{3N(N-1)}{2N+1}\ln\left(\frac{\mathcal{L}^0}{N\varepsilon^2}\right) \nonumber \\
        &\hspace{2em}+ \Big(1-\frac{2N+1}{3N(N-1)}\Big)^{\ceil{\frac{3N(N-1)}{2N+1}\ln\left(\frac{\mathcal{L}^0}{N\varepsilon^2}\right)}}  \frac{1}{1-(1-\frac{2N+1}{3N(N-1)})}\frac{\mathcal{L}^0}{N\varepsilon^2} \nonumber\\
        &\le \frac{3N(N-1)}{2N+1}\ln\left(\frac{\mathcal{L}^0}{N\varepsilon^2}\right) + \epsilon^0(N,\varepsilon),
    \end{align}
    where $\epsilon^0(N,\varepsilon) = \frac{\mathcal{L}^0}{N\varepsilon^2} \frac{3N(N-1)}{2N+1} \left(1-\frac{2N+1}{3N(N-1)}\right)^{\ceil{\frac{3N(N-1)}{2N+1} \ln\left(\frac{\mathcal{L}^0}{N\varepsilon^2}\right) }}$.
    
    To bound $\epsilon^0(N,\varepsilon)$:
    \begin{align}
        \label{al: bound epsilon 0}
        \epsilon^0(N,\varepsilon) &= \frac{\mathcal{L}^0}{N\varepsilon^2} \frac{3N(N-1)}{2N+1} e^{\ln\left(1-\frac{2N+1}{3N(N-1)}\right)\ceil{\frac{3N(N-1)}{2N+1} \ln\left(\frac{\mathcal{L}^0}{N\varepsilon^2}\right) } } \nonumber \\
        &\le \frac{\mathcal{L}^0}{N\varepsilon^2} \frac{3N(N-1)}{2N+1} e^{-\frac{2N+1}{3N(N-1)}\left(\frac{3N(N-1)}{2N+1} \ln\left(\frac{\mathcal{L}^0}{N\varepsilon^2}\right) \right)} \nonumber \\
        &\le \frac{\mathcal{L}^0}{N\varepsilon^2} \frac{3N(N-1)}{2N+1} e^{-\ln\left(\frac{\mathcal{L}^0}{N\varepsilon^2}\right)} \nonumber \\
        &\le \frac{3N(N-1)}{2N+1}.
    \end{align}
    
    Using \cref{prop: 1D T<=T'}, we conclude:
    \begin{equation}
        \mathbb{E}(T_\varepsilon\mid X_0) \le \mathbb{E}(T_\varepsilon'\mid X_0) \le  \frac{3N(N-1)}{2N+1}\ln\left(\frac{\mathcal{L}^0}{N\varepsilon^2}\right) + \frac{3N(N-1)}{2N+1}.
    \end{equation}
$\hfill\square$

\subsection{Proof of \cref{prop: 1D expec X_bar}}
\label{sm: proof 1D expec X_bar}
Using the well-known first moment of uniform random variables (see \cref{prop: E(x_i) and E(x_i^2)}), we get:
    \begin{align}
        \mathbb{E}(\bar{X}_{k+1}\mid X_k) &= \sum\limits_{i\neq j} \mathbb{E}_{i,j}(\bar{X}_{k+1}\mid X_k)\mathbb{P}\big((A_{k+1},B_{k+1}) = (i,j)\big) \nonumber\\
        &=\frac{1}{N^2(N-1)}\sum\limits_{i\neq j}\mathbb{E}_{i,j}(\sum\limits_{s\notin\{i,j\}}x_s^k+x_i^{k+1}+x_j^{k+1}\mid X_k) \nonumber\\
        &=\bar{X}_k + \frac{1}{N^2(N-1)}\sum\limits_{i\neq j}\mathbb{E}_{i,j}(x_i^{k+1}+x_j^{k+1} - x_i^k-x_j^k\mid X_k) \nonumber\\
        &=\bar{X}_k + \frac{1}{N^2(N-1)}\sum\limits_{i\neq j} \frac{x_i^k+x_j^k}{2} + \frac{x_i^k+x_j^k}{2} - x_i^k-x_j^k \nonumber\\
        &=\bar{X}_k.
    \end{align}
    
    Therefore, for all $k\in\mathbb{N}$:
    \begin{equation}
        \mathbb{E}(\bar{X}_{k}\mid X_0) = \bar{X}_0.
    \end{equation}
$\hfill\square$

\subsection{Proof of \cref{prop: 1D expect X_n}}
\label{sm: proof 1D expect X_n}
If we choose $(i,j)$ at time step $k+1$, then 
    \begin{equation}
        X_{k+1} = X_k + 
        \begin{pmatrix} 0&\cdots&0&x_i^{k+1}-x_i^k&0&\cdots&0&x_j^{k+1}-x_j^k&0&\cdots & 0 \end{pmatrix}^T.
    \end{equation}
    
    Thus according to \cref{eq: 1D expec cond i-j xi next}:
    \begin{align}
        \mathbb{E}(X_{k+1}\mid X_k) &= \frac{1}{N(N-1)}\sum\limits_{i\neq j}\mathbb{E}_{i,j}(X_{k+1}\mid X_k)\nonumber\\
        &= X_k \nonumber\\
        &{\kern-2em} +\frac{1}{N(N-1)}\sum\limits_{i\neq j}\begin{pmatrix} 0&\cdots&0&\frac{x_j^{k}-x_i^k}{2}&0&\cdots&0&\frac{x_i^{k}-x_j^k}{2}&0&\cdots & 0 \end{pmatrix}^T\nonumber\\
        &=X_k + \frac{1}{2N(N-1)}\sum\limits_{i=1}^N \begin{pmatrix} x_i^{k}-x_1^k\\\vdots\\x_i^{k}-x_{i-1}^k\\
        \sum\limits_{j\neq i}x_j^k - (N-1)x_i^k
        \\x_i^k-x_{i+1}^k\\\vdots\\x_i^k-x_N^k
        \end{pmatrix} \nonumber\\
        &= X_k + \frac{1}{2N(N-1)}\begin{pmatrix} \sum\limits_{j\neq 1} x_j^k - (N-1)x_1^k + \sum\limits_{j\neq 1} x_j^k - (N-1)x_1^k\\\vdots\\ \sum\limits_{j\neq N} x_j^k - (N-1)x_N^k + \sum\limits_{j\neq N} x_j^k - (N-1)x_N^k\end{pmatrix} \nonumber\\
        &= X_k + \frac{1}{N(N-1)}
        \begin{pmatrix} 
        -(N-1) & 1 & \cdots & \cdots & 1 \\
        1& \ddots& \ddots & & \vdots \\
        \vdots & \ddots & \ddots & \ddots & \vdots \\
        \vdots & & \ddots & \ddots & 1 \\
        1 & \cdots & \cdots & 1 & -(N-1)
        \end{pmatrix} X_k \nonumber\\
        &=\Big((1-\frac{1}{N}-\frac{1}{N(N-1)})I_N + \frac{1}{N(N-1)}\begin{pmatrix} 1 & \cdots & 1 \\\vdots & & \vdots \\ 1 & \cdots & 1\end{pmatrix}\Big) X_k \\
        &= \Big((1-\frac{1}{N-1})I_N+\frac{1}{N(N-1)}1_N 1_N^T \Big)X_k \nonumber\\
        &= \Big((1-\frac{1}{N-1})I_N+\frac{1}{N(N-1)}U\diag(N,0,\cdots,0)U^T \Big)X_k \nonumber\\
        &= U\diag(1,1-\frac{1}{N-1},\cdots,1-\frac{1}{N-1})U^T X_k.
    \end{align}
    
    Therefore, by induction and linearity of the expectation:
    \begin{equation}
        \mathbb{E}(X_k\mid X_0) = U\diag\Big(1,(1-\frac{1}{N-1})^k,\cdots,(1-\frac{1}{N-1})^k\Big)U^T X_0.
    \end{equation}
$\hfill\square$

\subsection{Proof of \cref{th: 1D gossip}}
\label{sm: proof 1D gossip}
Denote $r_k$ the range of opinions at step $k$.
	
	Assume $X_0$ is fixed. Denote $f$ the one dimensional function defined as $f(x) = \lVert X_k-x\begin{pmatrix} 1,\cdots,1\end{pmatrix}^T \rVert_2^2$, defined over $\left[\min X_k,\max X_k\right]$. By convexity, $f$ is minimal when $x$ is the average of $X_k$ and is maximal on its domain border. We easily see that $f(\min X_k) \le (N-1)r_k^2$ and $f(\max X_k) \le (N-1)r_k^2$. Thus, we have:
	\begin{equation}
		\lVert X_k-x_\infty \begin{pmatrix} 1,\cdots,1\end{pmatrix}^T\rVert_2 \ge \varepsilon \implies r_k\ge \frac{\varepsilon}{\sqrt{N-1}}.
	\end{equation}

	By taking the probability, this implies:
	\begin{equation}
		\label{eq: 1D gossip prob X bigger implies r bigger}
		\mathbb{P}\left(\frac{\lVert X_k -  x_\infty \begin{pmatrix} 1,\cdots,1\end{pmatrix}^T \rVert_2}{\lVert X_0\rVert_2}\ge\varepsilon\mid X_0\right) \le \mathbb{P}\left(r_k^2 \ge \frac{\varepsilon^2 \lVert X_0\rVert_2^2}{N-1}\mid X_0\right)
	\end{equation}

	Using \cref{prop: 1D bounds lyap}, we have:
	\begin{equation}
		r_k^2 \ge \frac{\varepsilon^2 \lVert X_0\rVert_2^2}{N-1} \implies \mathcal{L}^k \ge \frac{N\varepsilon^2\lVert X_0 \rVert_2^2}{N-1}.
	\end{equation}
	
	Thus by taking the probabilities and plugging into \cref{eq: 1D gossip prob X bigger implies r bigger}, we get:
	\begin{equation}
		\mathbb{P}\left(\frac{\lVert X_k -  x_\infty \begin{pmatrix} 1,\cdots,1\end{pmatrix}^T \rVert_2}{\lVert X_0\rVert_2}\ge\varepsilon\mid X_0\right) \le \mathbb{P}\left(\mathcal{L}^k \ge \frac{N\varepsilon^2\lVert X_0 \rVert_2^2}{N-1}\mid X_0\right).
	\end{equation}
	
	By applying the Markov inequality, and recalling \cref{prop: 1D expec lyapunov n+1}, we then have:
	\begin{equation}
		\label{eq: 1D gossip prob bounded quotient L0 X0}
		\mathbb{P}\left(\frac{\lVert X_k -  x_\infty 1_N \rVert_2}{\lVert X_0\rVert_2}\ge\varepsilon\mid X_0\right) \le \frac{\mathbb{E}(\mathcal{L}^k\mid X_0)}{\frac{N\varepsilon^2\lVert X_0 \rVert_2^2}{N-1}} = \frac{(1-\frac{2N+1}{3N(N-1)})^k\mathcal{L}^0}{\frac{N\varepsilon^2\lVert X_0 \rVert_2^2}{N-1}}.
	\end{equation}
	
	Furthermore, calculations give:
	\begin{align}
		\frac{\mathcal{L}^0}{\lVert X_0 \rVert_2^2} &= \frac{\sum\limits_{i,j}(x_i^0 - x_j^0)^2}{\sum\limits_{i=1}^N (x_i^0)^2} = \frac{\sum\limits_{i=1}^N N(x_i^0)^2 -2 \sum\limits_{i=1}^N\sum\limits_{j=1}^N(x_i^0)(x_j^0) + N\sum\limits_{j=1}^N (x_j^0)^2}{\sum\limits_{i=1}^N (x_i^0)^2} \nonumber\\
		&= 2N - 2\frac{\left(\sum\limits_{i=1}^N x_i^0\right)^2}{\sum\limits_{i=1}^N (x_i^0)^2} 
		\le \begin{cases}2N - 2\frac{(Na)^2}{Nb^2} = 2N\left(1-\frac{a^2}{b^2}\right) &\text{if } a\ge 0 \\
		2N\left(1-\frac{b^2}{a^2}\right) &\text{if } b\le 0 \\
		2N &\text{otherwise}\end{cases} \nonumber\\
		&=2N(1-q_{a,b}).
	\end{align}
	
	Plugging into \cref{eq: 1D gossip prob bounded quotient L0 X0}, we then get:
	\begin{equation}
		\mathbb{P}\left(\frac{\lVert X_k -  x_\infty 1_N \rVert_2}{\lVert X_0\rVert_2}\ge\varepsilon\mid X_0\right) \le  \left(1-\frac{2N+1}{3N(N-1)}\right)^k \frac{2(N-1)\left(1-q_{a,b}\right)}{\varepsilon^2}.
	\end{equation}
	
	Comparing to $\varepsilon$, we have:
	\begin{align}
		&\left(1-\frac{2N+1}{3N(N-1)}\right)^k \frac{1}{\varepsilon^2}\times 2(N-1)\left(1-q_{a,b}\right) \le \varepsilon \nonumber\\
		&\hspace{2em}\iff \left(1-\frac{2N+1}{3N(N-1)}\right)^k \le \frac{\varepsilon^3}{2(N-1)\left(1-q_{a,b}\right)} \nonumber \\
		&\hspace{2em}\iff k\ge \frac{-3\ln\varepsilon +\ln(N-1)+\ln2+\ln\left(1-q_{a,b}\right)}{-\ln\left( 1-\frac{2N+1}{3N(N-1)}\right)}.
	\end{align}
	
	We then upper-bound $T_{gossip}(\varepsilon)$ by infimum over $k$:
	\begin{equation}
		T_{gossip}(\varepsilon) \le  \frac{-3\ln\varepsilon +\ln(N-1)+\ln2+\ln\left(1-q_{a,b}\right)}{-\ln\left( 1-\frac{2N+1}{3N(N-1)}\right)}.
	\end{equation}
	
	We can further simplify this equation by using:
	\begin{equation}
	    \frac{1}{-\ln\left( 1-\frac{2N+1}{3N(N-1)}\right)} \le \frac{3N(N-1)}{2N+1} \le  \frac{3}{2}N,
	\end{equation}
	which we had already used for formula simplification in \cref{th:1D finite expected time}, along with $\ln(N-1)\le\ln{N}$.
	
$\hfill\square$

\subsection{Comparison between expected convergence time and gossiping}
\label{sm: comparison expect cv time and gossip}
We here prove a fundamental similarity between expected convergence time and gossiping under reasonable assumptions. At the end of the analysis we explain how to relate it to our problem. 

Let $Y_k$ be a sequence of real valued random variables that is positive and also an exponentially decreasing super-martingale (EDSM) converging almost surely to $0$, i.e.:

\begin{enumerate}[label=(\roman*)]
	\item (Positivity) For all $k\in\mathbb{N}$, $Y_k \ge 0$,
	\item (EDSM) There exists $\alpha>0$ such that for all $k\in\mathbb{N}$, $\mathbb{E}(Y_{k+1}\mid Y_k) = (1-\alpha)Y_k$,
	\item (Convergence) $Y_k \xrightarrow[]{} 0$ almost surely.
\end{enumerate}

For any $\varepsilon>0$, we would traditionally say that $Y_k$ has $\varepsilon$- converged if $Y_k\le \varepsilon$. However, in order to compare better to gossiping, we normalise this definition by its initial value. Thus we will here say that $Y_k$ has converged if $\frac{Y_k}{Y_0}\le \varepsilon$.

\begin{definition}
	For any $\varepsilon>0$, we denote $T_{cv}(\varepsilon)$ the stopping time with respect to the natural filtration defined as:
	$$ T_{cv}(\varepsilon) = \min\{n\in\mathbb{N}\mid \frac{Y_k}{Y_0}\le \varepsilon\}. $$
\end{definition}

Since $Y_k$ converges to $0$ almost surely, we can define the following gossip $\varepsilon$-convergence time.

\begin{definition}
	For any $\varepsilon>0$, we denote $T_{g}(\varepsilon)$ the gossip $\varepsilon$-convergence time defined as:
	$$ T_g(\varepsilon) = \sup\limits_{Y_0}\inf\limits_{k\in\mathbb{N}}\left\{k\mid\mathbb{P}\left(\frac{Y_k}{Y_0}\ge\varepsilon\mid Y_0\right)\le\varepsilon\right\}. $$
\end{definition}

\begin{theorem}
	For any $\varepsilon >0$, we have the following upper bounds:
	\begin{align*}
		&\mathbb{E}\left(T_{cv}(\varepsilon)\mid Y_0\right) \le \frac{-\ln\varepsilon}{-\ln(1-\alpha)} + \frac{1}{\alpha} \\
		&T_g(\varepsilon) \le \frac{-\ln\varepsilon}{-\ln(1-\alpha)}.
	\end{align*}
\end{theorem}

\begin{proof}
	The proof is based on the Markov inequality for both quantities. 
	First, let us focus on $T_{cv}(\varepsilon)$. If we have not yet $\varepsilon$-converged, then necessarily $Y_k>\varepsilon Y_0 $. This implies, without necessarily having equality, that for all $k$:
	\begin{equation}
		\label{eq: time convergence vs gossip prop Tcv lower Prob tail}
		\mathbb{P}(T_{cv}(\varepsilon)>k\mid Y_0) \le \mathbb{P}(Y_k > \varepsilon Y_0 \mid Y_k).
	\end{equation}
	
	Recall that the expectation can be written as a sum of tail distributions:
	\begin{equation}
		\mathbb{E}(T_{cv}(\varepsilon)\mid Y_0) = \sum\limits_{k\ge 0}\mathbb{P}(T_{cv}(\varepsilon)>n\mid Y_0).
	\end{equation}
	
	We then have, using the Markov inequality and by naive $1$-bounding:
	\begin{align}
		\mathbb{E}(T_{cv}(\varepsilon)\mid Y_0) &\le \sum\limits_{k\ge 0} \min\left\{ \frac{\mathbb{E}(Y_k\mid Y_0)}{\varepsilon Y_0}, 1\right\} = \sum\limits_{k\ge 0} \min\left\{ \frac{(1-\alpha)^k}{\varepsilon}, 1\right\} \\
		&\le \frac{-\ln\varepsilon}{-\ln(1-\alpha)} + \frac{1}{\varepsilon}\frac{(1-\alpha)^{\frac{-\ln\varepsilon}{-\ln(1-\alpha)}}}{\alpha} = \frac{-\ln\varepsilon}{-\ln(1-\alpha)} + \frac{1}{\alpha}.
	\end{align}
	
	Focus now on $T_g(\varepsilon)$. Once again, the Markov inequality provides:
	\begin{equation}
		\mathbb{P}(Y_k \ge \varepsilon Y_0 \mid Y_0) \le \frac{(1-\alpha)^k}{\varepsilon}.
	\end{equation}
	
	The bound is smaller than $\varepsilon$ if and only if $k\ge \frac{-\ln\varepsilon}{-\ln(1-\alpha)}$. Therefore by minimality, and since this is true for all $Y_0$, we have:
	\begin{equation}
		T_g(\varepsilon) \le  \frac{-\ln\varepsilon}{-\ln(1-\alpha)}.
	\end{equation}
\end{proof}

On the other hand, lower bounding both times is more challenging. Under further assumptions, we can prove that $T_g(\varepsilon)$ is lower bounded by a $\Omega(-\ln\varepsilon)$ term. However, lower bounding $T_{cv}(\varepsilon)$ efficiently is more difficult.

\begin{theorem}
	Further assume that there exists $C>0$ such that for all $k$, $Y_k$ is bounded by $C$, i.e. $Y_k \le C$. Let $C' = \sup Y_0 \le C$. For any $\varepsilon>0$, then:
	$$ T_g(\varepsilon) \ge \frac{-\ln\varepsilon}{-\ln(1-\alpha)} - \frac{\ln(1+\frac{C}{C'}-\varepsilon)}{-\ln(1-\alpha)}. $$ 
\end{theorem}

\begin{proof}
	The proof is based on an inequality very similar to the Markov inequality but for lower bounding tails of bounded random variables. Indeed, for any positive random variable $Y$ upper bounded by $C$, then for any $0<\eta<C$, $\mathbb{P}(Y\ge\eta) \ge \frac{\mathbb{E}(Y) - \eta}{C-\eta}$. The proof is analogous to the one of the Markov inequality, namely by looking at when $Y$ is greater or smaller than $\eta$ in the expectation. Indeed, $Y\le \eta\mathbbm{1}_{Y< \eta} + C\mathbbm{1}_{Y\ge \eta}$. By taking the expectation, we have $\mathbb{E}(Y) \le \eta(1-\mathbb{P}(Y\ge\eta)) + C\mathbb{P}(Y\ge\eta)$. Subtracting by $\eta$ and then dividing by $\eta$ gives the desired result.
	
	With the assumption that $Y_k$ is bounded, we then have:
	\begin{equation}
		\mathbb{P}(Y_k\ge\varepsilon Y_0\mid Y_0) \ge \frac{\mathbb{E}(Y_k\mid Y_0)-\varepsilon Y_0}{C-\varepsilon Y_0} = \frac{(1-\alpha)^k-\varepsilon}{\frac{C}{Y_0}-\varepsilon}.
	\end{equation}
	
	Comparing the lower bound gives the following condition on $k$:
	\begin{equation}
		\frac{(1-\alpha)^k-\varepsilon}{\frac{C}{Y_0}-\varepsilon} > \varepsilon \iff k < \frac{-\ln\varepsilon - \ln\left(1+\frac{C}{Y_0}-\varepsilon\right)}{-\ln(1-\alpha)}
	\end{equation}
	
	Thus, for $k < \frac{-\ln\varepsilon - \ln\left(1+\frac{C}{Y_0}-\varepsilon\right)}{-\ln(1-\alpha)}$, we have $\mathbb{P}(\frac{Y_k}{Y_0}\ge \varepsilon\mid Y_0) >\varepsilon$. This implies that, by considering random initialisations ever so close to $\sup Y_0$ (which can  probabilistically happen):
	\begin{equation}
		T_g(\varepsilon) \ge \sup\limits_{Y_0} \frac{-\ln\varepsilon - \ln\left(1+\frac{C}{Y_0}-\varepsilon\right)}{-\ln(1-\alpha)} = \frac{-\ln\varepsilon - \ln\left(1+\frac{C}{\sup Y_0}-\varepsilon\right)}{-\ln(1-\alpha)},
	\end{equation}
	 which is the desired result.
\end{proof}

However, we cannot apply the same strategy to yield a reasonable bound on $\mathbb{E}(T_{cv}(\varepsilon)\mid Y_0)$. First, we do not have equality in \cref{eq: time convergence vs gossip prop Tcv lower Prob tail} and so we cannot blindly apply the lower bounding strategy to each term of the sum in the expectation. Second, even if we use stronger assumptions such that $Y_k$ is non increasing almost surely, which gives equality in \cref{eq: time convergence vs gossip prop Tcv lower Prob tail}, then naively applying the lower bound Markov-like inequality will provide the lower bound $\frac{1}{\alpha}-\frac{1}{(\nicefrac{Y_0}{\varepsilon})-1}\frac{\ln\left(\nicefrac{Y_0}{\varepsilon}\right)}{-\ln(1-\alpha)}$, which converges to the non informative constant $\frac{1}{\alpha}$ when $\varepsilon$ goes to $0$, rather than a desired logarithmic divergence.

In order to compare with our results, we need to introduce a non standard alternative definition for the gossip time in our problem $T_{gossip}(\varepsilon)$.

\begin{definition}
	For any $\varepsilon>0$, let $\tilde{T}_{gossip}(\varepsilon)$ be the non standard ``$\varepsilon$-averaging time'' for gossiping defined as:
	\begin{align*} 
		\tilde{T}_{gossip}(\varepsilon) &= \sup\limits_{X_0\in I^N} \inf\limits_{k\in\mathbb{N}} \left\{k \mid   \mathbb{P}\left(\frac{\lVert X_k -  x_\infty \begin{pmatrix} 1,\cdots,1\end{pmatrix}^T \rVert_2^2}{\lVert X_0\rVert_2^2} \ge \varepsilon \mid X_0\right) \le \varepsilon  \right\} \\
		&=\sup\limits_{X_0\in I^N} \inf\limits_{k\in\mathbb{N}} \left\{k \mid   \mathbb{P}\left(\frac{\lVert X_k -  x_\infty \begin{pmatrix} 1,\cdots,1\end{pmatrix}^T \rVert_2}{\lVert X_0\rVert_2} \ge \sqrt{\varepsilon} \mid X_0\right) \le \varepsilon  \right\}
	\end{align*}
\end{definition}

We can now compare the general upper bound results linking expected convergence and gossiping in our problem. First, consider $Y_k = \mathcal{L}^k$. In the expected convergence time approach, we compare $\mathcal{L}^k$ to $N\varepsilon^2$. Thus, we are essentially studying $T_\varepsilon' = T_{cv}(\frac{N\varepsilon^2}{\mathcal{L}^0})$ for $Y_k = \mathcal{L}^k$. On the other hand, in standard gossiping, i.e. for $T_{gossip}(\varepsilon)$, we are comparing $\mathcal{L}^k$ to $\frac{N\varepsilon^2\lVert X_0\rVert_2^2}{N-1}$, but the tail of the distribution must still be smaller than $\varepsilon$. Therefore, it cannot be simply related to what we have derived here. However, with the non standard gossip definition, i.e. for $\tilde{T}_{gossip}(\varepsilon)$, the same approach yields that we are comparing $\mathcal{L}^k$ to $\frac{N\varepsilon\lVert X_0\rVert_2^2}{N-1}$, with a tail to be compared to $\varepsilon$. Thus, if we here use $Y_k = \frac{N-1}{N\lVert X_0\rVert_2^2}\mathcal{L}^k$, this comes down to $\tilde{T}_{gossip}(\varepsilon) \approx T_g(\varepsilon)$. Note that both $\mathcal{L}^k$ and $\frac{N-1}{N\lVert X_0\rVert_2^2}\mathcal{L}^k$ are EDSM with same $\alpha = \frac{2N+1}{3N(N-1)}$, for which we have used in the main paper for simplifying expressions that $\frac{1}{\alpha} \le \frac{3}{2}N$. The discrepancies and complications in the comparison solely come from the fact that the choices made in the definition of $T_{gossip}(\varepsilon)$ in the traditional gossip literature are somewhat arbitrary. We suggest to consider, when possible, a richer quantity for gossiping in general: $T_{gossip}^{alternative}(\varepsilon_1,\varepsilon_2)$, defined afterwards.

\begin{definition}
	For any $\varepsilon_1,\varepsilon_2>0$, we denote $T_{gossip}^{alternative}(\varepsilon_1,\varepsilon_2)$ an alternative non standard ``$\varepsilon$-averaging time'' for gossiping defined as:
	\begin{equation*} 
		T_{gossip}^{alternative}(\varepsilon_1,\varepsilon_2) = \sup\limits_{X_0\in I^N} \inf\limits_{k\in\mathbb{N}} \left\{k \mid   \mathbb{P}\left(\frac{\lVert X_k -  x_\infty \begin{pmatrix} 1,\cdots,1\end{pmatrix}^T \rVert_2^2}{\lVert X_0\rVert_2^2} \ge \varepsilon_1 \mid X_0\right) \le \varepsilon_2  \right\}
	\end{equation*}
\end{definition}

In particular, we have:
\begin{equation}
	\begin{cases}
		\tilde{T}_{gossip}(\varepsilon) \approx T_{gossip}^{alternative}(\varepsilon,\varepsilon) \\
		T_{gossip}(\varepsilon) \approx T_{gossip}^{alternative}(\varepsilon^2,\varepsilon)
	\end{cases}
\end{equation}

\section{Unconstrained \texorpdfstring{\boldmath$D$}{D}-dimensional case}

\subsection{Proof of \cref{prop: ND bounds lyap}}
\label{sm: proof ND bounds lyap}
Consider the problem on a per coordinate basis. For each $d\in\{1,\cdots,D\}$, we have using \cref{prop: 1D bounds lyap}, if we denote $r_d^k = \max\limits_{i\neq j}\lvert x_{i,d}^k-x_{j,d}^k\rvert$:
    \begin{equation}
        N(r_d^k)^2 \le \mathcal{L}_d^k \le \frac{N^2}{2}(r_d^k)^2.
    \end{equation}
    
    By summing up these equations, we get:
    \begin{equation}
        N\sum\limits_{d=1}^D (r_d^k)^2 \le \mathcal{L}_T^k \le \frac{N^2}{2} \sum\limits_{d=1}^D(r_d^k)^2.
    \end{equation}
    
    To conclude, we need only notice that by maximality:
    \begin{equation}
        \frac{1}{N} \sum\limits_{d=1}^D(r_d^k)^2\le \max\limits_{d} (r_d^k)^2 \le \max\limits_{i\neq j}\lVert x_i^k - x_j^k\rVert_2^2 \le \sum\limits_{d=1}^D (r_d^k)^2.
    \end{equation}

    The lower bound is tight. This is once again a consequence of \cite{nagy1918algebraische}. Given $N$ points $x_1,\cdots,x_N$ in a Euclidean $D$-dimensional space, assume without loss of generality that $\lVert x_N - x_1\rVert_2 = \max\limits_{i\neq j}\lVert x_i-x_j\rVert_2$. For all $i$, denote $\tilde{x}_i$ the orthonormal projection of $x_i$ onto the line $(x_1 x_N)$. Then, by maximality of $\lVert x_N - x_1\rVert_2$, the points $\tilde{x}_i$ lie in the segment $\left[x_1,x_N\right]$. Furthermore, the Pythagoras inequality gives:
    \begin{equation}
        \sum\limits_{i\neq j} \lVert \tilde{x}_i-\tilde{x}_j\rVert_2^2 \le \sum\limits_{i\neq j} \lVert x_i-x_j\rVert_2^2.
    \end{equation}
    
    Thus the configuration yielding the minimum sum of squared difference norms consists of points $x_2,\cdots,x_{N-1}$ that lie on the segment $\left[x_1,x_N\right]$. We can reparametrise these configurations and return to the dimensional problem given in \cref{prop: 1D bounds lyap}. The minimum is then reached when all points are the average of $x_1$ and $x_N$ yielding the desired bound.
    
    On the other hand, the upper bound is loose and pessimistic as the maximisation domain being the intersection of the $2$-balls around $x_1$ and $x_N$ of radius their distance does not allow as soon as $D\ge 2$ for a simultaneous maximisation of the sum of squared differences along each dimension.
$\hfill\square$

\subsection{Proof of \cref{th:nD finite expected time}}
\label{sm: proof nD finite expected time}
The proof is in essence similar to the one in the one dimensional case of \cref{th:nD finite expected time}. Using \cref{prop: ND T<=T'}, we have that if we have not reached convergence by step $k$, then $\mathcal{L}_T^k > N\varepsilon^2$. This implies that the probability of the event $\{\mathcal{L}_T^k > N\varepsilon^2\}$ is larger than of the event $T_\varepsilon > k$. Then:
    \begin{equation}
        \label{eq: ND expectation as sum of tails}
        \mathbb{E}(T_\varepsilon \mid X_0) = \sum\limits_{k=0}^\infty \mathbb{P}(T_\varepsilon > k\mid X_0) \le \sum\limits_{k=0}^\infty \mathbb{P}(\mathcal{L}_T^k > N\varepsilon^2 \mid X_0).
    \end{equation}
    
    Furthermore, each coordinate of the system, i.e. each column of $X_k$, follows the one-dimensional motion rule from \cref{eq:1D evolution}. Thus, the expectation of $\mathcal{L}_d^k$ for each $d$ is given by \cref{prop: 1D expec lyapunov n+1}:
    \begin{equation}
        \mathbb{E}(\mathcal{L}_d^k\mid X_0) = \left(1-\frac{2N+1}{3N(N-1)}\right)^k\mathcal{L}_d^0.
    \end{equation}
    
    The Markov inequality gives that:
    \begin{equation}
        \mathbb{P}(\mathcal{L}_T^k > N\varepsilon^2 \mid X_0) \le \frac{\mathbb{E}(\mathcal{L}_T^k\mid X_0)}{N\varepsilon^2} = \frac{ \sum\limits_{d=1}^D \mathbb{E}(\mathcal{L}_d^k\mid X_0)}{N\varepsilon^2} = \frac{\left(1-\frac{2N+1}{3N(N-1)}\right)^k \sum\limits_{d=1}^D \mathcal{L}_d^0}{N\varepsilon^2}.
    \end{equation}
    
    The bound provided by this inequality is reasonable when it is smaller than one, which happens when:
    \begin{equation}
        \frac{\left(1-\frac{2N+1}{3N(N-1)}\right)^k \sum\limits_{d=1}^D \mathcal{L}_d^0}{N\varepsilon^2} \le 1 \iff k\ge \frac{\ln\Bigg(\frac{\sum\limits_{d=1}^D \mathcal{L}_d^0}{N\varepsilon^2}\Bigg)}{-\ln\left(1-\frac{2N+1}{3N(N-1)}\right)}.
    \end{equation}
    
    For formula simplicity, as $\ln(x)\le x-1$, we have a simpler sufficient condition for the Markov inequality to be meaningful, i.e. providing a bound lower than 1:
    \begin{equation}
        k \ge \frac{3N(N-1)}{2N+1}\ln\Bigg(\frac{\sum\limits_{d=1}^D \mathcal{L}_d^0}{N\varepsilon^2}\Bigg).
    \end{equation}
    
    Thus \cref{eq: ND expectation as sum of tails} becomes:
    \begin{align}
        \mathbb{E}(T_\varepsilon\mid X_0) &\le \sum\limits_{k=1}^\infty \min\left\{\frac{1}{N\varepsilon^2}\left(1-\frac{2N+1}{3N(N-1)}\right)^k\sum\limits_{d=1}^D \mathcal{L}_d^0,1\right\} \nonumber\\
        &\le \frac{3N(N-1)}{2N+1}\ln\Bigg(\frac{\sum\limits_{d=1}^D \mathcal{L}_d^0}{N\varepsilon^2}\Bigg) \nonumber\\
        &\hspace{2em}+ \frac{\sum\limits_{d=1}^D \mathcal{L}_d^0}{N\varepsilon^2}\frac{1}{1-\left(1-\frac{2N+1}{3N(N-1)}\right)}\left(1-\frac{2N+1}{3N(N-1)}\right)^{\ceil{\frac{3N(N-1)}{2N+1}\ln\Bigg(\frac{\sum\limits_{d=1}^D \mathcal{L}_d^0}{N\varepsilon^2}\Bigg)}} \nonumber\\
        &\le \frac{3N(N-1)}{2N+1}\ln\Bigg(\frac{\sum\limits_{d=1}^D \mathcal{L}_d^0}{N\varepsilon^2}\Bigg) + \epsilon_D^0(N,\varepsilon),
    \end{align}
    where $\epsilon_D^0 = \frac{\sum\limits_{d=1}^D \mathcal{L}_d^0}{N\varepsilon^2}\frac{3N(N-1))}{2N+1}\left(1-\frac{2N+1}{3N(N-1)}\right)^{\ceil{\frac{3N(N-1)}{2N+1}\ln\Bigg(\frac{\sum\limits_{d=1}^D \mathcal{L}_d^0}{N\varepsilon^2}\Bigg)}}$.
    
    We then bound $\epsilon_D^0$ using the same strategy as for $\epsilon^0$ in \cref{al: bound epsilon 0} to get:
    \begin{equation}
        \epsilon_D^0 \le \frac{3N(N-1)}{2N+1}.
    \end{equation}
$\hfill\square$

\section{Constrained 2-dimensional case}

\subsection{Proof of \cref{th:circle to half disk finite expec time optimised}}
\label{sm: proof circle to half disk finite expec time optimised}
We introduce an auxiliary quantity which simplifies the proof, and that we will optimise at the end.

\begin{definition}
    For any $\delta\in(0,1)$, we define $\theta_\delta\in(0,\frac{\pi}{2})$ as $\theta_\delta = \frac{\pi}{2}-\arccos\delta$.
\end{definition}

\begin{lemma}
    \label{lem:circle to half disk finite expec time}
    Any system evolving according to \cref{eq:circle evolution} becomes in finite expected time within a half-disk. In particular, for any $\delta\in(0,1)$:
    $$ \mathbb{E}(T_{\textit{HD}}\mid \Theta_0) \le \frac{1}{\Big(\big(1-2\frac{\theta_\delta}{\pi}\big)\big(\frac{\theta_\delta}{\pi}\cdot\frac{2}{N(N-1)}\big)^2\Big)^{\floor{\frac{N}{2}}}} + 2\floor[\Big]{\frac{N}{2}}.$$
\end{lemma}

\begin{proof}
	The proof will consist in finding a simple sequence of events with lower bounded probability by a strictly positive constant to evolve from any configuration that is not contained in a half-disk to a configuration contained in one.

    Assume that the system at time step $k$ is not contained within a half-disk.
    
    Consider and choose a $\theta_{\textit{HD}}^k$ such that:
    \begin{equation}
        \theta_{\textit{HD}}^k\in\argmax\limits_{\theta}\left|\{\theta_i^k\mid\cos(\theta_i^k-\theta)> 0\}\right|.
    \end{equation}
    
    The maximal number of agents in the ``right'' open disk is denoted $N_{\textit{HD},\textrm{max}}^k$:
    \begin{equation}
        N_{\textit{HD},\textrm{max}}^k = \max\limits_{\theta} \left|\{\theta_i^k\mid\cos(\theta_i^k-\theta)> 0\}\right| = \left|\{\theta_i^k\mid\cos(\theta_i^k-\theta_{\textit{HD}}^k)> 0\}\right|.
    \end{equation}
    
    Denote:
    \begin{align}
        \mathcal{R}^k &= \{\theta\mid\cos(\theta-\theta_{\textit{HD}}^k)> 0\}\\
        \mathcal{L}^k &= \{\theta\mid\cos(\theta-\theta_{\textit{HD}}^k)\le 0\}
    \end{align}
    the set of possible opinions in the open ``right'' half-disk and in the closed ``left'' half-disk.
    
    For simplicity, we perform a change of angular parametrization around $\theta_{\textit{HD}}^k$:
    \begin{align}
        \tilde{\theta}_i^k &= \theta_i^k - \theta_{\textit{HD}}^k \mod 2\pi\\
        \tilde{\Theta}_k &= (\tilde{\theta}_1^k,\cdots,\tilde{\theta}_N^k)^T\\
       \tilde{\mathcal{R}}^k &= \{\tilde{\theta}\mid\cos(\tilde{\theta})> 0\}\\
        \tilde{\mathcal{L}}^k &= \{\tilde{\theta}\mid\cos(\tilde{\theta})\le 0\}.
    \end{align}
    
    Since we have assumed that not all agents are contained within a half-disk, then:
    \begin{equation}
        \label{eq:circle L not empty}
        \tilde{\mathcal{L}}^k\cap\tilde{\Theta}_k\neq \emptyset.
    \end{equation}
    
    There are two cases for configurations of agents in $\tilde{\mathcal{R}}^k$:
    \begin{enumerate}
        \item There is $r\in\{1,\cdots,N\}$ with $\tilde{\theta}_r^k\in\tilde{\mathcal{R}}^k$ such that $\cos{\tilde{\theta}_r^k} > \delta$.
        \item For all $r\in\{1,\cdots,N\}$ with $\tilde{\theta}_r^k\in\tilde{\mathcal{R}}^k$, $0<\cos{\tilde{\theta}_r^k} \le \delta$.
    \end{enumerate}
    
    We also partition $\tilde{\mathcal{R}}^k$ into three regions:
    \begin{align}
        \tilde{\mathcal{R}}_1^k &= \left[\frac{\pi}{2}-\theta_\delta,\frac{\pi}{2}\right) \\
        \tilde{\mathcal{R}}_2^k &= \left(\frac{3\pi}{2},\frac{3\pi}{2}+\theta_\delta\right] \\
        \tilde{\mathcal{R}}_3^k &= \left(\frac{3\pi}{2},2\pi\right)\cup \left[0,\frac{\pi}{2}-\theta_\delta\right).
    \end{align}

    For an illustration on how we split the circle and of the different cases we analyse, please see \cref{fig: ND circle cases partition}.
    
    \begin{figure}[tbhp]
      \centering
        \includegraphics[width=0.9\textwidth]{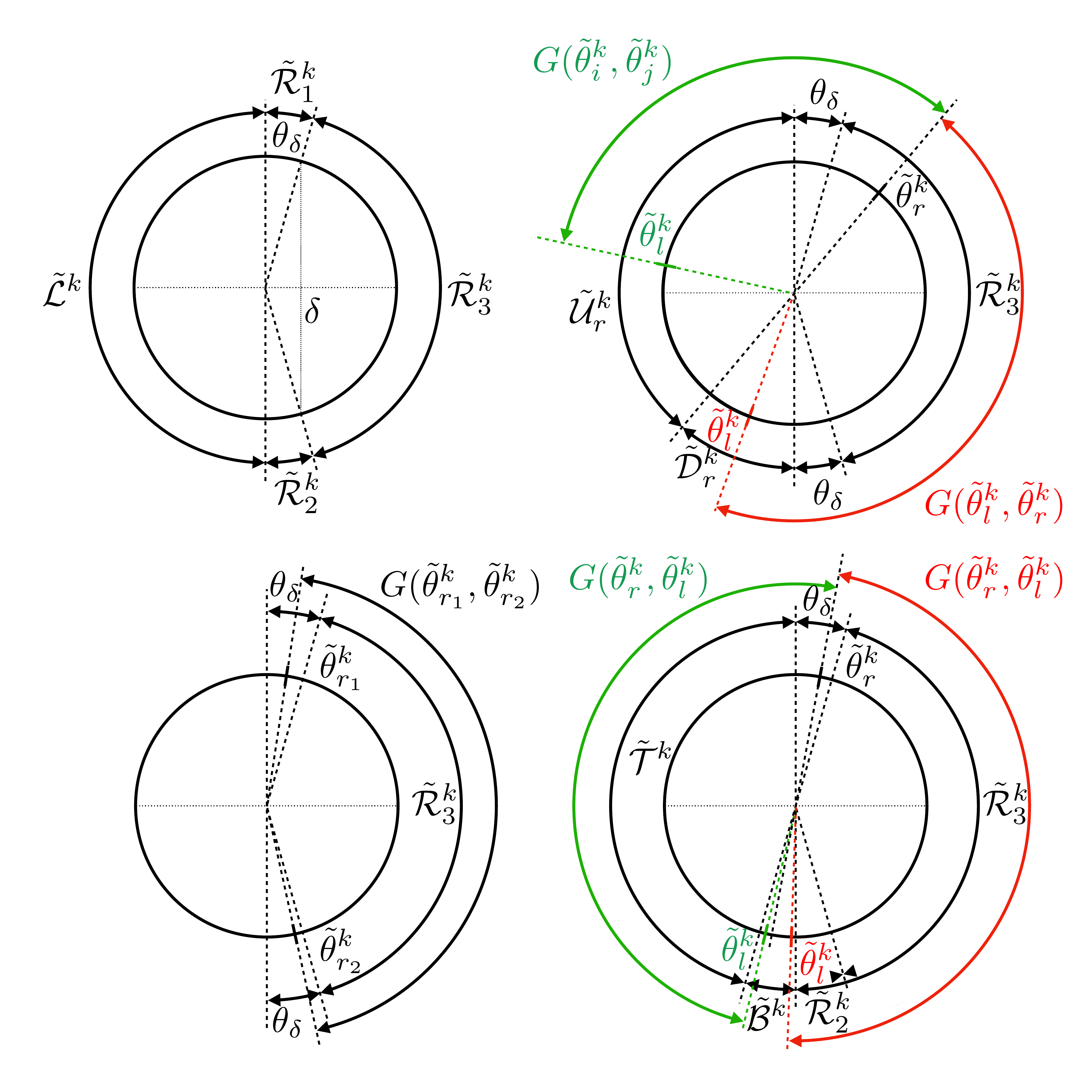}
        \caption{How we split the circle in all major cases of the proof. Top left: After reparametrization adapted to the half-disk with maximal number of opinions, the ``right'' side of the half-disk is split into three arcs: the two poles, the size of which is given by $\theta_\delta$, and the rest. Top right: Example of \hyperref[par:circle case 1.1]{Case 1.1} in green and of \hyperref[par:circle case 1.2]{Case 1.2} in red. Bottom left: Example of \hyperref[par:circle case 2.1]{Case 2.1}. Bottom right: Example of \hyperref[par:circle case 2.2.1]{Case 2.2.1} in red and of \hyperref[par:circle case 2.2.2]{Case 2.2.2} in green.}
        \label{fig: ND circle cases partition}
    \end{figure}

    \paragraph{Case 1}
    \label{par:circle case 1}
    Assume $\tilde{\mathcal{R}}_3^k \cap \tilde{\Theta}_k\neq \emptyset$ and choose $r$ such that:
    \begin{equation}
        \tilde{\theta}_r^k \in \tilde{\mathcal{R}}_3.
    \end{equation}
    
    Using \cref{eq:circle L not empty}, we can choose $l$ such that:
    \begin{equation}
        \tilde{\theta}_l^k\in\tilde{\mathcal{L}}^k.
    \end{equation}

    Define $\tilde{\mathcal{U}}_r^k$ and $\tilde{\mathcal{D}}_r^k$, the ``up'' and ``down'' sets of angles in $\mathcal{L}^k$ associated to $\tilde{\theta}_r$:
    \begin{align}
        \tilde{\mathcal{U}}_r^k &=
        \begin{cases}
            \left[\frac{\pi}{2},\theta_r^k+\pi\right) &\text{if } \tilde{\theta}_r^k \in \left[0,\frac{\pi}{2}\right)\\
            \left[\frac{\pi}{2},\theta_r^k-\pi\right) &\text{if } \tilde{\theta}_r^k \in \left(\frac{3\pi}{2},2\pi\right)\\
        \end{cases} \\
        \tilde{\mathcal{D}}_r^k &=
        \begin{cases}
            \left(\theta_r^k+\pi,\frac{3\pi}{2}\right] &\text{if } \tilde{\theta}_r^k \in \left[0,\frac{\pi}{2}\right)\\
            \left(\theta_r^k-\pi,\frac{3\pi}{2}\right] &\text{if } \tilde{\theta}_r^k \in \left(\frac{3\pi}{2},2\pi\right).
        \end{cases}
    \end{align}
    
    \paragraph{Case 1.1}
    \label{par:circle case 1.1}
    Assume $\tilde{\theta}_l^k\in \tilde{\mathcal{U}}_r^k$. Conditionally to choosing agents $l$ and $r$ for update, we have:
    \begin{align}
        \mathbb{P}_{l,r}(\tilde{\theta}_l^{k+1}\in \tilde{\mathcal{R}}_1^k \mid \Theta_k) &= \frac{\theta_\delta}{\widehat{(\tilde{\theta}_l^k,\tilde{\theta}_r^k})} \ge \frac{\theta_\delta}{\pi} \\
        \mathbb{P}_{l,r}(\tilde{\theta}_r^{k+1}\in \tilde{\mathcal{R}}^k \mid \Theta_k) &= \frac{\frac{\pi}{2}-\tilde{\theta}_r^k}{\widehat{(\tilde{\theta}_l^k,\tilde{\theta}_r^k})} \ge \frac{\theta_\delta}{\pi}.
    \end{align}
    
    By independence of the updates, and removing of the conditioning on the choice of pair for update:
    \begin{equation}
        \mathbb{P}(\tilde{\theta}_l^{k+1}\in\tilde{\mathcal{R}}_1^k \wedge \tilde{\theta}_r^{k+1}\in\tilde{\mathcal{R}}\mid \Theta_k) \ge \left(\frac{\theta_\delta}{\pi}\right)^2 \frac{2}{N(N-1)}.
    \end{equation}
    
    Since $\tilde{\mathcal{R}}_1^k\subset \tilde{\mathcal{R}}^k$, the event $\tilde{\theta}_l^{k+1}\in\tilde{\mathcal{R}}_1^k \wedge \theta_r^{k+1}\in\tilde{\mathcal{R}}$ would increase by one the number of agents in $\tilde{\mathcal{R}}^k$, and thus increase $N_{\textit{HD},\textrm{max}}^k$ by one:
    \begin{equation}
        \label{eq: circle increase N Case 1}
        \mathbb{P}(N_{\textit{HD},\textrm{max}}^{k+1}\ge N_{\textit{HD},\textrm{max}}^k +1\mid \Theta_k) \ge \left(\frac{\theta_\delta}{\pi}\right)^2 \frac{2}{N(N-1)}.
    \end{equation}

    \paragraph{Case 1.2}
    \label{par:circle case 1.2}
    Assume $\tilde{\theta}_l^k\in\tilde{\mathcal{D}}_r^k$. A similar analysis, where we replace the set $\tilde{\mathcal{R}}_1^k$ by $\tilde{\mathcal{R}}_2^k\subset G(\tilde{\theta}_l^k,\tilde{\theta}_r^k)$, gives the same result \cref{eq: circle increase N Case 1}.
    
    \paragraph{Case 1.3} Assume that $\tilde{\theta}_l^k = \tilde{\theta}_r^k\pm\pi$. This event, if it was not already the case initially in $\tilde{\Theta}_0$, never happens almost surely. Nevertheless, if this situation occurs, the choice of geodesic will lead to a problem analogous to one of the previous cases and we get the same bounds as previously derived.
    
    \paragraph{Case 2} Assume $\tilde{\mathcal{R}}_3^k=\emptyset$. We can further split the distribution of agents in $\tilde{\mathcal{R}}^k$ into two cases: depending on whether or not one of the sets $\tilde{\mathcal{R}}_1^k$ or $\tilde{\mathcal{R}}_2^k$ is empty.
    
    \paragraph{Case 2.1}
    \label{par:circle case 2.1}
    Assume $\tilde{\mathcal{R}}_1^k\cap\tilde{\Theta}_k\neq\emptyset$ and $\tilde{\mathcal{R}}_2^k\cap\tilde{\Theta}_k\neq\emptyset$. There exists then agents $r_1$ and $r_2$ such that:
    \begin{equation}
        (\tilde{\theta}_{r_1}^k,\tilde{\theta}_{r_2}^k) \in \tilde{\mathcal{R}}_1^k\times \tilde{\mathcal{R}}_2^k.
    \end{equation}
    
    The sketch of the proof is as follows: should we choose for update agents $r_1$ and $r_2$, then both agents stay in $\tilde{\mathcal{R}}^k$ and there is a significant chance that at least one agent falls into $\tilde{\mathcal{R}}_3^k$. Should this happen, then step $k+1$ would satisfy \hyperref[par:circle case 1]{Case 1}. More formally, we have:
    \begin{equation}
        \mathbb{P}_{r_1,r_2}(\tilde{\theta}_{r_1}^{k+1}\in\tilde{\mathcal{R}}_3^k\mid\Theta_k) = \mathbb{P}_{r_1,r_2}(\tilde{\theta}_{r_2}^{k+1}\in\tilde{\mathcal{R}}_3^k\mid\Theta_k) = \frac{\pi-2\theta_\delta}{\widehat{(\tilde{\theta}_{r_1},\tilde{\theta}_{r_2})}} \ge \frac{\pi - 2\theta_\delta}{\pi}.
    \end{equation}
    
    We can then write, since the probability of the union of the two events is larger than the minimum probability of those events, and by removing the conditioning:
    \begin{equation}
        \mathbb{P}\Big(\tilde{\theta}_{r_1}^{k+1}\in\tilde{\mathcal{R}}_3^k \lor \tilde{\theta}_{r_2}^{k+1}\in\tilde{\mathcal{R}}_3^k \mid \Theta_k\Big) \ge \left(1-\frac{ 2\theta_\delta}{\pi}\right)\frac{2}{N(N-1)}.
    \end{equation}
    
    Such an event leads into a configuration satisfying \hyperref[par:circle case 1]{Case 1} with $N_{\textit{HD},\textrm{max}}^{k+1}\ge N_{\textit{HD},\textrm{max}}^k$\footnote{We do not necessarily have equality as for instance $\tilde{\theta}_{r_1}^k$ could be updated with an agent in $\tilde{\mathcal{L}}^k$ and that agent could have its new opinion in $\tilde{\mathcal{R}}_3^k$.}. Using \cref{eq: circle increase N Case 1} at step $k+1$, and the independence of the random variables between steps, we get:
    \begin{equation}
        \mathbb{P}(N_{\textit{HD},\textrm{max}}^{k+2} \ge N_{\textit{HD},\textrm{max}}^k +1\mid \Theta_k) \ge \left(1-\frac{2\theta_\delta}{\pi}\right)\left(\frac{\theta_\delta}{\pi}\right)^2\left(\frac{2}{N(N-1)}\right)^2.
    \end{equation}

    \paragraph{Case 2.2}
    Assume now that $\tilde{\mathcal{R}}_2^k\cap\tilde{\Theta}_k = \emptyset$. The case when $\tilde{\mathcal{R}}_1^k\cap\tilde{\Theta}_k = \emptyset$ is solved in the same way and we will not cover it. Since $\tilde{\mathcal{R}}^k$ cannot be empty by maximality, then $\tilde{\mathcal{R}}_1^k\cap\tilde{\Theta}_k \neq \emptyset$. There exists an agent $r$ such that:
    \begin{equation}
        \tilde{\theta}_r^k\in\tilde{\mathcal{R}}_1^k.
    \end{equation}

    Denote $\tilde{\mathcal{T}}^k$ and $\tilde{\mathcal{B}}^k$ the ``top'' and ``bottom'' sets partitioning $\tilde{\mathcal{L}}^k$:
    \begin{align}
        \tilde{\mathcal{T}}^k &= \left[\frac{\pi}{2},\theta_\delta+\pi\right) &\text{if } \\
        \tilde{\mathcal{B}}^k  &= \left[\theta_\delta+\pi,\frac{3\pi}{2}\right].
    \end{align}
    
    If we can find $\tilde{\theta}_l\in\tilde{\mathcal{T}}^k$, then we can define a new half-disk angle such that its ``right'' open half-disk contains the region $\tilde{\mathcal{R}}_1^k$ and $\tilde{\theta}_l$ and thus has at least one more agent than in $\tilde{\mathcal{R}}^k$. This would violate the maximality of $N_{\textit{HD},\textrm{max}}^k$. Therefore:
    \begin{equation}
        \tilde{\mathcal{T}}^k \cap \tilde{\Theta}_k = \emptyset.
    \end{equation}
    
    Then, using \cref{eq:circle L not empty}, there exists an agent $l$ such that:
    \begin{equation}
        \tilde{\theta}_l^k\in\tilde{\mathcal{B}}^k.
    \end{equation}
    
    As we are considering geodesics for updates, we have further cases depending on the positioning of agent $l$ with respect to agent $r$.
    
    \paragraph{Case 2.2.1}
    \label{par:circle case 2.2.1}
    Assume $\tilde{\theta}_l^k - \tilde{\theta}_r^k > \pi$. Should we then choose agents $l$ and $r$ for update, we then have a significant probability to bring agent $l$ in $\tilde{\mathcal{R}}^k$ while maintaining agent $r$ in $\tilde{\mathcal{R}}^k$ in one step. We have:
    \begin{equation}
        \mathbb{P}_{l,r}(\tilde{\theta}_l^{k+1}\in\tilde{\mathcal{R}}^k\mid \Theta_k) = \mathbb{P}_{l,r}(\tilde{\theta}_r^{k+1}\in\tilde{\mathcal{R}}^k\mid \Theta_k) \ge \frac{\theta_\delta + (\pi- 2\theta_\delta)}{\widehat{(\tilde{\theta}_l^k,\tilde{\theta}_r^k)}} \ge 1-\frac{\theta_\delta}{\pi}.
    \end{equation}
    
    This event leads to an increase of the the number of agents in $\tilde{\mathcal{R}}^k$ and thus to an increase of $N_{\textit{HD},\textrm{max}}^k$. Thus, by independence of the updates and by removing the conditioning:
    \begin{equation}
        \mathbb{P}(N_{\textit{HD},\textrm{max}}^{k+1} \ge N_{\textit{HD},\textrm{max}}^k +1\mid \Theta_k) \ge \left(1-\frac{\theta_\delta}{\pi}\right)^2\frac{2}{N(N-1)}.
    \end{equation}
    
    \paragraph{Case 2.2.2}
    \label{par:circle case 2.2.2}
    Assume now that $\tilde{\theta}_l^k - \tilde{\theta}_r^k < \pi$. In this case it is not very likely to update $\tilde{\theta}_l^k$ into $\tilde{\mathcal
    R}^k$ conditionally to choosing agents $l$ and $r$. However, it is very likely conditionally to this choice that they both fall in $\tilde{\mathcal{T}}^k$. Indeed:
    \begin{equation}
        \mathbb{P}(\tilde{\theta}_l^{k+1}\in \tilde{\mathcal{T}}^k\mid \Theta_k) = \mathbb{P}(\tilde{\theta}_r^{k+1}\in \tilde{\mathcal{T}}^k\mid \Theta_k) \ge \frac{\pi-\theta_\delta}{\widehat{(\tilde{\theta}_l^k,\tilde{\theta}_r^k)}} \ge 1-\frac{\theta_\delta}{\pi}.
    \end{equation}
    
    Should this event occur, we can then consider a new ``right'' open half-disk containing $\tilde{\mathcal{R}}_1^k$ and the updates of the two agents, which would have one more agent, leading to an increase in $N_{\textit{HD},\textrm{max}}^k$. Therefore, by independence of the updates and by removing the conditioning:
    \begin{equation}
        \mathbb{P}(N_{\textit{HD},\textrm{max}}^{k+1} \ge N_{\textit{HD},\textrm{max}}^k +1\mid \Theta_k) \ge \left(1-\frac{\theta_\delta}{\pi}\right)^2\frac{2}{N(N-1)}.
    \end{equation}
    
    \paragraph{Case 2.2.3} Assume $\tilde{\theta}_l^k - \tilde{\theta}_r^k = \pi$. This event, if it was not already the case initially in $\tilde{\Theta}_0$, never happens almost surely. Nevertheless, if this situation occurs, the choice of geodesic will lead to a problem analogous to one of the previous cases and we get the same bounds as previously derived.
    
    We have explored all possibilities and can summarise the results as follows, denoting by abuse of notation $\mathcal{C}_{c}^k$ the event $\Theta_k$ satisfies case $c$ as mentioned previously:
    \begin{align}
        \label{eq:circle summary increase N Case 1}
        \mathbb{P}(N_{\textit{HD},\textrm{max}}^{k+1}\ge N_{\textit{HD},\textrm{max}}^k +1\mid N_{\textit{HD},\textrm{max}}^k \wedge \mathcal{C}_1^k) &\ge \eta_{1,\delta} \\
        \label{eq:circle summary increase N Case 2.1}
        \mathbb{P}(N_{\textit{HD},\textrm{max}}^{k+2} \ge N_{\textit{HD},\textrm{max}}^k +1\mid N_{\textit{HD},\textrm{max}}^k \wedge \mathcal{C}_{2.1}^k) &\ge \eta_{2.1,\delta} \\
        \label{eq:circle summary increase N Case 2.2}
        \mathbb{P}(N_{\textit{HD},\textrm{max}}^{k+1} \ge N_{\textit{HD},\textrm{max}}^k +1\mid N_{\textit{HD},\textrm{max}}^k \wedge \mathcal{C}_{2.2}^k) &\ge \eta_{2.2,\delta},
    \end{align}
    where $\eta_c$ are the derived bounds for case $c$:
    \begin{align}
        \eta_{1,\delta} &= \left(\frac{\theta_\delta}{\pi}\right)^2 \frac{2}{N(N-1)} \\
        \eta_{2.1,\delta} &= \left(1-\frac{2\theta_\delta}{\pi}\right)\left(\frac{\theta_\delta}{\pi}\right)^2\left(\frac{2}{N(N-1)}\right)^2 \\
        \eta_{2.2,\delta} &= \left(1-\frac{\theta_\delta}{\pi}\right)^2\frac{2}{N(N-1)}.
    \end{align}
    
    Since $
        0<1-\frac{2\theta_\delta}{\pi}<\left(1-\frac{\theta_\delta}{\pi}\right)^2<1$, $\frac{\theta_\delta}{\pi} < 1$, and $\frac{2}{N(N-1)} < 1$, we have that the smallest of the three bounds is:
    \begin{equation}
        \eta_{2.1,\delta} = \min\{\eta_{1,\delta},\eta_{2.1,\delta},\eta_{2.2,\delta}\} >0.
    \end{equation}
    
    We now wish to consider, starting from any possible configuration of angles $\Theta_k$, a series of consecutive events that will lead to all agents within a half-disk. We will call these consecutive events ``consecutive increases'', where a ``consecutive increase'' is an increase of at least $1$ of $N_{\textit{HD},\textrm{max}}^k$ at time step $k+1$ or $k+2$. Formally, a series of ``consecutive increases'' of length $l$, is a realisation of a series of events and a series of time steps $k_i$ with $k<k_1<\cdots<k_l$ such that, for all $i\in\{1,\cdots,l\}$:
    \begin{equation}
        \begin{cases}
            k_{i+1} = k_i+1 \text{ or } k_{i+1} = k_i + 2 \\
            N_{\textit{HD},\textrm{max}}^{k_{i+1}} \ge N_{\textit{HD},\textrm{max}}^{k_{i}} +1
        \end{cases}
    \end{equation}
    
    By maximality of $N_{\textit{HD},\textrm{max}}^k$, for any configuration $\Theta_k$, we have:
    \begin{equation}
        N_{\textit{HD},\textrm{max}}^k \ge \ceil[\Big]{\frac{N}{2}}.
    \end{equation}
    
    Therefore, from any configuration of opinions $\Theta_k$, it suffices to make a series of ``consecutive increases'' of length at most $\floor{\frac{N}{2}}$ to fall into a configuration where all opinions are within a half-disk. Since the half-disk configuration is stable, as proven in \cref{prop: circle half disk stable}, if we perform a series of ``consecutive increases'' of length $l$ with last time step $k_l$, then for any future time steps $k'\ge k_l$, all agents are still in a half-disk configuration. For the series of ``consecutive increases'' we have been considering, then $k_l$ is bounded by at most $k+2l$, which itself is smaller than $k+2\floor{\frac{N}{2}}$. We have thus shown that from any configuration of agents $\Theta_k$ not contained within a half-disk, we have that, if by abuse of notation we write that $\Theta^m\in \textit{HD}^m$ the event that all agents at time step $m$ are within a half-disk: 
   \begin{equation}
       \mathbb{P}(\Theta^{k+2\floor{\frac{N}{2}}}\in\textit{HD}^{k+2\floor{\frac{N}{2}}}\mid \Theta_k) \ge (\eta_{2.1,\delta})^{\floor{\frac{N}{2}}} >0.
   \end{equation}
   
   Recall the well known results from probability theory, that in a series of independent identically distributed trials, for an event $Q$ whose occurrence at each trial is independent, with probability $q$, then the expectation of the time of first occurrence of $Q$, denoted $k_{Q,1}$, is:
   \begin{equation}
        \label{eq: expect first occurrence Q iid}
       \mathbb{E}(k_{Q,1}) = q+2(1-q)q + 3(1-q)^2q+\cdots = \frac{1}{q}.
   \end{equation}
   
   We will consider at each time step $k$ the event $Q = \Theta^{k+2\floor{\frac{N}{2}}}\in\textit{HD}^{k+2\floor{\frac{N}{2}}}$. While the events $Q$ at each time step are not independent and not of the same probability, we can safely use the previous result \cref{eq: expect first occurrence Q iid} for an upper bound on the expected first occurrence of $Q$. Thus up to at most $2\floor{\frac{N}{2}}$, we get an upper bound on the expectation of $T_{HD}$:
    \begin{equation}
        \mathbb{E}(k_{\Theta^{k}\in\textit{HD}^{k+2\floor{\frac{N}{2}}},1}\mid \Theta_0) \le \frac{1}{(\eta_{2.1,\delta})^{\floor{\frac{N}{2}}}}+2\floor[\Big]{\frac{N}{2}}.
    \end{equation}
   
    Finally, using the definition of $\eta_{2.1,\delta}$:
    \begin{equation}
        \mathbb{E}(T_{\textit{HD}}\mid \Theta_0) \le \left(\frac{1}{\left(1-\frac{2\theta_\delta}{\pi}\right)\left(\frac{\theta_\delta}{\pi}\right)^2\left(\frac{2}{N(N-1)}\right)^2}\right)^{\floor{\frac{N}{2}}} + 2\floor[\Big]{\frac{N}{2}} <\infty
    \end{equation}
\end{proof}

We can now finish the proof using the following theorem:

\begin{theorem}
    \label{th: circle final bounds optimised expec hd}
    The bound in \cref{lem:circle to half disk finite expec time} is optimised when $\theta_\delta = \frac{\pi}{3}$, i.e. $\delta = \frac{\sqrt{3}}{2}$, which yields:
    \begin{align*}
        \mathbb{E}(T_{\textit{HD}}\mid \Theta_0) &\le \Big(\frac{
        27}{4}N^2(N-1)^2\Big)^{\floor{\frac{N}{2}}} + 2\floor[\Big]{\frac{N}{2}}.
    \end{align*}
\end{theorem}

\begin{proof}
    The bound is minimised when $f(x) = (1-2x)x^2$ is maximised using $x=\frac{\theta_\delta}{\pi}\in\left(0,\frac{1}{2}\right)$. This is a polynomial of degree $3$. A simple analysis of the sign of its derivative $f'(x) = -6x\left(x-\frac{1}{3}\right)$ gives that $f(x)$ is maximal in $\left(0,\frac{1}{2}\right)$ when $x = \frac{1}{3}$, for which $f\left(\frac{1}{3}\right) = \frac{1}{27}$. This corresponds to $\theta_\delta = \frac{\pi}{3}$ and thus to $\delta = \frac{\sqrt{3}}{2}$.
\end{proof}

\section{Open problems on the constrained 2-dimensional case}
\label{sm: Open problems on the constrained 2-dimensional case}

\subsection{Vector averaging} 

\begin{definition}
    Let $\mathcal{S}^k\in\mathbb{R}^2$ be the sum of all opinions at time step $k$:
    $$ \mathcal{S}^k = \sum\limits_{i=1}^N x_i^k.$$
\end{definition}

Note that $\mathcal{S}^k$ does not necessarily belongs to the opinion space $S$ but is simply a vector in the plane.

\begin{proposition}
    For all time steps $k$, we have: $\left\lVert \mathcal{S}^k\right\rVert_2 \le N$.
\end{proposition}

\begin{proof}
    This is due to the triangular inequality and since the opinion space $S$ is the unit circle.
\end{proof}

The purpose of studying $\mathcal{S}^k$ is that convergence of opinions in $S$ is equivalent to convergence of $\mathcal{S}^k$ in $\mathbb{R}^2$. Intuitively and experimentally, if $\left\lVert \mathcal{S}^k\right\rVert_2$ is ``large'', then there is a ``large'' number of opinions positively oriented with $\mathcal{S}^k$, and furthermore opinions positively oriented with $\mathcal{S}^k$ tend to be updated in a way that further increases the norm of $\mathcal{S}^k$. However $\lVert\mathcal{S}^k\rVert_2$ is upper-bounded by $N$ which can only happen for opinions arbitrarily close to each other. Therefore we can simply study the evolution of $\left\lVert \mathcal{S}^k\right\rVert_2^2$, which is an upper-bounded random real quantity and show that it converges to the upper-bound and study the speed of convergence.

\begin{definition}
    For time step $k$ and agents $i$ and $j$, denote respectively $\theta_{\mathcal{S}^k}$ and $x_{\mathcal{S}^k}$ the oriented angle of $\mathcal{S}^k$ with respect to the positive $x$-axis and the unit vector in the direction of $\mathcal{S}^k$.
\end{definition}

\begin{definition}
    For time step $k$ and agents $i$ and $j$, denote for conciseness $\alpha_{i,j}^k  = \widehat{(x_i^k,x_j^k)}$ as half the geometric angle along the geodesic between opinions $x_i^k$ and $x_j^k$.
\end{definition}

\begin{definition}
    For time step $k$ and agents $i$ and $j$, denote $x_{\textit{bis},i,j}^k$ the unit vector along the bisector of the opinions $x_i^k$ and $x_j^k$ on the geodesic side $G(\theta_i^k,\theta_j^k)$.
\end{definition}

\begin{definition}
    For time step $k$ and agents $i$ and $j$, if $\mathcal{S}^k\neq 0$, denote $\beta_{i,j}^k$ the oriented angle of $x_{\textit{bis},i,j}^k$ with respect to $\mathcal{S}^k$.
\end{definition}

For an example of these definitions, please see \cref{fig: ND circle S}.

\begin{figure}[tbhp]
      \centering
        \includegraphics[width=0.4\textwidth]{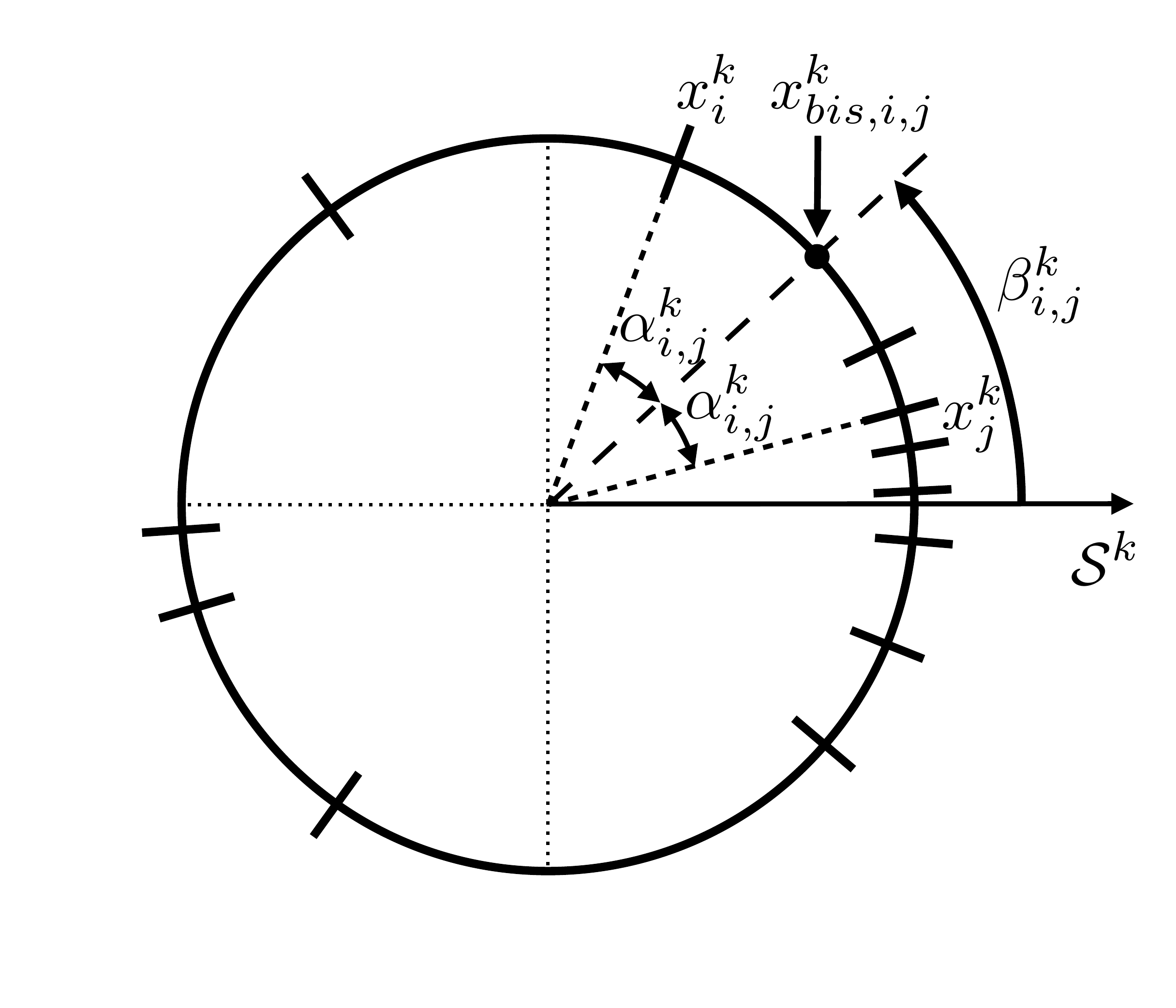}
        \caption{Example of configuration of opinions for the vector averaging approach. The reference direction at each step is updated to be $\mathcal{S}^k$, the current sum of all opinions.}
        \label{fig: ND circle S}
    \end{figure}

\begin{definition}
    For time step $k+1$, denote $\Delta^{k+1} = \mathcal{S}^{k+1} - \mathcal{S}^k$ the random difference between the sum vectors, and conditionally for the choice of opinions $i$ and $j$ to be updated as $\Delta_{i,j}^{k+1} = x_i^{k+1} + x_j^{k+1} - x_i^k - x_j^k$.
\end{definition}

\begin{proposition}
    \label{prop: circle expec xi}
    We have for all $k$, $i$, and $j$:
    $$\mathbb{E}_{i,j}(x_i^{k+1}\mid X_k) = \mathbb{E}_{i,j}(x_j^{k+1}\mid X_k) = \frac{\sin{\alpha_{i,j}^k}}{\alpha_{i,j}^k} x_{\textit{bis},i,j}^k.$$
\end{proposition}

\begin{proof}
    Due to the rules of motion \cref{eq:circle evolution}:
    \begin{equation}
        \mathbb{E}_{i,j}(x_i^{k+1}\mid X^k) = \mathbb{E}_{i,j}(x_j^{k+1}\mid X_k).
    \end{equation}
    
    Up to a rotation of $\beta_{i,j}^k$, we can assume that $x_{\textit{bis},i,j}^k = (1,0)^T$. In this case:
    \begin{align}
        \mathbb{E}(x_i^{k+1}\mid X_k) &= \int\limits_{-\alpha_{i,j}^k}^{\alpha_{i,j}^k} (\cos{\theta},\sin{\theta})^T \frac{d\theta}{2\alpha_{i,j}^k} \nonumber\\
        &= \frac{\sin{\alpha_{i,j}^k}}{\alpha_{i,j}^k} (1,0)^T.
    \end{align}
    
    By removing the rotation of $\beta_{i,j}^k$, we have in the general case:
    \begin{equation}
        \mathbb{E}_{i,j}(x_i^{k+1}\mid X_k) = \frac{\sin{\alpha_{i,j}^k}}{\alpha_{i,j}^k} x_{\textit{bis},i,j}^k.
    \end{equation}
\end{proof}

\begin{proposition}
    \label{prop: circle scal expec delta n+1 Sn no simplification}
    We have for all $k$, $i$, and $j$:
    $$ \left\langle\mathbb{E}(\Delta^{k+1}\mid X^k),\mathcal{S}^k\right\rangle = \frac{4}{N(N-1)}\sum\limits_{i<j}\left(\frac{\sin{\alpha_{i,j}^k}}{\alpha_{i,j}^k}-\cos{\alpha_{i,j}^k}\right) \cos{\beta_{i,j}^k} \left\lVert \mathcal{S}^k \right\rVert_2. $$
\end{proposition}

\begin{proof}
    We have, by conditioning on the choice of agents for update, and by \cref{prop: circle expec xi}:
    \begin{align}
        \mathbb{E}(\Delta^{k+1}\mid X^k) &= \frac{1}{N(N-1)}\sum\limits_{i\neq j} \mathbb{E}(\Delta_{i,j}^{k+1}\mid X_k) \nonumber \\
        &= \frac{1}{N(N-1)}\sum\limits_{i\neq j} \mathbb{E}(x_i^{k+1}+x_j^{k+1}\mid X_k) - x_i^k -x_j^k \nonumber \\
        &= \frac{2}{N(N-1)}\sum\limits_{i< j} 2\frac{\sin{\alpha_{i,j}^k}}{\alpha_{i,j}^k}x_{\textit{bis},i,j}^k - x_i^k -x_j^k .
    \end{align}
    
    We also have:
    \begin{align}
        \left\langle x_i^k + x_j^k,\mathcal{S}^k\right\rangle &= \left(\cos(\beta_{i,j}^k+\alpha_{i,j}^k) + \cos(\beta_{i,j}^k-\alpha_{i,j}^k)\right)\left\lVert \mathcal{S}^k \right\rVert_2 \nonumber\\
        &= 2\cos{\alpha_{i,j}^k}\cos{\beta_{i,j}^k} \left\lVert \mathcal{S}^k \right\rVert_2.
    \end{align}
    
    If we assume that $x_i^k$ and $x_j^k$ are not diametrically opposite, which almost surely never happens except for the cases where in the original setting $X^0$ they are already in that configuration, then:
    \begin{equation}
        x_{\textit{bis},i,j}^k = \frac{x_i^k+x_j^k}{\left\lVert x_i^k+x_j^k \right\rVert_2} = \frac{x_i^k+x_j^k}{2\cos{\alpha_{i,j}^k}}.
    \end{equation}
    
    We then have that:
    \begin{equation}
        \left\langle\mathbb{E}(\Delta^{k+1}\mid X_k),\mathcal{S}^k\right\rangle = \frac{4}{N(N-1)}\sum\limits_{i<j}\left(\frac{\sin{\alpha_{i,j}^k}}{\alpha_{i,j}^k}-\cos{\alpha_{i,j}^k}\right) \cos{\beta_{i,j}^k} \left\lVert \mathcal{S}^k \right\rVert_2.
    \end{equation}
    
    Note that the result still holds in the exceptional almost surely never happening case where the opinions are diametrically opposite, although depending on the choice of geodesic the term inside the summation can be multiplied by minus $1$.
\end{proof}

\begin{proposition}
    \label{prop: circle norm S sum cos}
    We have for all $k$, $i$, and $j$:
    $$ \left\lVert \mathcal{S}^k \right\rVert_2 = \sum\limits_{i=1}^N \cos(\theta_i^k-\theta_{\mathcal{S}^k}).$$
\end{proposition}

\begin{proof}
    This directly comes from the fact that the norm of a sum of vectors is the sum of the projections of each vector onto the unit vector of the sum, $x_{\mathcal{S}^k}$ in our case.
\end{proof}

\begin{proposition}
    \label{prop: circle sum cos alpha cos beta}
    We have for all $k$, $i$, and $j$:
    $$ \sum\limits_{i<j} \cos{\alpha_{i,j}^k}\cos{\beta_{i,j}^k} = \frac{N-1}{2}\left\lVert \mathcal{S}^k \right\rVert_2.$$
\end{proposition}

\begin{proof}
    This is a consequence of \cref{prop: circle norm S sum cos}. Indeed, for all agents $i<j$, then we have the set equality:
    \begin{equation}
        \left\{\theta_i^k,\theta_j^k\right\} = \left\{\beta_{i,j}^k+\alpha_{i,j}^k,\beta_{i,j}^k-\alpha_{i,j}^k\right\}.
    \end{equation} 
    
    Thus if we sum over all possible sets, then for all agents $i$, the opinion $\theta_i$ appears in exactly $N-1$ terms of the sum:
    \begin{align}
        \sum\limits_{i<j}\cos(\beta_{i,j}^k+\alpha_{i,j}^k)+\cos(\beta_{i,j}^k-\alpha_{i,j}^k) &= \sum\limits_{i<j}\cos{(\theta_i^k-\theta_{\mathcal{S}^k})}+\cos{(\theta_j^k-\theta_{\mathcal{S}^k})} \\
        &= (N-1)\left\lVert \mathcal{S}^k \right\rVert_2.
    \end{align}
    
    We conclude by using trigonometry:
    \begin{equation}
        \cos(\beta_{i,j}^k+\alpha_{i,j}^k)+\cos(\beta_{i,j}^k-\alpha_{i,j}^k) = 2\cos{\alpha_{i,j}^k}\cos{\beta_{i,j}^k}.
    \end{equation}
\end{proof}

\begin{proposition}
    \label{prop: circle scal expec delta n+1 Sn with simplification}
    We have for all $k$, $i$, and $j$:
    $$ \left\langle\mathbb{E}(\Delta^{k+1}\mid X_k),\mathcal{S}^k\right\rangle = -\frac{2}{N} \left\lVert \mathcal{S}^k \right\rVert_2^2 + \frac{4}{N(N-1)}\sum\limits_{i<j}\frac{\sin{\alpha_{i,j}^k}}{\alpha_{i,j}^k} \cos{\beta_{i,j}^k} \left\lVert \mathcal{S}^k \right\rVert_2. $$
\end{proposition}

\begin{proof}
    This is a direct consequence of \cref{prop: circle scal expec delta n+1 Sn no simplification,prop: circle sum cos alpha cos beta}.
\end{proof}

\begin{proposition}
    \label{prop: circle expec scal ui uj}
    We have for all $k$, $i$, and $j$:
    $$ \mathbb{E}_{i,j}(\left\langle x_i^{k+1},x_j^{k+1} \right\rangle \mid X_k) = \frac{1}{2(\alpha_{i,j}^k)^2}\left(1-\cos(2\alpha_{i,j}^k)\right). $$
\end{proposition}

\begin{proof}
    Calculations give:
    \begin{align}
        \mathbb{E}_{i,j}(\left\langle x_i^{k+1},x_j^{k+1} \right\rangle \mid X_k) &= \int\limits_{(\theta_i,\theta_j)\in\left[-\alpha_{i,j}^k,\alpha_{i,j}^k\right]^2}\cos(\theta_j-\theta_i)\frac{d\theta_i}{2\alpha_{i,j}^k}\frac{d\theta_j}{2\alpha_{i,j}^k} \\
        &= \int\limits_{\theta_i = -\alpha_{i,j}^k}^{\alpha_{i,j}^k} \big(\sin(\alpha_{i,j}-\theta_i) - \sin(-\alpha_{i,j}-\theta_i)\big)\frac{d\theta_i}{4(\alpha_{i,j}^k)^2} \\
        &= \frac{1}{2(\alpha_{i,j}^k)^2}\left(1-\cos(2\alpha_{i,j}^k)\right).
    \end{align}
\end{proof}

\begin{proposition}
    \label{prop: circle expec Delta ij norm square}
    We have for all $k$, $i$, and $j$:
    $$ \mathbb{E}_{i,j}(\left\lVert \Delta_{i,j}\right\rVert_2^2\mid X_k) = 4+\frac{1}{(\alpha_{i,j}^k)^2}(1-\cos(2\alpha_{i,j}^k))+2\cos(2\alpha_{i,j}^k) -8\frac{\sin{\alpha_{i,j}^k}}{\alpha_{i,j}^k}\cos{\alpha_{i,j}^k}. $$
\end{proposition}

\begin{proof}
    We have for all agents $i\neq j$:
    \begin{align}
        \left\lVert \Delta_{i,j}\right\rVert_2^2 &= \left\langle x_i^{k+1}+x_j^{k+1}-x_i^k-x_j^k,x_i^{k+1}+x_j^{k+1}-x_i^k-x_j^k\right\rangle \nonumber\\
        &= 4+2\left\langle x_i^{k+1},x_j^{k+1}\right\rangle + 2\left\langle x_i^k,x_j^k\right\rangle - 2\left\langle x_i^{k+1},x_j^k\right\rangle -2 \left\langle x_i^k,x_j^{k+1}\right\rangle  \nonumber\\
        &\hspace{2em}- 2 \left\langle x_i^{k+1},x_i^k\right\rangle -2 \left\langle x_j^{k+1},x_j^k\right\rangle.
    \end{align}
    
    We get the desired result using \cref{prop: circle expec xi,prop: circle expec scal ui uj}, and since:
    \begin{equation}
        \begin{cases}
            \left\langle x_{\textit{bis},i,j}^k ,x_i^k\right\rangle = \left\langle x_{\textit{bis},i,j}^k ,x_j^k\right\rangle = \cos{\alpha_{i,j}^k} \\
            \left\langle x_i^k ,x_j^k\right\rangle = \cos(2\alpha_{i,j}^k).
        \end{cases}
    \end{equation}
\end{proof}

\begin{proposition}
    \label{prop: circle sum cos 2alpha ij}
    We have for all $k$, $i$, and $j$:
    $$ \sum\limits_{i<j} \cos(2\alpha_{i,j}^k) = \frac{\left\lVert \mathcal{S}^k \right\rVert_2^2-N}{2}. $$
\end{proposition}

\begin{proof}
    We have:
    \begin{align}
        \sum\limits_{i\neq j} \cos(2\alpha_{i,j}^k) &= 
        \sum\limits_{i\neq j} \left\langle x_i^k,x_j^k \right\rangle \nonumber\\
        &= \sum\limits_{i=1}^N  \left\langle x_i^k,\mathcal{S}^k - x_i^k \right\rangle \nonumber\\
        &= \sum\limits_{i=1}^N \Big( \left\langle x_i^k,\mathcal{S}^k \right\rangle -1\Big) \nonumber\\
        &= \left\lVert \mathcal{S}^k \right\rVert_2^2-N.
    \end{align}
\end{proof}

\begin{theorem}
    \label{th: circle expec S n+1 norm squared}
    For all time steps $k$, we have:
    \begin{align*}
        \mathbb{E}(\left\lVert \mathcal{S}^{k+1} \right\rVert_2^2\mid X_k) &= \left\lVert \mathcal{S}^k \right\rVert_2^2 +4\left(1-\frac{1}{2(N-1)}\right)\left(1-\frac{\left\lVert \mathcal{S}^k \right\rVert_2^2}{N}\right)\\
        &\hspace{2em}+\frac{4}{N(N-1)}\sum\limits_{i<j} \Bigg[\left(\frac{\sin{\alpha_{i,j}^k}}{\alpha_{i,j}^k}\right)^2+2\frac{\sin{\alpha_{i,j}^k}}{\alpha_{i,j}^k}\cos{\beta_{i,j}^k}\left\lVert \mathcal{S}^k \right\rVert_2\\
        &\hspace{2em}-4\frac{\sin{\alpha_{i,j}^k}}{\alpha_{i,j}^k}\cos{\alpha_{i,j}^k}\Bigg].
    \end{align*}
\end{theorem}

\begin{proof}
    We have:
    \begin{equation}
        \left\lVert \mathcal{S}^{k+1} \right\rVert_2^2 = \left\langle \mathcal{S}^k+\Delta^{k+1}, \mathcal{S}^k+\Delta^{k+1} \right\rangle = \left\lVert \mathcal{S}^k \right\rVert_2^2 +2\left\langle \Delta^{k+1},\mathcal{S}^k \right\rangle+\left\lVert \Delta^{k+1} \right\rVert_2^2.
    \end{equation}
    
    By taking the expectation and using \cref{prop: circle scal expec delta n+1 Sn with simplification,prop: circle expec Delta ij norm square}, we get:
    \begin{align}
        \mathbb{E}(\left\lVert \mathcal{S}^{k+1} \right\rVert_2^2\mid X_k) &= \left\lVert \mathcal{S}^k \right\rVert_2^2 -\frac{4}{N}\left\lVert \mathcal{S}^k \right\rVert_2^2+\frac{8}{N(N-1)}\sum\limits_{i<j}\frac{\sin{\alpha_{i,j}^k}}{\alpha_{i,j}^k}\cos{\beta_{i,j}^k}\left\lVert \mathcal{S}^k \right\rVert_2 \nonumber\\
        &\hspace{2em}+ 4 +\frac{2}{N(N-1)}\sum\limits_{i<j}\frac{1-\cos(2\alpha_{i,j}^k)}{(\alpha_{i,j}^k)^2} \nonumber\\
        &\hspace{2em}+\frac{4}{N(N-1)}\sum\limits_{i<j}\left(\cos(2\alpha_{i,j}^k)-4\cos{\alpha_{i,j}^k}\frac{\sin{\alpha_{i,j}^k}}{\alpha_{i,j}^k} \right).
    \end{align}

    Recall that:
    \begin{equation}
        1-\cos(2\alpha_{i,j}^k) = 2\left(\sin\alpha_{i,j}^k\right)^2,
    \end{equation}
    thus by using \cref{prop: circle sum cos 2alpha ij} and regrouping the terms outside the summation symbol we get the desired result:
    \begin{align}
        \mathbb{E}(\left\lVert \mathcal{S}^{k+1} \right\rVert_2^2\mid X_k) &= \left\lVert \mathcal{S}^k \right\rVert_2^2 +4\left(1-\frac{1}{2(N-1)}\right)\left(1-\frac{\left\lVert \mathcal{S}^k \right\rVert_2^2}{N}\right) \nonumber\\
        &\hspace{2em}+\frac{4}{N(N-1)}\sum\limits_{i<j} \Bigg[\left(\frac{\sin{\alpha_{i,j}^k}}{\alpha_{i,j}^k}\right)^2+2\frac{\sin{\alpha_{i,j}^k}}{\alpha_{i,j}^k}\cos{\beta_{i,j}^k}\left\lVert \mathcal{S}^k \right\rVert_2 \nonumber\\
        &\hspace{2em}-4\frac{\sin{\alpha_{i,j}^k}}{\alpha_{i,j}^k}\cos{\alpha_{i,j}^k}\Bigg].
    \end{align}
    
\end{proof}

We conjecture that the terms of the sum in \cref{th: circle expec S n+1 norm squared} can be bounded in a wise way depending only on $N$ and on $\left\lVert \mathcal{S}^k \right\rVert_2$ such that the expectation of $N^2-\left\lVert \mathcal{S}^k\right\rVert_2^2$ not only decreases but also is bounded by an exponentially decreasing function.

\subsection{Maximal empty angle}

\begin{definition}
    At time step $k$, sort indices, using permutation $\sigma_k$, such that the angles are sorted in trigonometric order along the circle. For each index $i\in\{1,\cdots,N\}$, we define $\gamma_i^k$ the angle between agents $\sigma_{k}^{-1}(i)$ and $\sigma_{k}^{-1}(i+1)$ (indices computed modulo $N$) along the non necessarily geodesic arc from $\theta_{\sigma_{k}^{-1}(i)}^k$ to $\theta_{\sigma_{k}^{-1}(i+1)}^k$ in trigonometric order:
    \begin{equation*}
        \gamma_i^k =
        \begin{cases}
            \theta_{\sigma_{k}^{-1}(i+1)}^k - \theta_{\sigma_{k}^{-1}(i)}^k &\text{if } \theta_{\sigma_{k}^{-1}(i+1)}^k \ge \theta_{\sigma_{k}^{-1}(i)}^k \\
            \theta_{\sigma_{k}^{-1}(i+1)}^k - \theta_{\sigma_{k}^{-1}(i)}^k - 2\pi &\text{otherwise}
        \end{cases}
    \end{equation*}  
\end{definition}

Note that in the previous definition, the second case in the formula will occur for exactly one $i$ since we defined the angles to be in the range $\left[0,2\pi\right)$. For instance, if our reference direction for starting the sorting of the opinions is the positive $x$-axis then this case can only occur between agents $\sigma_{k}^{-1}(N)$ and $\sigma_{k}^{-1}(N+1)$, where agent $\sigma_{k}^{-1}(N+1)$ is agent $\sigma_{k}^{-1}(1)$ by modulo representation of indices.

\begin{definition}
    Let $\gamma_{\mathrm{max}}^k$ be the maximal empty angle: $\gamma_{\mathrm{max}}^k = \max\limits_{i}\gamma_i^k$.
\end{definition}

See \cref{fig: ND Circle empty angles} for an example of maximal empty angle.

\begin{figure}[tbhp]
  \centering
    \includegraphics[width=0.9\textwidth]{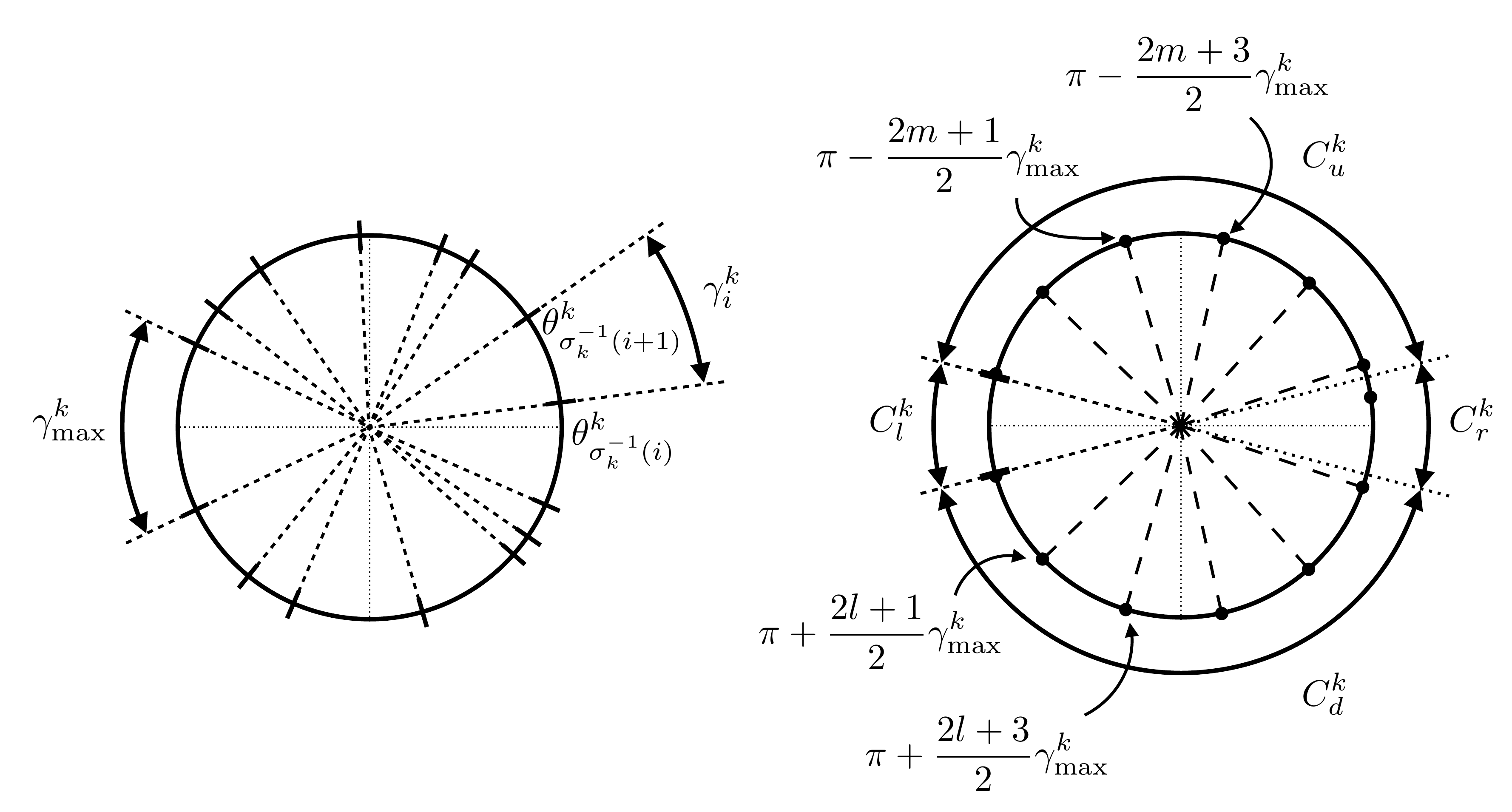}
    \caption{Examples of maximal empty angle. Left: the reference direction is taken to be the opposite of the bisector of the maximal empty angle. Other opinions are then distributed inside the circle arc, geodesic if and only if $\gamma_{\mathrm{max}}^k\ge\pi$, but not necessarily uniformly. Right: Worst case scenario of the distribution of other opinions given the maximal empty angle such that the probability to have an opinion at the next step fall into the left maximal empty arc is maximal: a unique opinion is on the border of each top and bottom chunk of size $\gamma_{\mathrm{max}}^k$, one opinion is on the right side and cannot contribute to the probability as it cannot fall into the left arc, and all other opinions are on the borders of the maximal empty angle arc. }
    \label{fig: ND Circle empty angles}
\end{figure}

\begin{proposition}
    If at time step $k$ we have $\gamma_{\mathrm{max}}^k>\pi$, then all opinions are within a half disk.
\end{proposition}

\begin{proof}
    This is by definition of $\gamma_i^k$ to be the ``empty angles'', thus there are no agents inside the open arc $\theta_{\sigma_{k}^{-1}(i_{\mathrm{max}})}^k$ to $\theta_{\sigma_{k}^{-1}(i_{\mathrm{max}}+1)}^k$ where $i_{\mathrm{max}}$ is such that $\gamma_{\sigma_{k}^{-1}(i_{\mathrm{max}})}^k = \gamma_{\mathrm{max}}^k$.
\end{proof}

The idea in this approach is that $\gamma_{\mathrm{max}}^k$ will most likely increase in size between each step, and, at some point, will be larger than $\pi$. That would mean that all angles are then contained in a half-disk for which we then know that we have convergence in finite expected time.

We will from now on further assume that $\gamma_{\mathrm{max}}^k$ is unique for all $k$. This is almost surely true, except for the cases when the initial distributions of opinion have several maximal $\gamma_{\mathrm{max}}^0$.

\begin{theorem}
    If $\gamma_{\mathrm{max}}^k < \frac{\pi}{2}$, then:
    \begin{equation*}
        \mathbb{P}(\gamma_{\mathrm{max}}^{k+1} < \gamma_{\mathrm{max}}^k \mid \gamma_{\mathrm{max}}^k) \le \frac{1}{2}\left(1-\frac{1}{N}\right)
    \end{equation*}
\end{theorem}

\begin{proof}
    The proof is based on a worst case analysis. Without loss of generality, up to a rotation of the origin axis for angles, assume the negative $x$-axis is the natural bisector of $\gamma_{\mathrm{max}}^k$.
    
    Then partition the circle into four regions:
    \begin{align}
        C_l^k &= \left(\pi-\frac{\gamma_{\mathrm{max}}^k}{2},\pi+\frac{\gamma_{\mathrm{max}}^k}{2}\right) \\
        C_r^k &= \left(2\pi-\frac{\gamma_{\mathrm{max}}^k}{2},2\pi\right)\cup\left[0,\frac{\gamma_{\mathrm{max}}^k}{2}\right) \\
        C_u^k &= \left[\frac{\gamma_{\mathrm{max}}^k}{2},\pi-\frac{\gamma_{\mathrm{max}}^k}{2}\right] \\
        C_d^k &= \left[\pi+\frac{\gamma_{\mathrm{max}}^k}{2},2\pi - \frac{\gamma_{\mathrm{max}}^k}{2}\right].
    \end{align}
    
    We desire to find the worst case scenario for agents to fall in $C_l^k$ at the next time step. By definition, there are no agents in $C_l^k$. This implies that any agent in $C_r^k$ can never be updated for the next time in $C_l^k$. Thus agents in $C_r^k$ cannot contribute to reduce $\gamma_{\mathrm{max}}^k$. We will then not care about the distribution of agents in that circle arc, given their number.
    
    Furthermore, by unicity of $\gamma_{\mathrm{max}}^k$, we have at least one agent in $C_r^k$. Note that if we don't have the uniqueness assumption, then the proof will remain identical unless $\gamma_{\mathrm{max}}^k$ is not an integer fraction of $\pi$, which never happens almost surely (except when it is already the case in the original distribution of agents).
    
    Denote $K_u^k$ and $K_d^k$ the number of agents in $C_u^k$ and $C_d^k$ respectively. For an agent in $C_u^k$ to be updated into $C_l^k$, it needs to be selected with an agent of $C_d^k$ for update.
    
    Assume for now that $K_u^k$ and $K_d^k$ are fixed. We want to find what are the worst distributions of the agents in $C_u^k$ and $C_d^k$ such that the probability is the highest for updates to fall in $C_l^k$. This happens when the agents are as close as possible to the ``left'' border of their arc, i.e. to the borders of $C_l^k$, since given the choice of agents for update $i$ in $C_u^k$ and $j$ in $C_d^k$, the probability to fall in the $C_l^k$, if $\theta_j-\theta_i<\pi$, is:
    \begin{equation}
        \mathbb{P}_{i,j}(\theta_i^{k+1}\in C_l^k\mid \Theta_k) = \mathbb{P}_{i,j}(\theta_j^{k+1}\in C_l^k\mid \Theta_k) = 
        \frac{\gamma_{\mathrm{max}}^k}{\theta_j-\theta_i}.
    \end{equation}

    However, we cannot necessarily have all agents along the ``left'' border of their domain since we need to respect the maximality assumption of $\gamma_{\mathrm{max}}^k$. Thus if we partition $C_u^k$ and $C_d^k$ into chunks of circle arcs of length $\gamma_{\mathrm{max}}^k$ (up to a residual arc), we must have at least one agent in each chunk. More formally, we have at least one agent with angle in $\left(\pi-\frac{2m+3}{2}\gamma_{\mathrm{max}}^k,\pi-\frac{2m+1}{2}\gamma_{\mathrm{max}}^k\right]$ and at least one agent with angle in $\left[\pi+\frac{2m+1}{2}\gamma_{\mathrm{max}}^k,\pi+\frac{2m+3}{2}\gamma_{\mathrm{max}}^k\right)$ for all $m\in \{0,\cdots,\floor{\frac{\pi-2\gamma_{\mathrm{max}}^k}{\gamma_{\mathrm{max}}^k}}\}$. Note that the last chunk of $C_u^k$ and $C_d^k$ might overlap with $C_r^k$ but the claim remains valid. The worst case would then be achieved if there is exactly one agent in each of these intervals, located on the ``left'' extremity, and if all other agents in $C_u^k$ and $C_d^k$ are along the extremities of $C_l^k$. See \cref{fig: ND Circle empty angles} for an example of worst case scenario.
    
    We can now look at the probability, given the worst-case scenario that gives the highest chance of having an agent at step $k+1$ in $C_l^k$, to have such an event. Note that we could also use this scenario and work similarly to what we are going to do to easily estimate an upper bound on the decrease distribution of $\gamma_{\mathrm{max}}^k$  and its impact on the expectation: $\mathbb{E}\Big(\big( \gamma_{\mathrm{max}}^{k+1}\mid \gamma_{\mathrm{max}}^k \big) \mid \gamma_{\mathrm{max}}^{k+1} \le \gamma_{\mathrm{max}}^{k}\Big)$.

    Let $i$ and $j$ be agents for update of $C_u^k$ and $C_d^k$ respectively such that their update can fall into $C_l^k$. We distinguish four cases.
    
    First, both $i$ and $j$ are along the borders of $C_l^k$. There are $N_{ll}$ such pairs, and for them:
    \begin{equation}
        \begin{cases}
            \mathbb{P}_{i,j}(\gamma_{\mathrm{max}}^{k+1}<\gamma_{\mathrm{max}}^k\mid \Theta_k) = \mathbb{P}_{j,i}(\gamma_{\mathrm{max}}^{k+1}<\gamma_{\mathrm{max}}^k\mid \Theta_k) = 1\\
            N_{ll} = \left(K_d^k-\floor{\frac{\pi-2\gamma_{\mathrm{max}}^k}{\gamma_{\mathrm{max}}^k}}\right)\left(K_u^k-\floor{\frac{\pi-2\gamma_{\mathrm{max}}^k}{\gamma_{\mathrm{max}}^k}}\right).
        \end{cases}
    \end{equation}
    
    Second, $i$ is along the borders of $C_l^k$ but $j$ is not. We can write $\theta_j^k = \pi+\frac{2m+1}{2}\gamma_{\mathrm{max}}^k$ where $m$ a non zero integer. Then for this $j$ there are $N_{lr,m}$ such pairs, and for them:
    \begin{equation}
        \begin{cases}
            \mathbb{P}_{i,j}(\gamma_{\mathrm{max}}^{k+1}<\gamma_{\mathrm{max}}^k\mid \Theta_k) = \mathbb{P}_{j,i}(\gamma_{\mathrm{max}}^{k+1}<\gamma_{\mathrm{max}}^k\mid \Theta_k) = \frac{\gamma_{\mathrm{max}}^k}{\gamma_{\mathrm{max}}^k+m\gamma_{\mathrm{max}}^k} = \frac{1}{1+m}\\
            N_{lr,m} = K_u^k-\floor{\frac{\pi-2\gamma_{\mathrm{max}}^k}{\gamma_{\mathrm{max}}^k}}.
        \end{cases}
    \end{equation}
    
    Third, $j$ is along the borders of $C_l^k$ but $i$ is not. This case is symmetric to the second one, and for $\theta_i^k = \pi-\frac{2m+1}{2}\gamma_{\mathrm{max}}^k$ where $m$ non zero, there are $N_{rl,m}$ such pairs and for them:
    \begin{equation}
        \begin{cases}
            \mathbb{P}_{i,j}(\gamma_{\mathrm{max}}^{k+1}<\gamma_{\mathrm{max}}^k\mid \Theta_k) = \mathbb{P}_{j,i}(\gamma_{\mathrm{max}}^{k+1}<\gamma_{\mathrm{max}}^k\mid \Theta_k) = \frac{1}{1+m}\\
            N_{rl,m} = K_d^k-\floor{\frac{\pi-2\gamma_{\mathrm{max}}^k}{\gamma_{\mathrm{max}}^k}}.
        \end{cases}
    \end{equation}
    
    Fourth, neither $i$ nor $j$ are along the borders of $C_l^k$. We have $\theta_i^k = \pi-\frac{2m+1}{2}\gamma_{\mathrm{max}}^k$ and $\theta_j^k = \pi+\frac{2l+1}{2}\gamma_{\mathrm{max}}^k$, where $m$ and $l$ are non zero integers. For this choice of pair there are $N_{rr,m,l}$ pairs and for them:
    \begin{equation}
        \begin{cases}
            \mathbb{P}_{i,j}(\gamma_{\mathrm{max}}^{k+1}<\gamma_{\mathrm{max}}^k\mid \Theta_k) = \mathbb{P}_{j,i}(\gamma_{\mathrm{max}}^{k+1}<\gamma_{\mathrm{max}}^k\mid \Theta_k) = \frac{\gamma_{\mathrm{max}}^k}{\gamma_{\mathrm{max}}^k+m\gamma_{\mathrm{max}}^k + l\gamma_{\mathrm{max}}^k} = \frac{1}{1+m+l}\\
            N_{rr,m,l} = 1.
        \end{cases}
    \end{equation}
    
    Note that for simplicity of the proof we have not been careful in the fourth case whether, depending on $m$ and $l$, the agents actually can update into $C_l^k$, as their geodesic must pass through $C_l^k$ which is not always the case. If that is the case, the probability of decrease of the $\gamma_{\mathrm{max}}^k$ conditional to choosing these agents is simply $0$. Thus by taking in their contribution, we slightly over estimate the maximum probability of decrease in the worst case scenario, which still gives a valid final bound, and it provides easy calculations.
    
    We can now write, by conditioning on all four cases, conditioned to the number of agents in $C_u^k$ and $C_l^k$, which by abuse of notation we will denote in the conditioning by $(K_u^k,K_d^k)$ as a subindex to the probability operator:
    \begin{align}
        \mathbb{P}_{(K_u^k,K_d^k)}(\gamma_{\mathrm{max}}^{k+1}<\gamma_{\mathrm{max}}^k\mid \gamma_{\mathrm{max}}^k) &\le \frac{1}{\binom{N}{2}}\Bigg[N_{ll} + \sum\limits_{m=1}^{\ceil{\frac{\pi-2\gamma_{\mathrm{max}}^k}{\gamma_{\mathrm{max}}^k}}}\frac{1}{1+m}\Big(N_{rl,m}+N_{lr,m}\Big) \nonumber\\
        &\hspace{2em}+
        \sum\limits_{m=1}^{\ceil{\frac{\pi-2\gamma_{\mathrm{max}}^k}{\gamma_{\mathrm{max}}^k}}}
        \sum\limits_{l=1}^{\ceil{\frac{\pi-2\gamma_{\mathrm{max}}^k}{\gamma_{\mathrm{max}}^k}}} \frac{1}{1+m+l}N_{rr,m,l}
        \Bigg] \nonumber\\
        &\le \frac{1}{\binom{N}{2}}\Bigg[ \left(K_d^k-\floor{\frac{\pi-2\gamma_{\mathrm{max}}^k}{\gamma_{\mathrm{max}}^k}}\right)\left(K_u^k-\floor{\frac{\pi-2\gamma_{\mathrm{max}}^k}{\gamma_{\mathrm{max}}^k}}\right) \nonumber\\
        &\hspace{2em}+ \sum\limits_{m=1}^{\ceil{\frac{\pi-2\gamma_{\mathrm{max}}^k}{\gamma_{\mathrm{max}}^k}}}\frac{1}{1+m}\Big(K_u^k+K_d^k-2\floor{\frac{\pi-2\gamma_{\mathrm{max}}^k}{\gamma_{\mathrm{max}}^k}}\Big) \nonumber\\
        &\hspace{2em}+
        \sum\limits_{m=1}^{\ceil{\frac{\pi-2\gamma_{\mathrm{max}}^k}{\gamma_{\mathrm{max}}^k}}}
        \sum\limits_{l=1}^{\ceil{\frac{\pi-2\gamma_{\mathrm{max}}^k}{\gamma_{\mathrm{max}}^k}}} \frac{1}{1+m+l}
        \Bigg].
    \end{align}
    
    We now look for the worst case scenario for $(K_u^k,K_d^k)$ leading to the worst upper bound. For calculation simplicity, we will forget that $K_u^k$ and $K_d^k$ are integers. Denote $f_1$, $f_2$, and $g_1$, the functions:
    \begin{align}
        f_1(x,y) &= \left(x-\floor{\frac{\pi-2\gamma_{\mathrm{max}}^k}{\gamma_{\mathrm{max}}^k}}\right)\left(y-\floor{\frac{\pi-2\gamma_{\mathrm{max}}^k}{\gamma_{\mathrm{max}}^k}}\right) \\
        g_1(x,y) &= xy \\
        f_2(x,y) &= (x+y-2\floor{\frac{\pi-2\gamma_{\mathrm{max}}^k}{\gamma_{\mathrm{max}}^k}}).
    \end{align}
    
    Since $C_r^k$ is non empty, we have:
    \begin{equation}
        K_u^k+K_d^k+1\le N.
    \end{equation}
    
    Thus $f_2(K_u^k,K_d^k)$ is maximised for $K_u^k+K_d^k = N-1$ leading for such cases to the maximum value $N-1-2\floor{\frac{\pi-2\gamma_{\mathrm{max}}^k}{\gamma_{\mathrm{max}}^k}}$.
    
    Maximising $f_1(K_u^k,K_d^k)$ is equivalent to maximising $g_1(x,y)$ with $x$ corresponding to $K_u^k - \floor{\frac{\pi-2\gamma_{\mathrm{max}}^k}{\gamma_{\mathrm{max}}^k}}$ and $y$ corresponding to $K_d^k-\floor{\frac{\pi-2\gamma_{\mathrm{max}}^k}{\gamma_{\mathrm{max}}^k}}$. We have $x$ and $y$ are constrained by:
    \begin{equation*}
        x+y \le N-1-2\floor{\frac{\pi-2\gamma_{\mathrm{max}}^k}{\gamma_{\mathrm{max}}^k}}.
    \end{equation*}
    
    Along each line $x+y = t$ where $t$ a strictly positive parameter, we have that $g_1$ is maximised for $x=y$ and this maximal value is simply $\frac{t^2}{4}$, which increases with $t$.
    
    By maximality of $\gamma_{\mathrm{max}}^k$, we have that:
    \begin{equation}
        \gamma_{\mathrm{max}}^k \ge \frac{2\pi}{N}.
    \end{equation}
    
    Therefore:
    \begin{equation}
        \frac{N-1}{2}-\frac{\pi-2\gamma_{\mathrm{max}}^k}{\gamma_{\mathrm{max}}^k} \ge \frac{N-1}{2} - \frac{\pi - 2\frac{2\pi}{N}}{\frac{2\pi}{N}} = \frac{N-1}{2} - \frac{N}{2} +2 = \frac{3}{2} > 1,
    \end{equation}
    which leads, to:
    \begin{equation}
        \frac{N-1}{2}-\floor{\frac{\pi-2\gamma_{\mathrm{max}}^k}{\gamma_{\mathrm{max}}^k}} > 0.
    \end{equation}
    
    Thus, for our optimization, the domain is the triangle in the upper right quadrant limited by the segment $x+y = N-1-2\floor{\frac{\pi-2\gamma_{\mathrm{max}}^k}{\gamma_{\mathrm{max}}^k}}$. Therefore the maximal value reached by $g_1$ is reached at the point $x=y=\frac{N-1}{2}-\floor{\frac{\pi-2\gamma_{\mathrm{max}}^k}{\gamma_{\mathrm{max}}^k}}$. This means that $f_1$ is maximal under the constraints for its variables at $K_u^k=K_d^k = \frac{N-1}{2}$. Its maximum value is $\left(\frac{N-1}{2}-\floor{\frac{\pi-2\gamma_{\mathrm{max}}^k}{\gamma_{\mathrm{max}}^k}}\right)^2$
    
    Note that both $f_1$ and $f_2$ are maximised, subject to the constraint, for $K_u^k=K_d^k = \frac{N-1}{2}$. We can now remove the conditioning on $(K_u^k,K_d^k)$ on the worst-case bound:
    \begin{equation}
        \mathbb{P}(\gamma_{\mathrm{max}}^{k+1}\mid \gamma_{\mathrm{max}}^{k}\mid \gamma_{\mathrm{max}}^{k}) \le B(N,\gamma_{\mathrm{max}}^k),
    \end{equation}
    where $B(N,\gamma_{\mathrm{max}}^k)$ is the previously derived bound in which we plugged in the maximum for $f_1$ and $f_2$. That is:
    \begin{align}
        B(N,\gamma_{\mathrm{max}}^k) &= 
        \frac{1}{\binom{N}{2}}\Bigg[ \left(\frac{N-1}{2}-\floor{\frac{\pi-2\gamma_{\mathrm{max}}^k}{\gamma_{\mathrm{max}}^k}}\right)^2 \nonumber\\
        &\hspace{2em}+ 2\sum\limits_{m=1}^{\ceil{\frac{\pi-2\gamma_{\mathrm{max}}^k}{\gamma_{\mathrm{max}}^k}}}\frac{1}{1+m}\Big(\frac{N-1}{2}-\floor{\frac{\pi-2\gamma_{\mathrm{max}}^k}{\gamma_{\mathrm{max}}^k}}\Big) \nonumber\\
        &\hspace{2em}+
        \sum\limits_{m=1}^{\ceil{\frac{\pi-2\gamma_{\mathrm{max}}^k}{\gamma_{\mathrm{max}}^k}}}
        \sum\limits_{l=1}^{\ceil{\frac{\pi-2\gamma_{\mathrm{max}}^k}{\gamma_{\mathrm{max}}^k}}} \frac{1}{1+m+l}
        \Bigg].
    \end{align}
    
    By using the following naive bounds, for all positive $m$ and $l$:
    \begin{equation}
        \begin{cases}
        \frac{1}{1+m}\le 1 \\
        \frac{1}{1+m+l} \le 1,
        \end{cases}
    \end{equation}
    we have:
    
    \begin{align}
        B(N,\gamma_{\mathrm{max}}^k) &= \frac{2}{N(N-1)}\Bigg[
        \left(\frac{N-1}{2}-\floor{\frac{\pi-2\gamma_{\mathrm{max}}^k}{\gamma_{\mathrm{max}}^k}}\right)\nonumber\\
        &\hspace{2em}\times\left(\frac{N-1}{2}-\floor{\frac{\pi-2\gamma_{\mathrm{max}}^k}{\gamma_{\mathrm{max}}^k}}+2\sum\limits_{m=1}^{\ceil{\frac{\pi-2\gamma_{\mathrm{max}}^k}{\gamma_{\mathrm{max}}^k}}}\frac{1}{1+m}\right) \nonumber\\
        &\hspace{2em}+ \sum\limits_{m=1}^{\ceil{\frac{\pi-2\gamma_{\mathrm{max}}^k}{\gamma_{\mathrm{max}}^k}}}
        \sum\limits_{l=1}^{\ceil{\frac{\pi-2\gamma_{\mathrm{max}}^k}{\gamma_{\mathrm{max}}^k}}} \frac{1}{1+m+l}\Bigg] \\
        &\le \frac{2}{N(N-1)}\left[\left(\frac{N-1}{2}\right)^2 - \floor{\frac{\pi-2\gamma_{\mathrm{max}}^k}{\gamma_{\mathrm{max}}^k}}^2 + \floor{\frac{\pi-2\gamma_{\mathrm{max}}^k}{\gamma_{\mathrm{max}}^k}}^2 \right] \\
        &\le \frac{N-1}{2N} = \frac{1}{2}\left(1-\frac{1}{N}\right).
    \end{align}

\end{proof}

Unfortunately, simply having that the probability of decrease is upper bounded by a value strictly smaller than $\frac{1}{2}$ is not enough, we need to study with more detail the probability distribution of $\gamma_{\mathrm{max}}^k$. We can further study the distribution of the decrease of $\gamma_{\mathrm{max}}^k$ in a similar way to what is done in the previous proof, giving a bound on the decrease part of the expectation of $\gamma_{\mathrm{max}}^{k+1}$ conditionally to $\gamma_{\mathrm{max}}^k$. However, studying the increase of $\gamma_{\mathrm{max}}^k$ is significantly harder and remains an open challenge.

\subsection{Random Markov chain}
In this approach, we wish to model the evolution using tools from random Markov chains theory. If there is a unique absorbing state corresponding to all opinions are within a half-disk, then we can use these well developed tools to find the expected time to reach the absorbing state. The task consists in defining what the states and the transition probabilities are. In Markov chains, these probabilities should only depend on the current state and the next potential state. This is our main challenge.

At time step $k$, all agents are within a half-disk configuration happens if and only if $N_{\textit{HD},\textrm{max}}^k = N$. Once this is the case, then according to \cref{prop: circle half disk stable}, all agents stay within a half-disk forever, or in other words $N_{\textit{HD},\textrm{max}}^{k'} = N$ for all $k\ge k'$. Furthermore, prior to the half-disk configuration, $N_{\textit{HD},\textrm{max}}^k$ is a random variable in the finite integer set $\{\ceil{\frac{N}{2}},\cdots,N-1\}$. Therefore, by hand-waiving, $N_{\textit{HD},\textrm{max}}^k$ does some kind of random walk in $\{\ceil{\frac{N}{2}},\cdots,N-1\}$ until at some point reaching $N$ where it then stays there forever.

This behaviour suggests to look at a Markov chain with $N-\ceil{\frac{N}{2}}+1 = \floor{\frac{N}{2}}+1$ states, each denoted as the integers in the range $\{\ceil{\frac{N}{2}},\cdots,N-1\}$, and 
where the last state is absorbing. For simplicity we define the following translation of $N_{\textit{HD},\textrm{max}}^k$.

\begin{definition}
    Let $Z_k\in\{0,\cdots,\floor{\frac{N}{2}}\}$ be the state measurement of the system at step $k$ defined by:
    $$ Z_k = N_{\textit{HD},\textrm{max}}^k - \floor{\frac{N}{2}}. $$
\end{definition}

\begin{definition}
    Let $n$ be the highest possible state of $Z_k$, in other words:
    $$ n = \floor{\frac{N}{2}}.$$
\end{definition}

\begin{proposition}
    For a system evolving according to \cref{eq:circle evolution}, for all time step $k\in\mathbb{N}$ we have:
    \begin{equation*}
        \begin{cases} 
            Z_{k+1} \in \{Z_{k-1},Z_k,Z_{k+1}\} &\text{if } Z_k\neq n\\
            Z_{k+1} = n &\text{if } Z_k = n. 
        \end{cases}
    \end{equation*}
\end{proposition}

\begin{proof}
    This result is due to the fact that for each time step only two agents are updated and that the update is done following the geodesic arc between them. Thus, in one step, there cannot be an increase or a decrease of more than two agents in any half-disk. This implies that the maximal number of agents in a half-disk cannot vary by two or more in absolute value at any time step, i.e. that $N_{\textit{HD},\textrm{max}}^k$ varies by at most $1$ in absolute value. Furthermore, due to \cref{prop: circle half disk stable}, as  $Z_k = n$ is equivalent to all agents are within a half-disk, this will remain true forever, implying that $Z_{k'} = n$ for all $k'\ge k$.
\end{proof}

We thus wish to use the following graph $\mathcal{G} = (\mathcal{V},\mathcal{E},W)$ for a Markov chain analysis: $\mathcal{V} = \{0,\cdots,n\}$ and $\mathcal{E} = \{e_{0,0},e_{0,1}\}\bigcup\limits_{i=1}^{n-1}
\{e_{i,i-1}, e_{i,i}, e_{i,i+1}\}\cup\{e_{n,n}\}$, where $e_{i,j}$ is the oriented edge from state $i$ to state $j$. The weights $W = (w_{i,j})$ of the edges $e_{i,j}\in \mathcal{E}$ are yet to be defined, apart from $w_{n,n} = 1$. 

The issue is that for our system, the distribution of $Z_{k+1}$ is not solely dependent on the current state $Z_k$. It actually depends on the geometry of the opinions. Worse, there is no non zero worst case lower bound for the increase probability of $Z_k$. For instance, if on the maximal half-disk all agents are arbitrarily close the to border of the half-disk, then the probability to get a new opinion within that half-disk becomes very small. A good example is to consider the state $Z_k = n-1$ for $N\ge 4$. The worst case consists in half of those agents are arbitrarily close to one side of the border of the maximal half-disk and the other half is arbitrarily close to the other side. There is one agent that is not within the maximal half-disk. Given this geometry, the only way for $Z_k$ to increase in one step is to select that agent and an agent of the maximal half-disk and have them both update their opinions in the maximal half-disk. However, since the opinions in the maximal half-disk are arbitrarily close to the border, the probability for that event to happen becomes arbitrarily close to $0$.

In the previously defined Markov chain, the geometry is not taken into account. In order to use Markov chain theory, we need to find a way to study a well-defined Markov chain, with edge weights that are defined constants depending only on the states themselves, and with non zero probabilities for the states to increase. To do this, we will reason in a worst case scenario that provides larger expected time.

From \cref{eq:circle summary increase N Case 1,eq:circle summary increase N Case 2.1,eq:circle summary increase N Case 2.2}, we have lower bounds independent on the geometry for the increase probability in one or two steps depending on which geometric configuration we are in, for any state $Z_k<n$. Since for any $\delta\in]0,1[$ the smallest of these probabilities is $\eta_{2.1,\delta}$, we would like to define $w_{i,i+1} = \eta_{2.1,\delta}$. The issue is that this case requires two time steps whereas the others require one. In a worst case philosophy, we can naively bound $1$ step by $2$ steps when there is an increase of the measured state. This allows us to choose $w_{i,i+1} = \eta_{2.1,\delta}$. However, from now on, if at one step in the Markov chain the state increases, then this corresponds to the state has increased by one in at most two steps for a system of agents. Therefore, from now on, we have two clocks. The first one is in the primal space which is the natural time of the agents. The second one is in the dual space corresponding to the worst-case Markov chain. In the dual space, taking an edge $e_{i,i+1}$ with $i\in\{1,\cdots,n-1\}$ in the Markov chain will imply two steps in the primal space. The Markov chain we will study from now on is a worst case Markov chain that does not correspond exactly to what is happening to the opinions but that provides an average time to reach the absorbing last state $n$ larger than the average time to get into a half-disk configuration up to a constant multiplicative factor.

We need to define the other edge weights $w_{i,j}$. Following the worst case philosophy, it is worse to systematically decrease the current state than to randomly decrease or remain at the same state. Therefore we set $w_{i,i} = 0$ for all $i\in\{1,\cdots,n-1\}$. We must then choose $w_{0,0} = w_{i,i-1} = 1-\eta_{2.1,\delta}$ for all $i\in\{1,\cdots,n-1\}$. Finally, for simplicity, we can once again upper bound $1$ step by $2$ steps when there is a decrease of the measured state. This translates to a (dual) time step along the edge $e_{i,i-1}$ as two time-steps in the primal space. By doing this, we maintain the property that the expected time to reach the absorbing state is an upper bound of the expected time to reach a half-disk configuration up to a constant multiplicative factor of $2$.

Finally, according to \cref{th: circle final bounds optimised expec hd}, $\eta_{2.1,\delta}$ is maximal when $\delta = \frac{\sqrt{3}}{2}$ giving $\eta_{2.1,\frac{\sqrt{3}}{2}} = \frac{1}{27}\left(\frac{2}{N(N-1)}\right)^2$. For this choice of $\eta_{2.1,\delta}$ in the Markov chain, we will get the lowest upper bound for the convergence time.

We can now summarise this discussion rigorously.

\begin{definition}
    Let $c$ be the positive constant for the Markov chain increase probability with $c = \frac{1}{27}\left(\frac{2}{N(N-1)}\right)^2$.
\end{definition}

\begin{definition}
    Define $\mathcal{G} = (\mathcal{V},\mathcal{E},W)$ the oriented weighted graph, with vertices $\mathcal{V}$, edges $\mathcal{E}$, and edge weights $W$, defined as:
    \begin{equation*}
        \begin{cases}
            \mathcal{V} = \{0,\cdots,n\} \\
            \mathcal{E} = \{e_{0,0},e_{0,1}\}\bigcup\limits_{i=1}^{n-1}
\{e_{i,i-1}, e_{i,i+1}\}\cup\{e_{n,n}\} \\
            W = 
            \begin{pmatrix}
              1-c   & c     & 0     & \cdots & \cdots &\cdots&\cdots&0 \\
              1-c   & 0     & c     & \ddots &        &      & &\vdots \\
              0     &\ddots &\ddots & \ddots & \ddots &      & &\vdots\\
              \vdots&\ddots &\ddots & \ddots & \ddots &\ddots& & \vdots\\
              \vdots& &\ddots & \ddots & \ddots &\ddots&\ddots & \vdots\\
              \vdots& &       & \ddots & \ddots &\ddots&\ddots & 0\\
              0     &\hdots&\hdots&\hdots& 0    &  1-c &0& c\\
              0     &\hdots&\hdots&\hdots&\hdots&0     &c& 1
            \end{pmatrix}
        \end{cases}
    \end{equation*}
\end{definition}

For an illustration of $\mathcal{G}$ as a Markov chain see \cref{fig: circle markov chain}.

\begin{figure}[tbhp]
      \centering
        \includegraphics[width=0.7\textwidth]{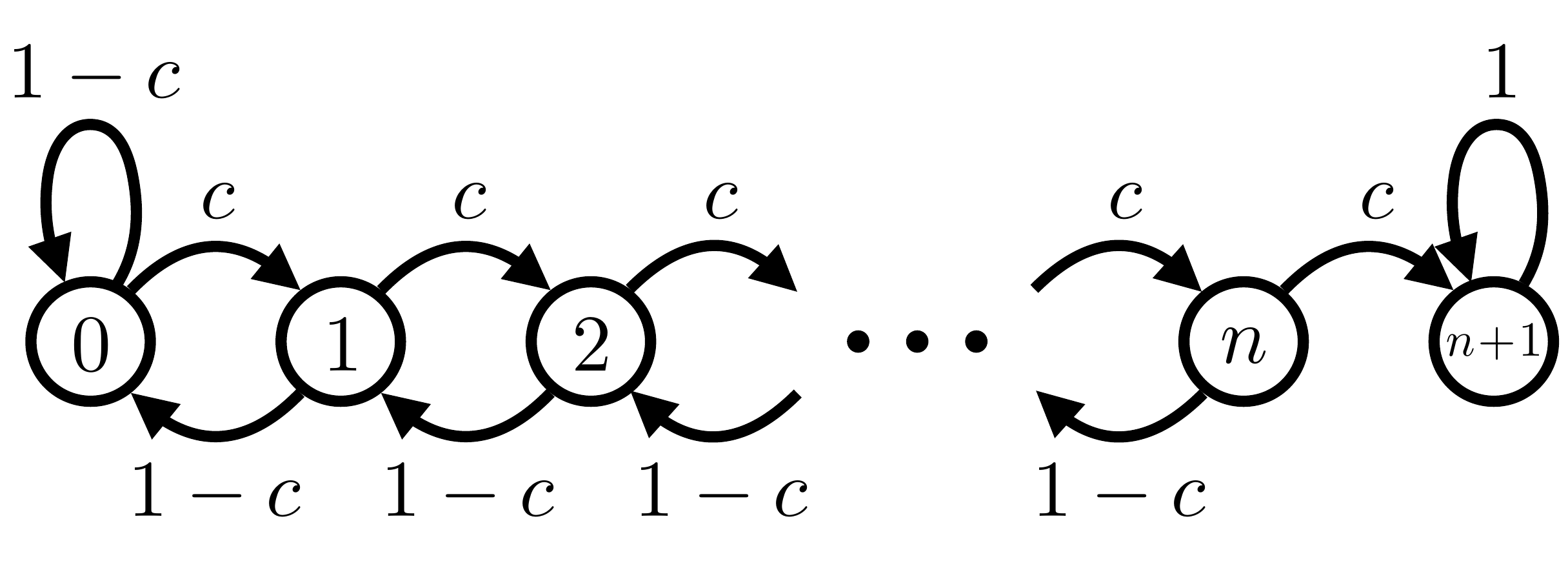}
        \caption{The Markov chain $\mathcal{G}$.}
        \label{fig: circle markov chain}
    \end{figure}

\begin{definition}
    Denote $M_k$ the random Markov process at (dual) step $k$ whose evolution is given by the Markov chain $\mathcal{G}$.
\end{definition}

\begin{definition}
    Denote $T_M$ the (dual) stopping time with respect to the natural filtration induced by the $(M_k)$ sequence to reach the absorbing state $n$, defined as:
    $$ T_M = \min\{k\in\mathbb{N}\mid M_k = n\}.$$
\end{definition}

\begin{definition}
    Let $E_{i} = \mathbb{E}(T_M\mid M_0 = i)$ be the (dual) expected time to reach the absorbing state $n$ starting from state $i$. Let $E = (E_0,\cdots,E_n)^T\in\mathbb{R}^{n+1}$.
\end{definition}

\begin{proposition}
    \label{prop: circle markov e0>e1>...>en}
    The Markov random variable $M_k$ satisfies:
    $$ E_0 > E_1 > \cdots > E_{n-1} > E_n = 0. $$
\end{proposition}

\begin{proof}
    This is a direct consequence of the structure of the graph $\mathcal{G}$, which can be seen as a doubly chained list from state $0$ to state $n-1$ and state $n$ is absorbing. 
\end{proof}

\begin{theorem}
    For a system evolving according to \cref{eq:circle evolution}, we have the following bound between the primal expectation of $T_{HD}$ and the dual expectation of $T_M$:
    \begin{equation*}
        \begin{cases}
            \mathbb{E}(T_{HD}\mid Z_0 = i) \le 2 E_i &\text{for all } i\in\{0,\cdots,n\}\\
            \mathbb{E}(T_{HD}) \le 2 E_0.
        \end{cases}
    \end{equation*}
\end{theorem}

\begin{proof}
    The first result holds immediately by construction of the Markov chain. The second comes from \cref{prop: circle markov e0>e1>...>en}.
\end{proof}

We can now use the tools from Markov chain theory to compute $E$ and especially $E_0$.

\begin{theorem}
    \label{th: circle markov E0 explicit}
    The $E_i$ can be explicitly derived for all $i\in\{0,\cdots,n-1\}$. In particular:
    \begin{equation*}
        E_0 
        = \left[1+\left(\frac{1-c}{c}\right)^n\frac{1-\left(\frac{c}{1-c}\right)^n}{1-\frac{c}{1-c}}\right]  \left[ \frac{1}{1-2c} \frac{1-\left(\frac{c}{1-c}\right)^n}{1-\left(\frac{c}{1-c}\right)^{n+1}}    - \frac{\frac{n}{c}\left(\frac{c}{1-c}\right)^{n+1}}{1-\left(\frac{c}{1-c}\right)^{n+1}} \right] 
    \end{equation*}
\end{theorem}

\begin{proof}
    Assume that we are at state $i<n$, i.e. $M_k=i$. Then we can either do an increasing unit step, with probability $c$, or a decreasing unit step with probability $1-c$. The dual time that has passed in one step is simply $1$. From the next step, we then need to reach the last state and look at when this first happens. Thus we have:
    \begin{equation}
        \begin{cases}
            E_i = 1+ (1-c) E_{i-1} + c E_{i+1} &\text{if } i\in\{1,\cdots,n-1\} \\
            E_0 = 1+(1-c) E_0 + c E_1 \\
            E_n = 0.
        \end{cases}
    \end{equation}
    
    This system of equations can be written in matrix form as, recall that $1_{n}$ is the vector of size $n$ containing only ones:
    \begin{equation}
        E = \begin{pmatrix}1_n \\ 0\end{pmatrix} + W E.
    \end{equation}
    
    If we denote $E' = (E_0,\cdots, E_{n-1})^T\in\mathbb{R}^n$, $W' = (W_{i,j})_{(i,j)\in\{1,\cdots,n\}^2}\in\mathbb{R}^{n\times n}$ to be restrictions of $E$ and $W$ without the last $n+1$ dimension, and $I_n$ the identity matrix of size $n\times n$, this matrix equality is equivalent to:
    \begin{equation}
        (I_n-W') E' = 1_n^T.
    \end{equation}
    
    We thus need to invert the matrix $A = I_n - W'$:
    \begin{equation}
        A = \begin{pmatrix}
              c   & -c     & 0     & \cdots & \cdots &\cdots&0 \\
              -(1-c)& 1     & -c     & \ddots &        &      &\vdots \\
              0     &\ddots &\ddots & \ddots & \ddots &      &\vdots\\
              \vdots&\ddots &\ddots & \ddots & \ddots &\ddots&\vdots\\
              \vdots& &\ddots & \ddots & \ddots &\ddots& 0\\
              \vdots& &       & \ddots & \ddots &\ddots& -c\\
              0     &\hdots&\hdots&\hdots& 0    &  -(1-c) &1\\
            \end{pmatrix}\in\mathbb{R}^n.
    \end{equation}
    
    This matrix is tridiagonal and nearly Toeplitz. We can write $A = B + uv^T$ where $B$ is tridiagonal Toeplitz and $u^Tv$ is a rank one matrix:
    \begin{equation}
        B = \begin{pmatrix}
                1     & -c   & 0    &\hdots&\hdots& 0 \\
                -(1-c)&\ddots&\ddots&      &      &\vdots\\
                \vdots&\ddots&\ddots&\ddots&      &\vdots\\
                \vdots&      &\ddots&\ddots&\ddots&\vdots\\
                \vdots&      &      &\ddots&\ddots&-c\\
                0     &\hdots&\hdots&\hdots&-(1-c)&1
            \end{pmatrix}, \quad u = (c-1)\begin{pmatrix} 1 \\ 0 \\ \vdots \\ 0 \end{pmatrix},\quad v = \begin{pmatrix} 1 \\ 0 \\ \vdots \\ 0 \end{pmatrix}.
    \end{equation}
    
    If $B$ is invertible, then according to the Sherman-Morrison formula, $A$ is invertible if and only $1+v^T B^{-1} u \neq 0$, and then:
    \begin{equation}
        A^{-1} = B^{-1} - \frac{B^{-1}uv^TB^{-1}}{1+v^T B^{-1} u}.
    \end{equation}
    
    The matrix $B$ is a tridiagonal Toeplitz matrix. Its diagonal value is non zero and the displaced diagonals are non zero and of opposite sign, i.e. $(-c)(-(1-c)) = c(1-c) >0$ as $c\in]0,1[$. Therefore $B$ is invertible. Inversion of $B$ is a classic problem and can be done explicitly using Tchebychev polynomials of the second kind \cite{dafonseca2001}. If we denote $U_i(x)$ the $i$-th Tchebychev polynomial of the second kind, then:
    \begin{equation}
        U_i(x) = \begin{cases} 
                    \frac{\sin((i+1)\theta)}{\sin\theta} &\text{if } \lvert x \rvert <1 \quad \text{with } \cos\theta = x \\
                    \frac{\sinh((i+1)\theta)}{\sinh\theta} &\text{if } \lvert x \rvert >1 \quad \text{with } \cosh\theta = x \\
                    (\pm 1)^i (i+1) &\text{if } x = \pm 1.
                \end{cases}
    \end{equation}
    
    Let $d = \frac{1}{2\sqrt{c(1-c)}}$. Then for $(i,j)\in\{1,\cdots,n\}$, the $(i,j)$-th entry of the invert of $B$ is given by:
    \begin{equation}
        (B^{-1})_{i,j} = 
        \begin{cases} 
            (-1)^{i+j}\frac{(-c)^{j-i}}{(\sqrt{c(1-c)})^{j-i+1}} \frac{U_{i-1}(d)U_{n-j}(d)}{U_n(d)} &\text{if } i\le j \\
            (-1)^{i+j}\frac{(-(1-c))^{i-j}}{(\sqrt{c(1-c)})^{i-j+1}} \frac{U_{j-1}(d)U_{n-i}(d)}{U_n(d)} &\text{if } i> j.
        \end{cases}
    \end{equation}
    
    We have $d = \frac{1}{2\sqrt{c(1-c)}} = \frac{\sqrt{27}N^2}{4}(1-\frac{1}{N})\frac{1}{\sqrt{1-c}}\ge \frac{3\sqrt{27}}{2} \approx 7.8 >1$ since $N\ge 3$. Thus, let $\theta = \arccosh{d}$, then:
    \begin{equation}
        (B^{-1})_{i,j} = 
        \begin{cases} 
            \frac{1}{\sqrt{c(1-c)}}\left(\sqrt{\frac{c}{1-c}}\right)^{j-i} \frac{\sinh(i\theta)\sinh((n-j+1)\theta)}{\sinh\theta \sinh((n+1)\theta)} &\text{if } i\le j \\
            \frac{1}{\sqrt{c(1-c)}}\left(\sqrt{\frac{c}{1-c}}\right)^{j-i} \frac{\sinh(j\theta)\sinh((n-i+1)\theta)}{\sinh\theta \sinh((n+1)\theta)} &\text{if } i> j.
        \end{cases}
    \end{equation}
    
    We then have:
    \begin{equation}
        (B^{-1}u)_i = (c-1)\left[\frac{1}{\sqrt{c(1-c)}}\left(\sqrt{\frac{c}{1-c}}\right)^{-(i-1)} \frac{\sinh((n-i+1)\theta)}{ \sinh((n+1)\theta)} \right].
    \end{equation}
    
    As $v^T = (1,0,\cdots,0)$, we have:
    \begin{align}
        (B^{-1}uv^TB^{-1})_{i,j} &= (B^{-1}u)_i(B^{-1})_{1,j} \nonumber \\
        &= -\frac{1}{c} \left(\sqrt{\frac{c}{1-c}}\right)^{j-i} \frac{\sinh((n-i+1)\theta)\sinh((n-j+1)\theta)}{(\sinh((n+1)\theta))^2}.
    \end{align}
    
    Furthermore, $v^T B^{-1} u  = (B^{-1}u)_1$. Therefore:
    \begin{equation}
        1+v^T B^{-1} u = 1 - \sqrt{\frac{1-c}{c}}\frac{\sinh(n\theta)}{\sinh((n+1)\theta)}.
    \end{equation}
    
    We need to test whether this quantity is equal to $0$ or not for applying the Sherman-Morrison formula. For that, we will first show that since in our problem $c\le\frac{1}{2}$, we have $e^\theta = \sqrt{\frac{1-c}{c}}$. Note that if we had had $c>\frac{1}{2}$, then a similar approach would have given $e^{-\theta} = \sqrt{\frac{1-c}{c}}$, and that for a constant $c>\frac{1}{2}$ the associated Markov chain would have linear expected time to reach the absorbing state. We will continue from now on with $c\frac{1}{2}$. Calculations give, by recalling the alternative definition of the $\arccosh$ function as $\arccosh(x) = \ln(x+\sqrt{x^2-1}))$:
    \begin{equation}
        e^\theta = d\left(1+\sqrt{1-\frac{1}{d^2}}\right).
    \end{equation}
    
    We can then calculate, using the fact that $c\le\frac{1}{2}$ implies $\sqrt{(1-2c)^2} = 1-2c$:
    \begin{equation}
        e^{\theta} = d\left(1+\sqrt{1-\frac{1}{d^2}}\right) =  \frac{1+\sqrt{1-4c(1-c)}}{2\sqrt{c(1-c)}} = \frac{2(1-c)}{2\sqrt{c(1-c)}} = \sqrt{\frac{1-c}{c}}
    \end{equation}
    
    We now go back to inverting $A$. For that:
    \begin{align}
        1+v^TB^{-1}u = 0 &\iff \frac{1-c}{c} = \left(\frac{\sinh((n+1)\theta)}{\sinh(n\theta)}\right)^2 = e^{2\theta}\left(\frac{1-e^{-2(n+1)\theta}}{1-e^{-2n\theta}}\right)^2 \\
        &\iff 1 = \left(\frac{1-e^{-2(n+1)\theta}}{1-e^{-2n\theta}}\right)^2,
    \end{align}
    which is never true. Thus $A$ is invertible and the invert is given using the Sherman-Morrison formula. We have:
    \begin{align}
        \left(\frac{B^{-1}uv^TB^{-1}}{1+v^TB^{-1}u}\right)_{i,j} &= \frac{-\frac{1}{c}}{1-\sqrt{\frac{1-c}{c}}\frac{\sinh(n\theta)}{\sinh((n+1)\theta)}}\left(\sqrt{\frac{c}{1-c}}\right)^{j-i} \nonumber\\
        &\hspace{2em}\times\frac{\sinh((n-i+1)\theta)\sinh((n-j+1)\theta)}{(\sinh((n+1)\theta))^2}.
    \end{align}
    
    And thus, as $(A^{-1})_{i,j} = (B^{-1})_{i,j} - (\frac{B^{-1}uv^TB^{-1}}{1+v^TB^{-1}u})_{i,j}$, we have:
    \begin{equation}
        (A^{-1})_{i,j} = 
        \begin{cases} 
            \frac{1}{\sqrt{c(1-c)}}\left(\sqrt{\frac{c}{1-c}}\right)^{j-i} \frac{\sinh(i\theta)\sinh((n-j+1)\theta)}{\sinh\theta \sinh((n+1)\theta)}  \\
            \hspace{2em} + \frac{\frac{1}{c}}{1-\sqrt{\frac{1-c}{c}}\frac{\sinh(n\theta)}{\sinh((n+1)\theta)}}\left(\sqrt{\frac{c}{1-c}}\right)^{j-i} \\
            \hspace{2em}\times\frac{\sinh((n-i+1)\theta)\sinh((n-j+1)\theta)}{(\sinh((n+1)\theta))^2} &\text{if } i\le j\\
            \frac{1}{\sqrt{c(1-c)}}\left(\sqrt{\frac{c}{1-c}}\right)^{j-i} \frac{\sinh(j\theta)\sinh((n-i+1)\theta)}{\sinh\theta \sinh((n+1)\theta)} \\
            \hspace{2em} + \frac{\frac{1}{c}}{1-\sqrt{\frac{1-c}{c}}\frac{\sinh(n\theta)}{\sinh((n+1)\theta)}}\left(\sqrt{\frac{c}{1-c}}\right)^{j-i} \\
            \hspace{2em}\times\frac{\sinh((n-i+1)\theta)\sinh((n-j+1)\theta)}{(\sinh((n+1)\theta))^2} &\text{if } i>j.
        \end{cases}
    \end{equation}
    
    Now that we have inverted $A$, we can simply get back to the expected time to reach the absorbing state, i.e. the entries of $E'$, by summing up the rows of $A^{-1}$. For all $i\in\{0,\cdots,n-1\}$:
    \begin{equation}
        E_i = \sum\limits_{j=1}^n (A^{-1})_{i+1,j}.
    \end{equation}
    
    For simplicity, we will only do the explicit calculation of $E_0$, but computing $E_i$ for other $i$ can be done in the same way. We have:
    \begin{align}
        E_0 &= \sum\limits_{j=1}^n \left(\sqrt{\frac{c}{1-c}}\right)^{j-1}\sinh((n-j+1)\theta) \nonumber\\
        &\hspace{2em} \times\left[ \frac{1}{\sqrt{c(1-c)}} \frac{1}{\sinh((n+1)\theta)}  + \frac{\frac{1}{c}}{1-\sqrt{\frac{1-c}{c}}\frac{\sinh(n\theta)}{\sinh((n+1)\theta)} } \frac{\sinh(n\theta)}{(\sinh((n+1)\theta))^2}    \right]
    \end{align}
    
    For the summation part, we can write explicitly the hyperbolic term as:
    \begin{equation}
        \sinh((n+j-1)\theta) = \frac{1}{2}(e^{(n-j+1)\theta} - e^{(n-j+1)\theta}) = e^{-j\theta}\frac{e^{(n+1)\theta}}{2} - e^{j\theta}\frac{e^{-(n+1)\theta}}{2}.
    \end{equation}
    
    The summation term is then the difference between two geometric sums, one of reason $\sqrt{\frac{c}{1-c}}e^{-\theta} = e^{-2\theta}\neq 1$ and the other of reason $\sqrt{\frac{c}{1-c}}e^{\theta} = 1$ since $c<\frac{1}{2}$. Thus:
    \begin{align}
        E_0 &= \frac{1}{2}\frac{\frac{1}{c}}{\sinh((n+1)\theta)} \left[e^{-2\theta} \frac{1-e^{-2n\theta}}{1-e^{-2\theta}}e^{(n+1)\theta} - n e^{-(n+1)\theta} \right] \nonumber\\
        &\hspace{2em} \times \left[1+\frac{\sqrt{\frac{1-c}{c}}e^{-\theta}}{1-\frac{1-e^{-2n\theta}}{1-e^{-2(n+1)\theta}}} \frac{1-e^{-2n\theta}}{1-e^{-2(n+1)\theta}}\right]
    \end{align}

    Recall now that $\sqrt{\frac{1-c}{c}}e^{-\theta} = 1$ and $e^{-2\theta} = \frac{c}{1-c}$ since $c<\frac{1}{2}$. We then get:
    \begin{equation}
        E_0 
        = \left[1+\frac{\frac{1-\left(\frac{c}{1-c}\right)^n}{1-\left(\frac{c}{1-c}\right)^{n+1}}}{1-\frac{1-\left(\frac{c}{1-c}\right)^n}{1-\left(\frac{c}{1-c}\right)^{n+1}}}\right]  \left[ \frac{1}{1-2c} \frac{1-\left(\frac{c}{1-c}\right)^n}{1-\left(\frac{c}{1-c}\right)^{n+1}}    - \frac{\frac{n}{c}\left(\frac{c}{1-c}\right)^{n+1}}{1-\left(\frac{c}{1-c}\right)^{n+1}} \right].
    \end{equation}
    
    For the large fraction in the left term, we have:
    \begin{align}
        \frac{\frac{1-\left(\frac{c}{1-c}\right)^n}{1-\left(\frac{c}{1-c}\right)^{n+1}}}{1-\frac{1-\left(\frac{c}{1-c}\right)^n}{1-\left(\frac{c}{1-c}\right)^{n+1}}} = \frac{1-\left(\frac{c}{1-c}\right)^{n+1}}{\left(\frac{c}{1-c}\right)^{n}-\left(\frac{c}{1-c}\right)^{n+1}}\frac{1-\left(\frac{c}{1-c}\right)^n}{1-\left(\frac{c}{1-c}\right)^{n+1}} = \left(\frac{1-c}{c}\right)^n \frac{1-\left(\frac{c}{1-c}\right)^n}{1-\frac{c}{1-c}}.
    \end{align}
    
    We thus get the desired result:
    \begin{equation}
        E_0 
        = \left[1+\left(\frac{1-c}{c}\right)^n\frac{1-\left(\frac{c}{1-c}\right)^n}{1-\frac{c}{1-c}}\right]  \left[ \frac{1}{1-2c} \frac{1-\left(\frac{c}{1-c}\right)^n}{1-\left(\frac{c}{1-c}\right)^{n+1}}    - \frac{\frac{n}{c}\left(\frac{c}{1-c}\right)^{n+1}}{1-\left(\frac{c}{1-c}\right)^{n+1}} \right].
    \end{equation}
    
\end{proof}

\begin{corollary}
    Asymptotically in $N$, we have:
    $$ E_0 = \left(\frac{1-c}{c}\right)^n \big(1+o(1)\big) $$
\end{corollary}

\begin{proof}
    This behaviour is immediate from \cref{th: circle markov E0 explicit} after recalling that $n=\floor{\frac{N}{2}}\xrightarrow[]{}\infty$ and $c = \frac{4}{27N^4}\frac{1}{\left(1-\frac{1}{N}\right)^2} \xrightarrow[]{}0$ when $N\xrightarrow[]{}\infty$.
\end{proof}

Unfortunately, by using the naive bound $N\ge 3$, we have: \begin{equation}
    \left(3\left(1-\frac{1}{3^5}\right)\right)^n N^{4n} \le\left(\frac{1-c}{c}\right)^n \le \left(\frac{27}{4}N^4\right)^n\le \left(\frac{3\sqrt{3}}{2}\right)^N N^{2N}.
\end{equation}

This implies that the bound provided by this strategy is of similar order of magnitude as the one simply looking at successive $n$ jumps with probability $c$ given in \cref{th:circle to half disk finite expec time optimised}. The reason this happens is because $c<\frac{1}{2}$. For intuition, what happens is that when the Markov process is at state $i>0$, then it has a small chance lower than $\frac{1}{2}$ to increase, which is what is needed to go to the last state, and a much higher chance to decrease its state. Thus most likely, the state will decrease and after several iterations we will find ourselves at the smallest state $0$. This implies that reaching the absorbing state from any other state should be on average of the same order of magnitude as to reach it from the smallest state. 

In order to alleviate this problem, possible solutions would be to work with less worst case scenarios. Indeed, the geometry is not arbitrary when the state number becomes high as during the iterations the geometry empirically becomes more and more biased towards convergence. Thus we should be able to model the average behaviour using a similar Markov chain using the same graph but with probabilities significantly higher. They cannot all be higher than $\frac{1}{2}$ as we would then get an at most linear expected time to reach the absorbing state, which is too fast empirically for $\mathbb{E}(T_{HD})$. Thus, we suspect we would have some states with smaller than $\frac{1}{2}$ probability to increase and the others with higher probability. Another option is to consider probabilities that evolve with time, and thus the transition probabilities would not be constant, although the graph would. For small time steps, the average behaviour should be biased by the geometry and so the increase probabilities should be small. But for higher time steps, the geometry should be biased towards convergence and a half-disk configuration, thus the probabilities would increase and become higher than $\frac{1}{2}$. For both possibilities, the combination of state increase weights smaller and larger than $\frac{1}{2}$ might give us a pseudo-linear expected time to reach the absorbing state $O(N\log N)$ which is what we desire based on empirical observations.

\section{Numerical results}

\subsection{Empirical dependency on the number of agents in the one dimensional case}
\label{sm: 1D emp reg dep on N}
When performing one dimensional linear regression on $\hat{T}_{\varepsilon}$ versus $-\ln\varepsilon$, we find that $\hat{T}_\varepsilon$ can be accurately modelled as $-3g_N \ln\varepsilon + e_N$ where $g_N$ and $e_N$ are constants with respect to $\varepsilon$ depending on $N$. We plot the regressed curves along with the empirical ones for $\hat{T}_{\varepsilon}$ for each tested $N$ in \cref{fig:1D mean cv minus log eps with regression}. 

Since $g_N$ seems to be linear with respect to $N$, we define $c_N = \frac{g_N}{N}$ and study this quantity instead. We plot $c_N$ versus $N$ in \cref{fig:1D c_N evolution of N} and find that $c_N\le 1$ and converges to $1$ towards infinity. Furthermore, for $c_N\ge 100$, $c_N\ge 0.95$, thus yielding the approximation $c_N\approx1$ for $N\ge 100$. 

On the other hand, based on our bound in \cref{th:1D finite expected time expected L0}, we believe $e_N$ should be modelled in the following way: $e_N = a N\ln{N} + b N + f$ where $a$, $b$, and $f$ are constants independent of $N$. Furthermore, when plotting $\frac{e_N}{N\ln{N}}$ versus $N$ in \cref{fig:1D e_N_over_NlnN_evol_N}, we see that $\frac{e_N}{N\ln{N}}$ is asymptotically bounded and thus that the dominant term in $e_N$ should be of magnitude $N\ln{N}$, confirming that our model for $e_N$ is reasonable. Using regression for this model, we find that $e_N \approx \frac{3}{2}\left(0.59 N\ln{N} -1.5N + 3.9\right) \approx 0.89 N\ln{N} - 2.3N + 5.8$. We plot the regressed curve along with the empirical one for $e_N$ in \cref{fig:1D e_N_evol_N_withRegression}. In summary, $\hat{T}_{\varepsilon} \approx -3 c_N  N\ln\varepsilon + 0.89 N\ln{N}  -2.3 N + 5.8$.

Finally, we plot $\hat{T}_{\varepsilon}$ versus $N$ and $N\ln{N}$ in \cref{fig:1D cv N}. We find confirmation that asymptotically $\hat{T}_{\varepsilon}$ is quasi-linear, as predicted by the derived bound and the conjectured empirical regression models.

\begin{figure}[tbhp]
    \centering
    \subfloat[]{\label{fig:1D mean cv minus log eps with regression}\includegraphics[width=0.49\textwidth]{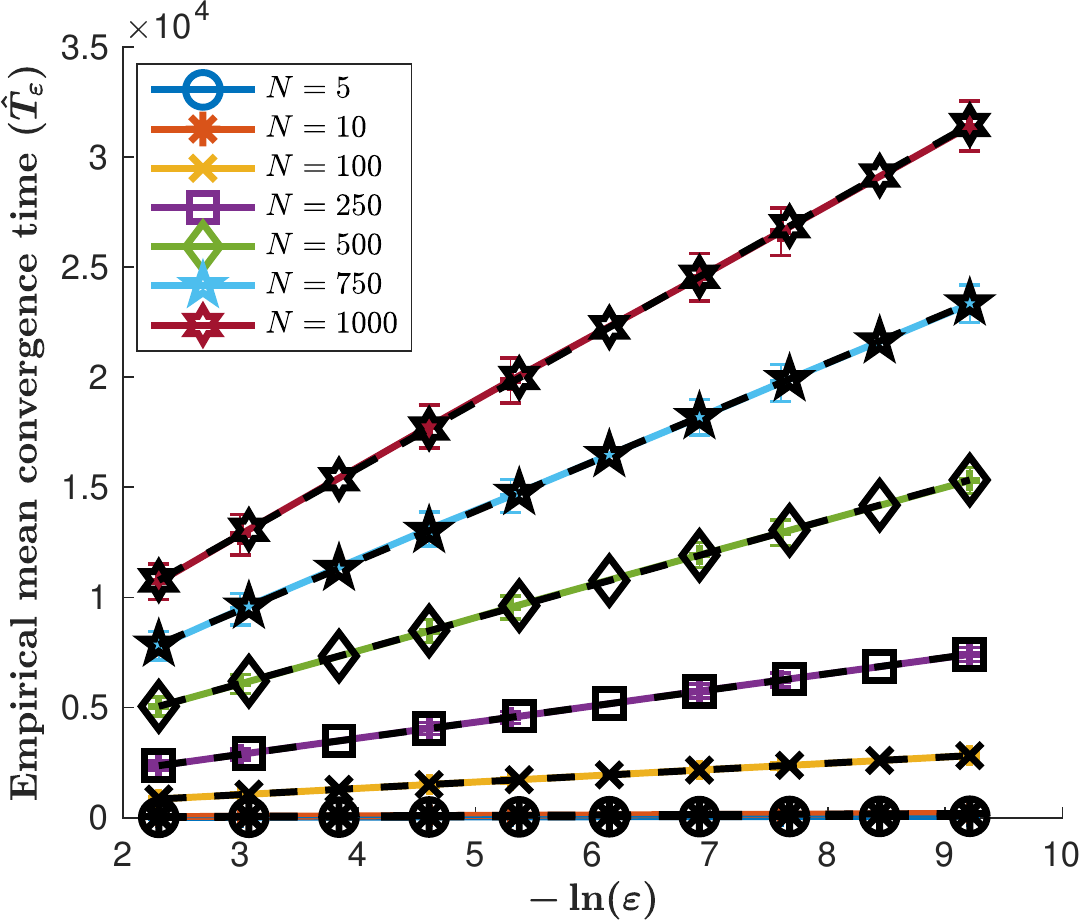}}
   \subfloat[]{\label{fig:1D c_N evolution of N}\includegraphics[width=0.49\textwidth]{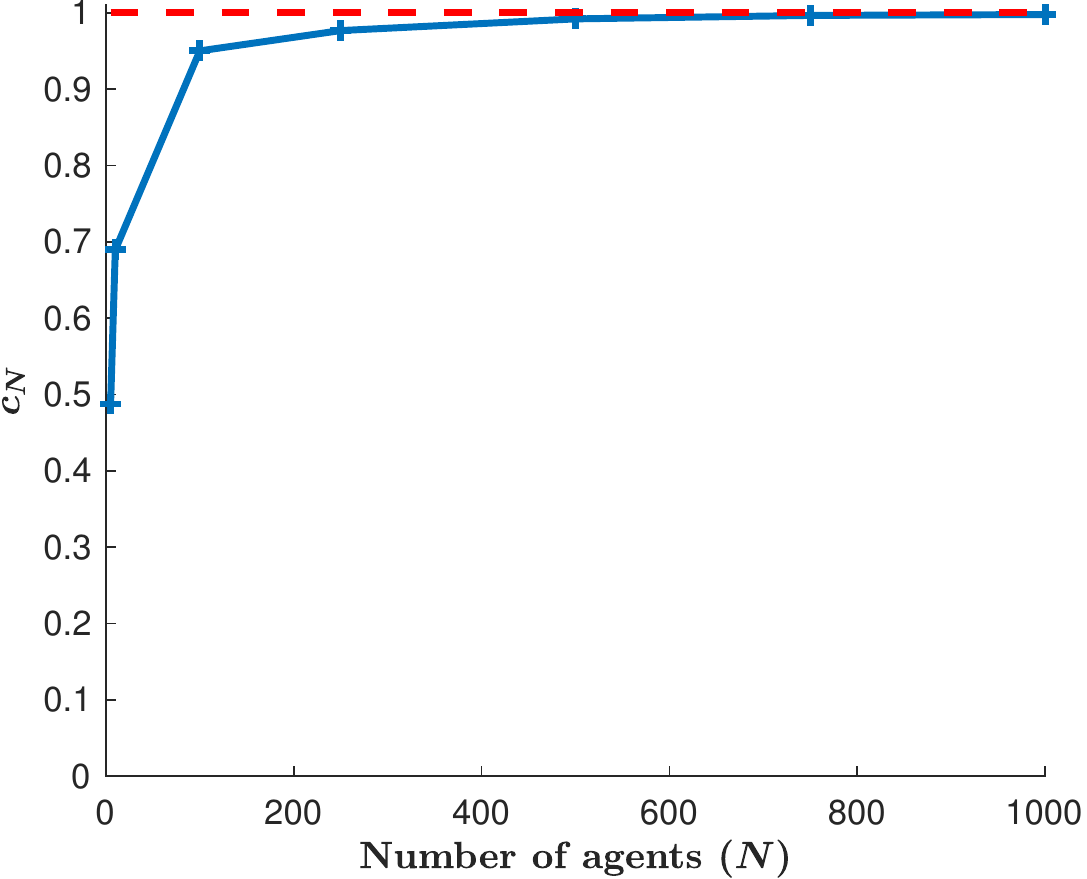}}\\
    \caption{One dimensional evolution: study of the linear part of the modelled dependency of the convergence time on the threshold level $\varepsilon$. Left: we superimpose on the empirical convergence time $\hat{T}_\varepsilon$ the regressed modelled one $-3g_N \ln\varepsilon + e_N$ in dashed black. Right: evolution of the regressed model coefficient $c_N = \frac{g_N}{N}$ with respect to $N$.}
    \label{fig:1D cv eps regression}
\end{figure}

\begin{figure}[tbhp]
    \centering
    \subfloat[]{\label{fig:1D e_N_over_NlnN_evol_N}\includegraphics[width=0.49\textwidth]{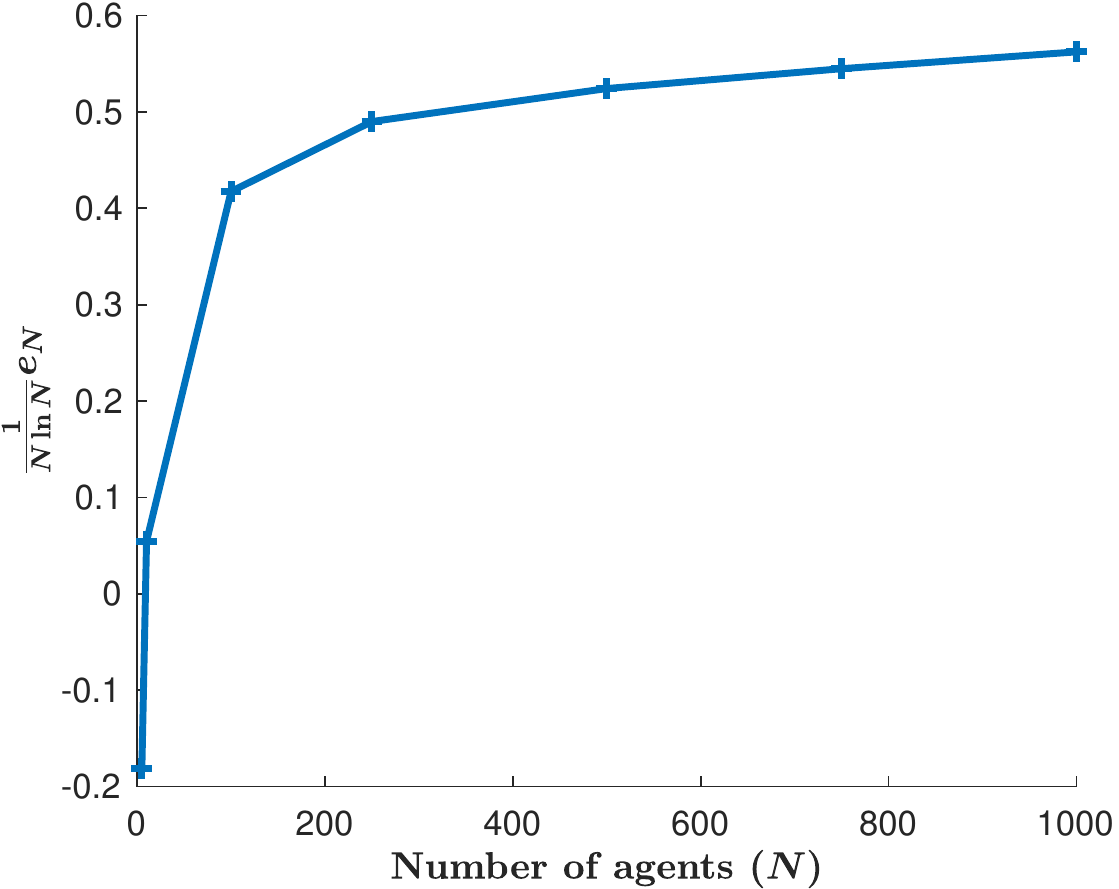}}
    \subfloat[]{\label{fig:1D e_N_evol_N_withRegression}\includegraphics[width=0.49\textwidth]{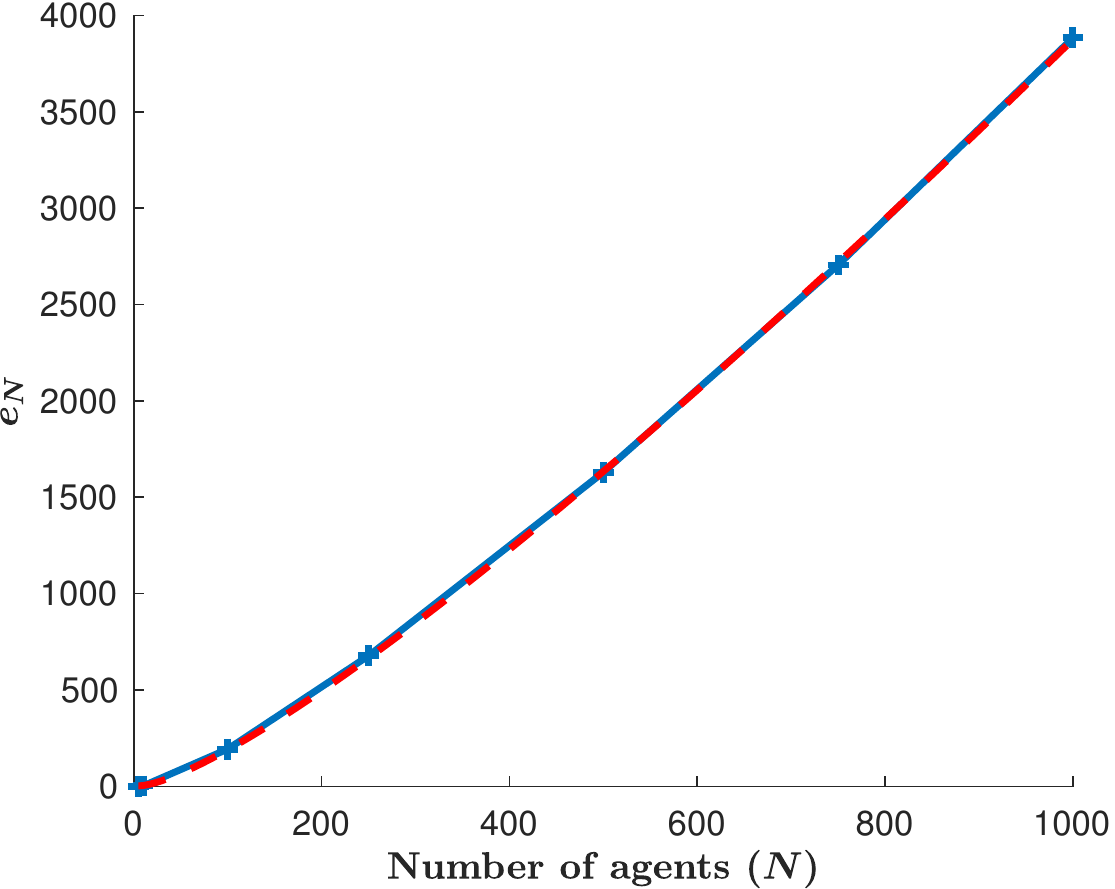}}\\
    \caption{One dimensional evolution: study of the offset part of the modelled dependency of the convergence time on the threshold level $\varepsilon$. Left: evolution of $\frac{e_N}{N\ln{N}}$ with respect to $N$. Right: evolution of $e_N$ with respect to $N$ on which we superimpose the regressed model $aN\ln{N}+bN+f$.}
    \label{fig:1D e_N evol and regression}
\end{figure}

\begin{figure}[tbhp]
    \centering
    \subfloat{\label{fig:1D mean cv N with bounds full 4 curves}\includegraphics[width=0.49\textwidth]{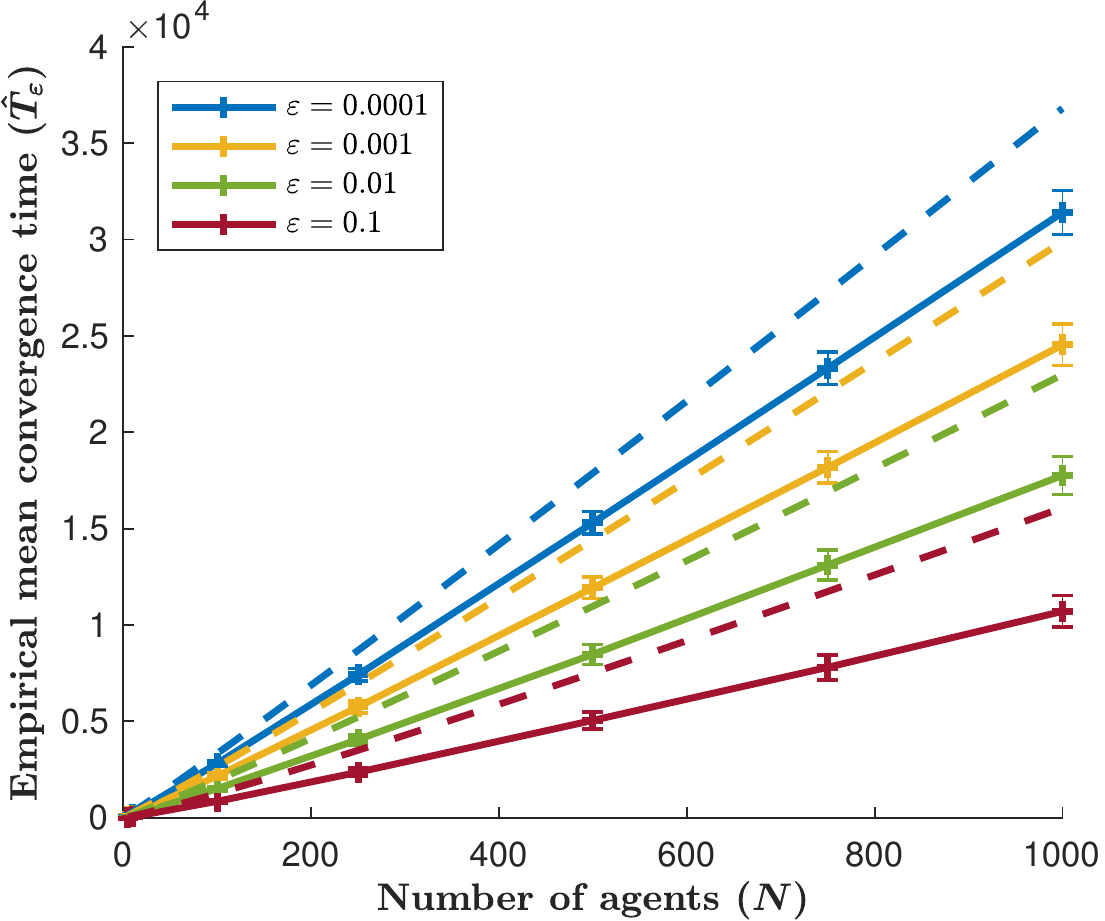}}
    \subfloat{\label{fig:1D mean cv N ln N with bounds full 4 curves}\includegraphics[width=0.49\textwidth]{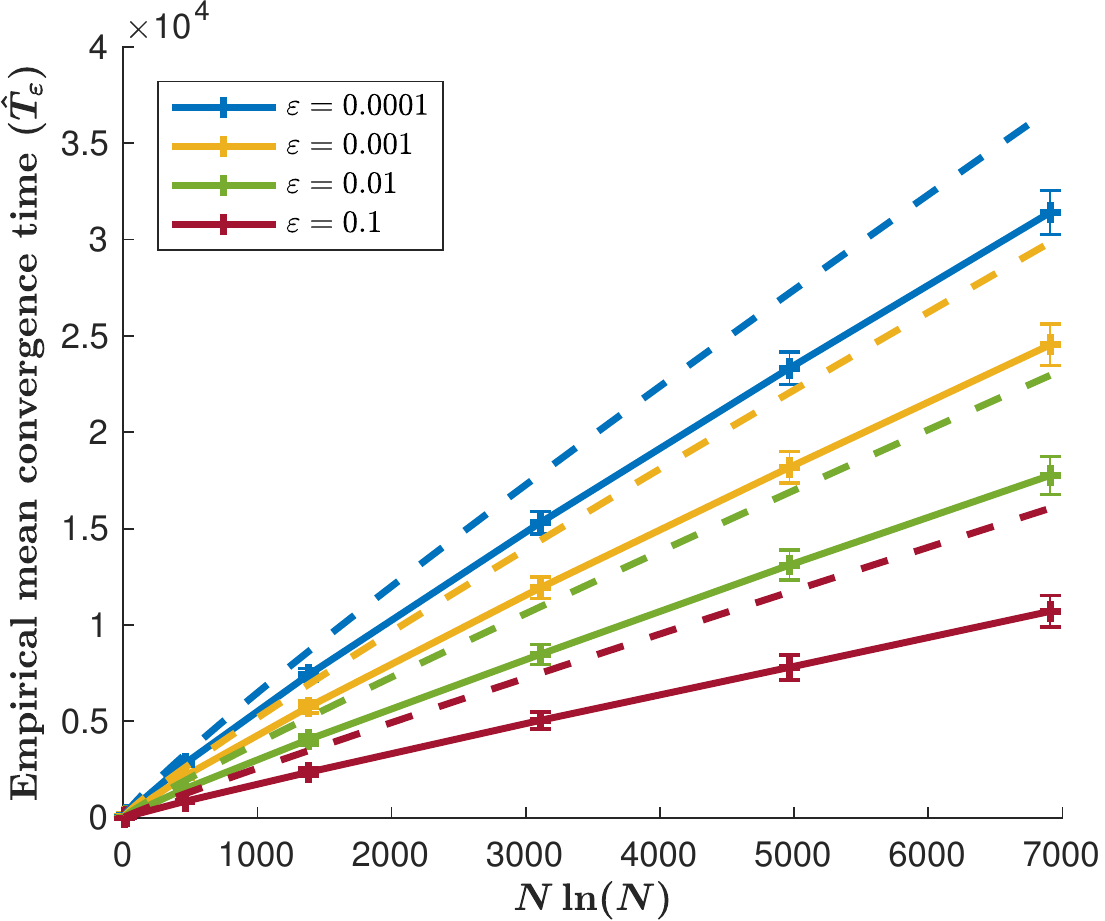}}\\    
    \subfloat{\label{fig:1D mean cv N with bounds full}\includegraphics[width=0.49\textwidth]{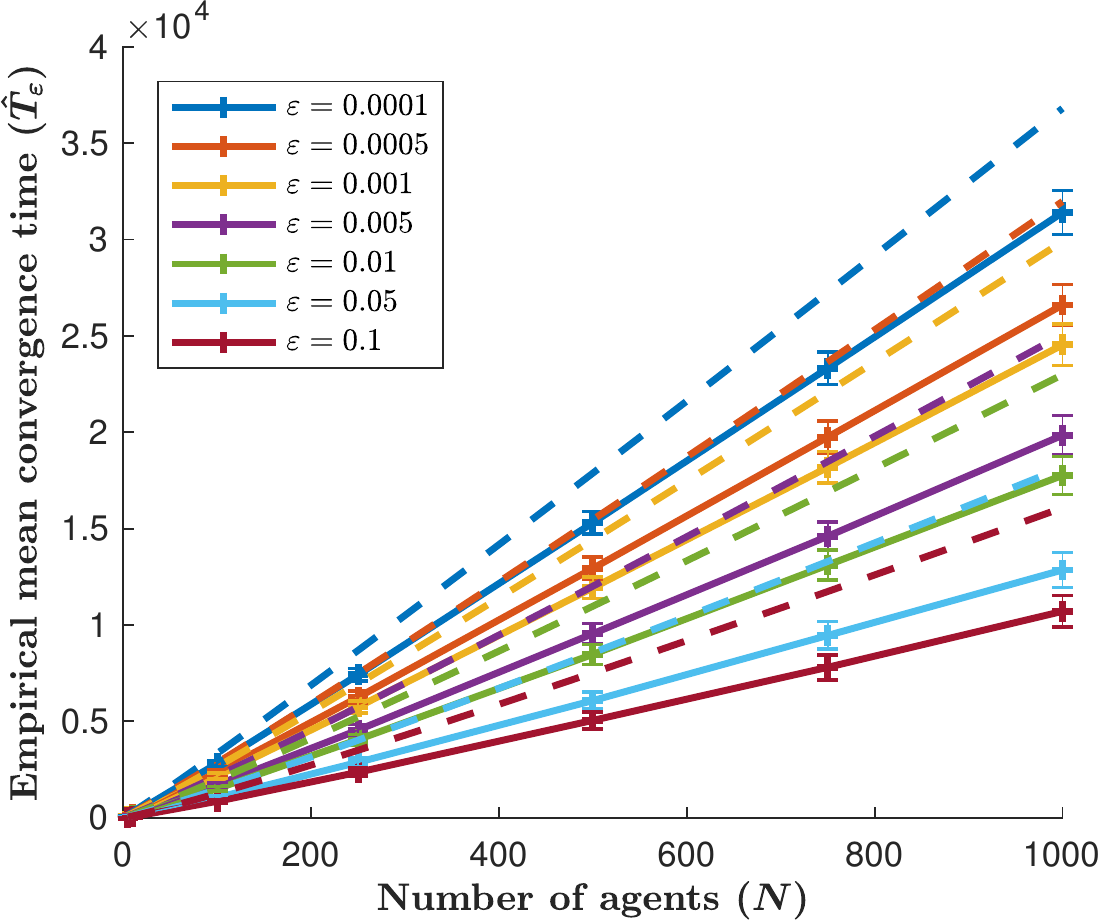}}
    \subfloat{\label{fig:1D mean cv N ln N with bounds full}\includegraphics[width=0.49\textwidth]{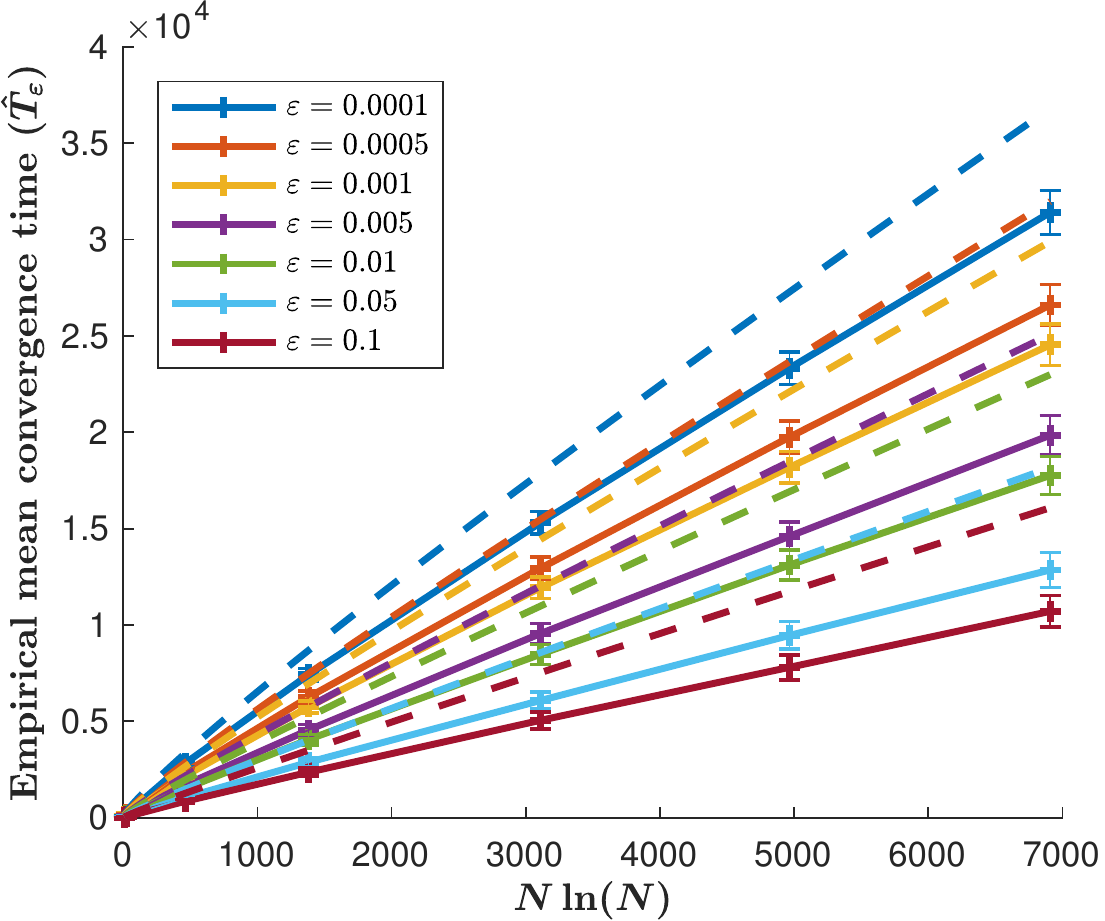}}\\    
    \caption{One dimensional evolution: dependency of the empirical mean convergence time on the number of agents $N$.
    Top: for figure clarity, the results for only four different tested values of $\varepsilon$. Bottom: results for all tested $\varepsilon$. Left: $N$ abscissa. Right: $N\ln N$ abscissa. The plain curves correspond to the empirical results whereas the dashed ones correspond to the theoretical bounds. We superimpose on the empirical curves the traditional unbiased estimator of the standard deviation of each data point.}
    \label{fig:1D cv N}
\end{figure}

\subsection{Empirical dependency on the number of agents in the unconstrained \texorpdfstring{\boldmath$D$}{D}-dimensional case}
\label{sm: ND emp reg dep on N}
Similarly to the one dimensional case, we can accurately model the convergence time as $\hat{T}_\varepsilon\approx -3g_{N,D}\ln{\varepsilon} + e_{N,D}$, as shown in \cref{fig:ND cv eps regression}. However, when comparing with the derived bound, $D$ does not influence the bound with respect to $\varepsilon$ in any other way than by adding a constant offset. Thus $g_{N,D}$ should not depend on $D$, and we will thus model $\hat{T}_\varepsilon\approx -3g_N\ln{\varepsilon} + e_{N,D}$. Since, $g_{N,D}$ no longer depends on $D$, then we should have $g_{N,D} = c_{N,D} N = c_N N$ where $c_N$ is from the one dimensional case. We find this result empirically when plotting $c_{N,D}$ for tested $D$ in \cref{fig: all_d_c_N_evol_N}.

Unfortunately, due to computation and time limitations, we were not able to simulate results for high $D$, therefore we cannot accurately measure how  $e_{N,D}$ evolves with $D$. In order to avoid overfitting, we did not further analyse the dependency of $e_{N,D}$ with respect to $D$. However, for each $D$, we can accurately model it as $e_{N,D} \approx a_D N\ln{N} +b_D N +f_D$ where $a_D$, $b_D$, and $f_D$ depend only on $D$. Due to their unknown dependency on $D$, we do not provide the regression coefficients, but do plot the overall regression functions and the evolution of $\frac{e_{N,D}}{N\ln{N}}$ in \cref{fig:ND e_N evol and regression}. However if the behaviour is similar to that in the bounds, then only $b_D$ should depend on $D$ asymptotically, and the dependency would be logarithmic.

Finally, we plot $\hat{T}_{\varepsilon}$ versus $N$ and $N\ln{N}$ in \cref{fig:ND cv N}. We find confirmation that asymptotically $\hat{T}_{\varepsilon}$ is quasi-linear, as predicted by the derived bound and the conjectured empirical model.

\begin{figure}[tbhp]
    \centering
    \subfloat[$D=2$]{\label{fig:2D mean cv minus log eps with regression}\includegraphics[width=0.33\textwidth]{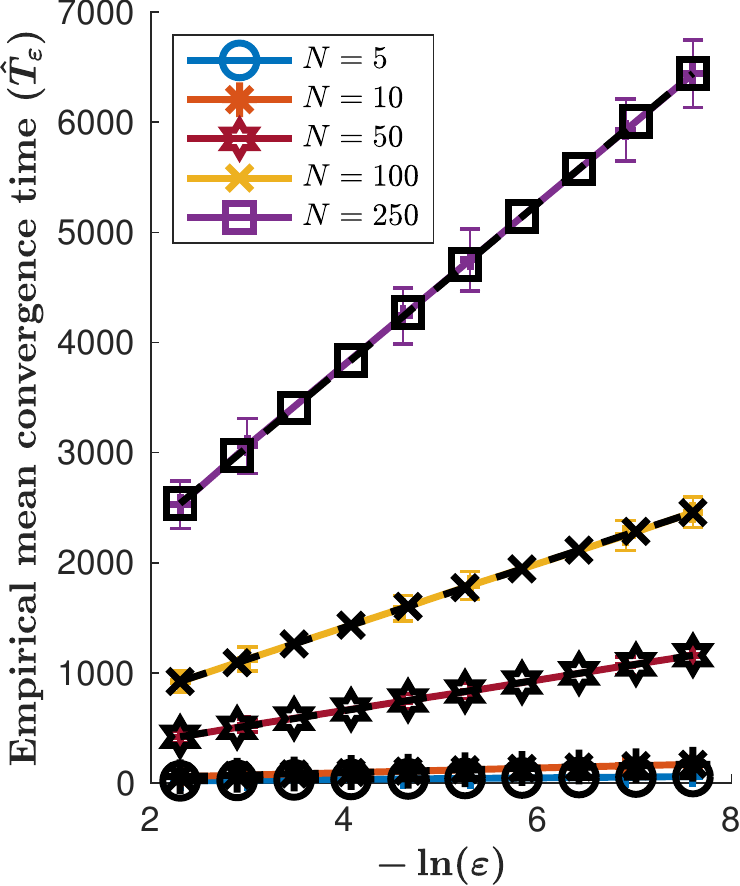}}
    \subfloat[$D = 3$]{\label{fig:3D mean cv minus log eps with regression}\includegraphics[width=0.33\textwidth]{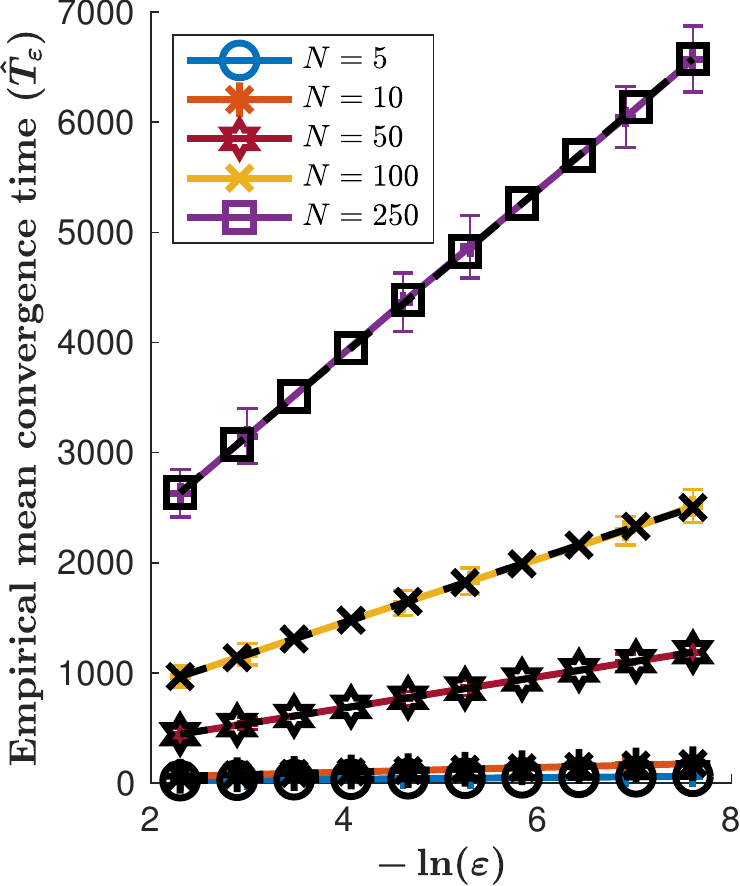}}
    \subfloat[$D = 4$]{\label{fig:4D mean cv minus log eps with regression}\includegraphics[width=0.33\textwidth]{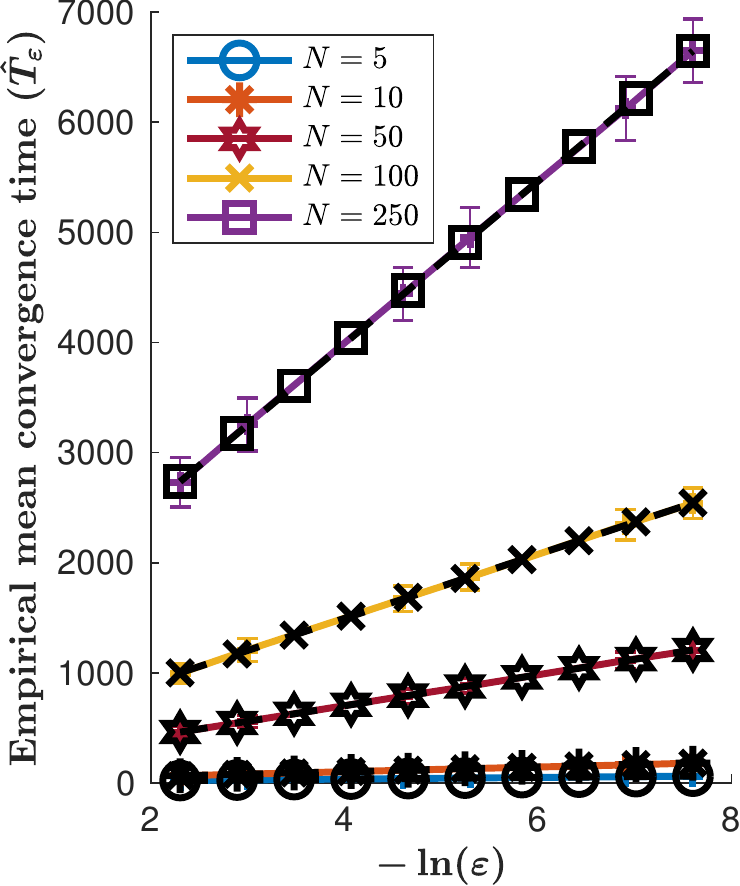}}\\
    \caption{$D$-dimensional evolution: we superimpose on the empirical convergence time $\hat{T}_\varepsilon$ the regressed modelled one $-3g_{N,D} \ln\varepsilon + e_{N,D}$ in dashed black.}
    \label{fig:ND cv eps regression}
\end{figure}

\begin{figure}[tbhp]
  \centering
    \includegraphics[width=0.5\textwidth]{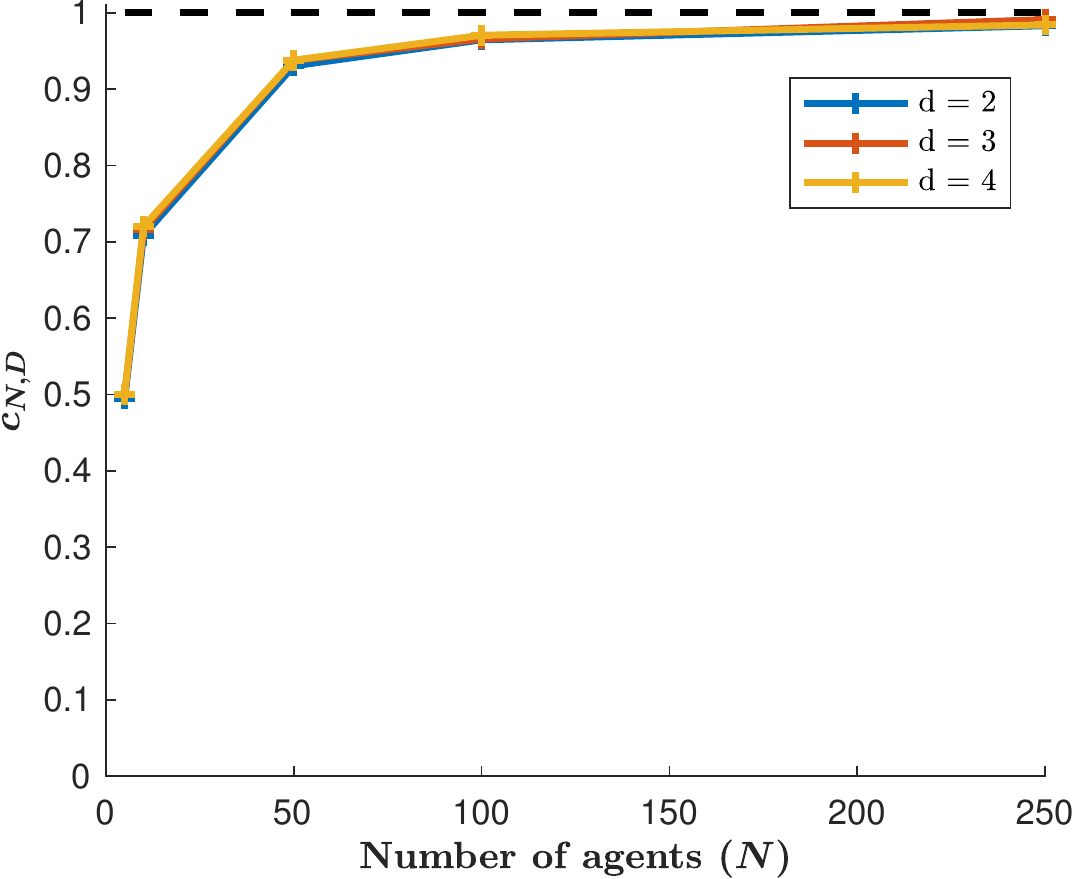}
    \caption{$D$-dimensional evolution: Evolution of the regressed model coefficient $c_{N,D} = \frac{g_{N,D}}{N}$ with respect to $N$. Clearly $c_{N,D}$ is independent from $D$.}
    \label{fig: all_d_c_N_evol_N}
\end{figure}

\begin{figure}[tbhp]
    \centering
    \subfloat{\label{fig:2D e_N_over_NlnN_evol_N}\includegraphics[width=0.33\textwidth]{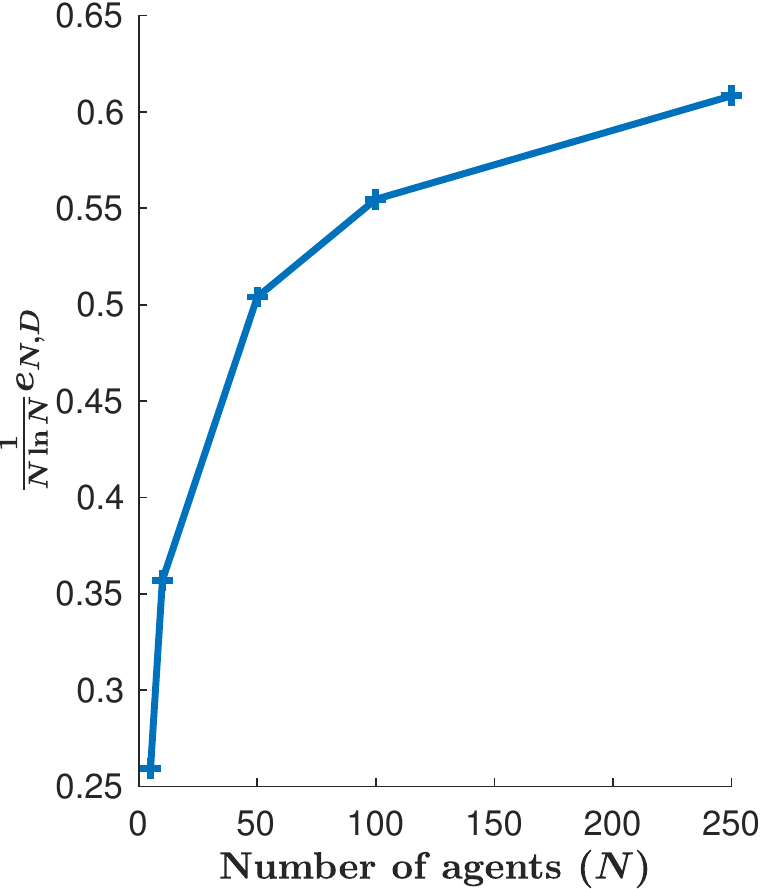}}
    \subfloat{\label{fig:3D e_N_over_NlnN_evol_N}\includegraphics[width=0.33\textwidth]{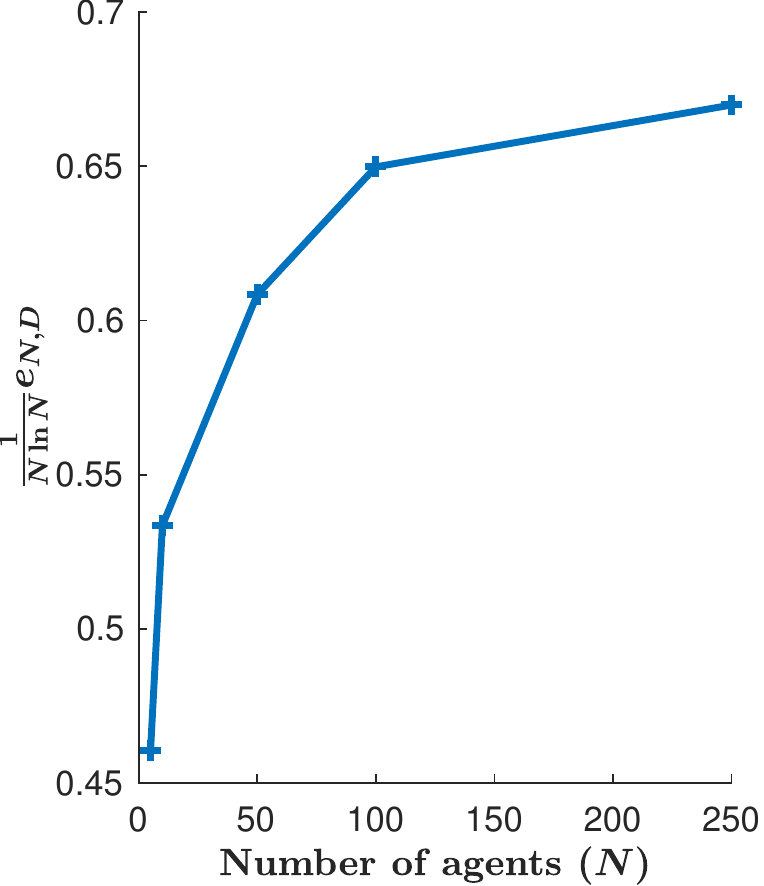}}
    \subfloat{\label{fig:4D e_N_over_NlnN_evol_N}\includegraphics[width=0.33\textwidth]{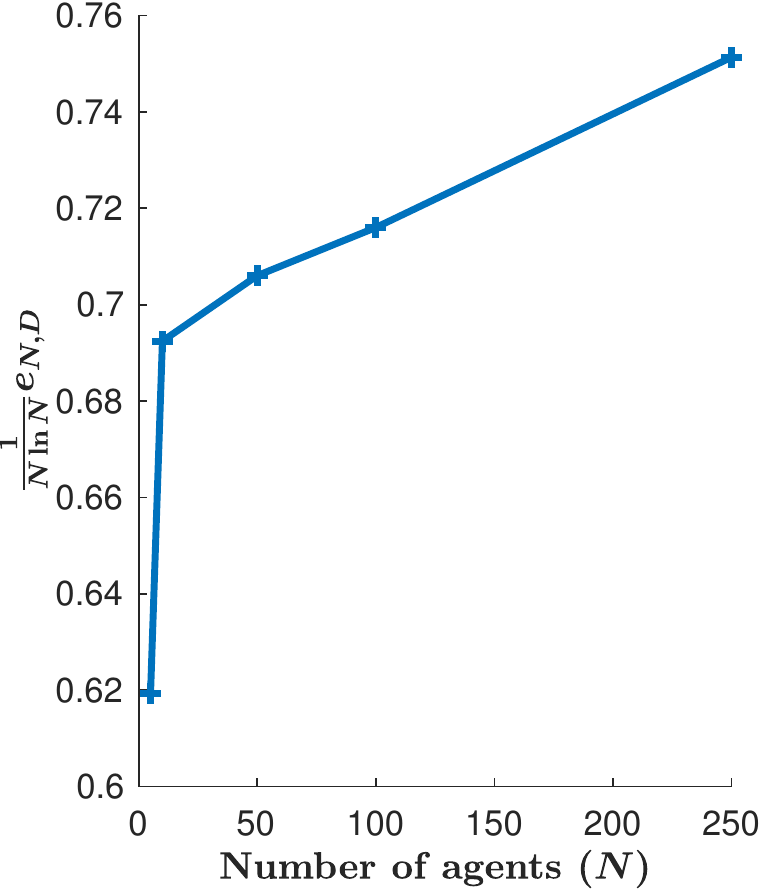}}\\
    \subfloat{\label{fig:2D e_N_evol_N_withRegression}\includegraphics[width=0.33\textwidth]{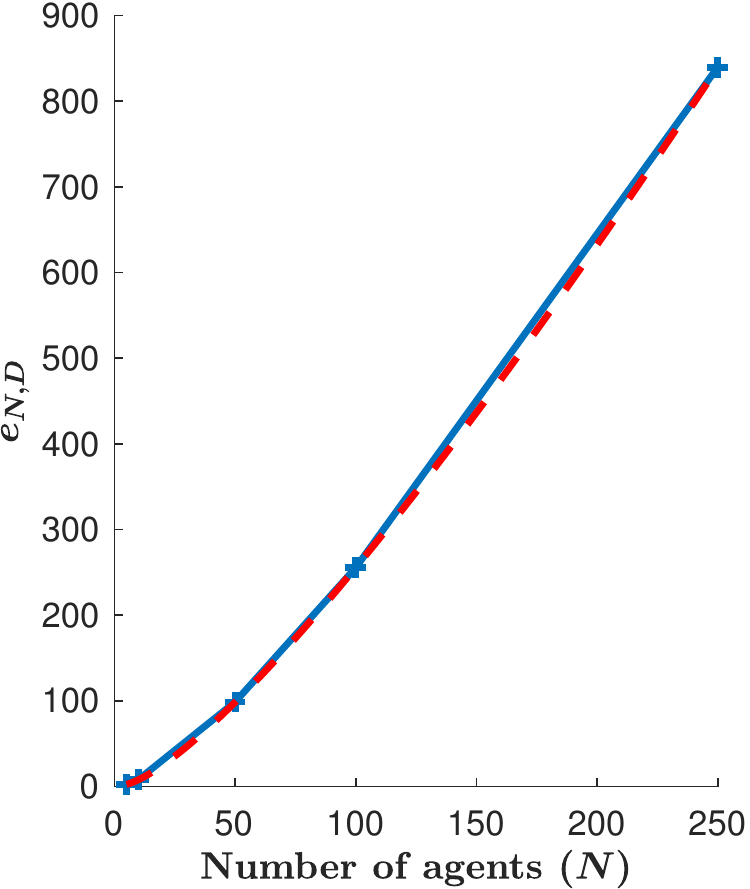}}
    \subfloat{\label{fig:3D e_N_evol_N_withRegression}\includegraphics[width=0.33\textwidth]{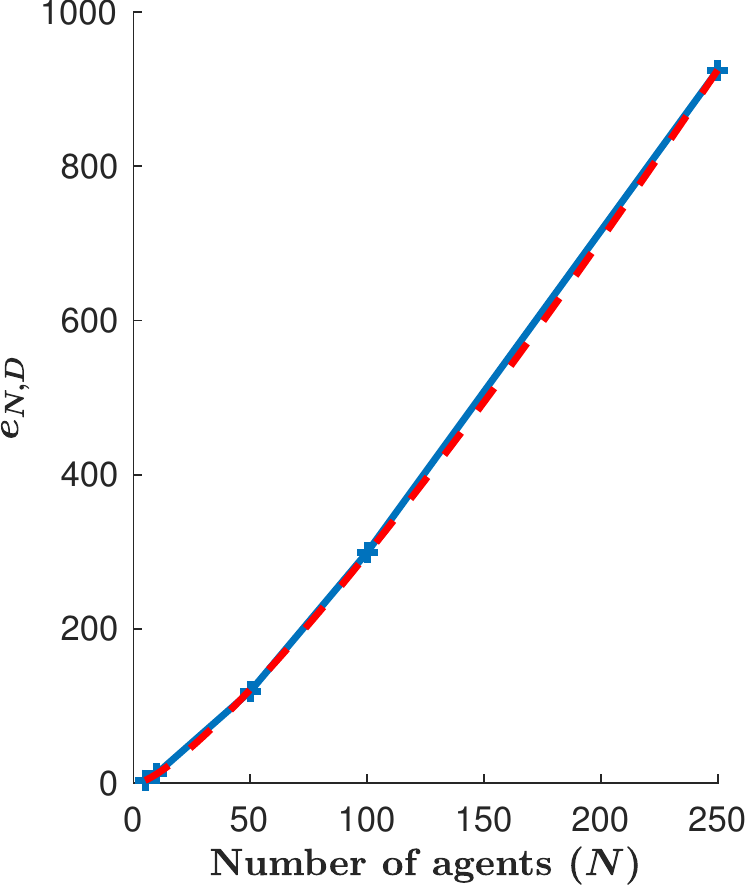}}
    \subfloat{\label{fig:4D e_N_evol_N_withRegression}\includegraphics[width=0.33\textwidth]{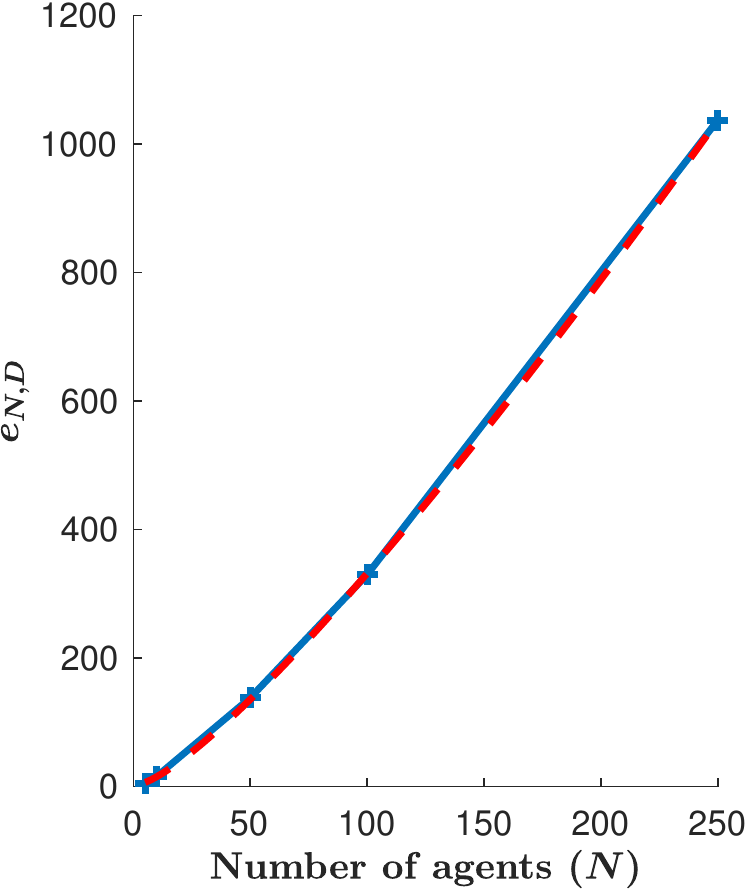}}\\
    \caption{$D$-dimensional evolution: study of the offset part of the modelled dependency of the convergence time on the threshold level $\varepsilon$, in the $2$, $3$, and $4$ dimensional cases from left to right. Top: evolution of $\frac{e_{N,D}}{N\ln{N}}$ with respect to $N$. Bottom: evolution of $e_{N,D}$ with respect to $N$ on which we superimpose the regressed model $a_D N\ln{N}+b_D N+f_D$.}
    \label{fig:ND e_N evol and regression}
\end{figure}

\begin{figure}[tbhp]
    \centering
    \subfloat{\label{fig:2D mean cv N with bounds full 3 curves}\includegraphics[width=0.25\textwidth]{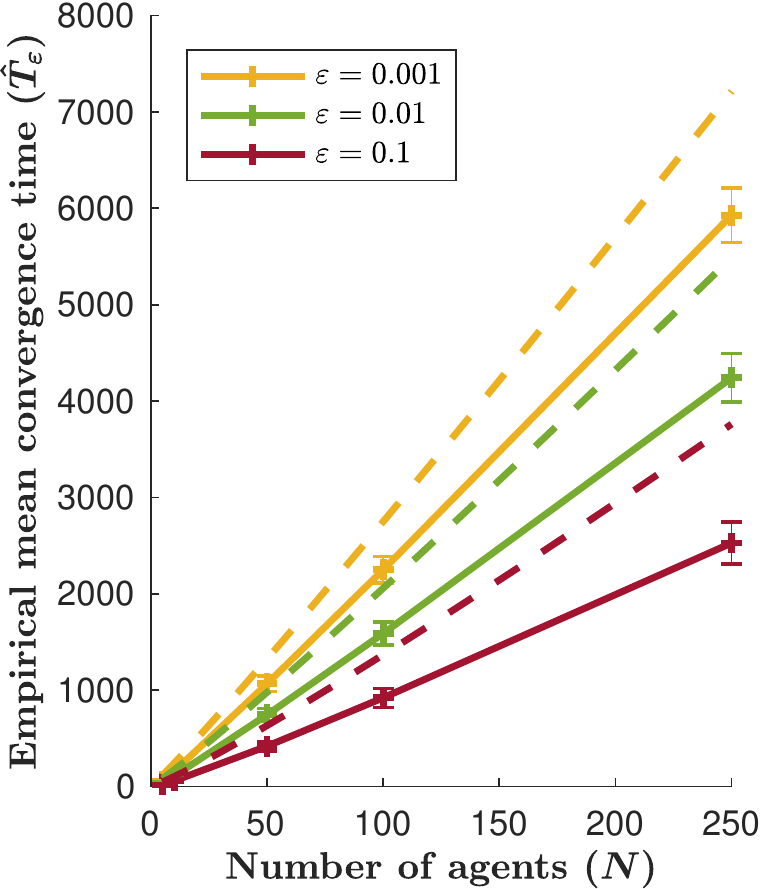}}
     \subfloat{\label{fig:3D mean cv N with bounds full 3 curves}\includegraphics[width=0.25\textwidth]{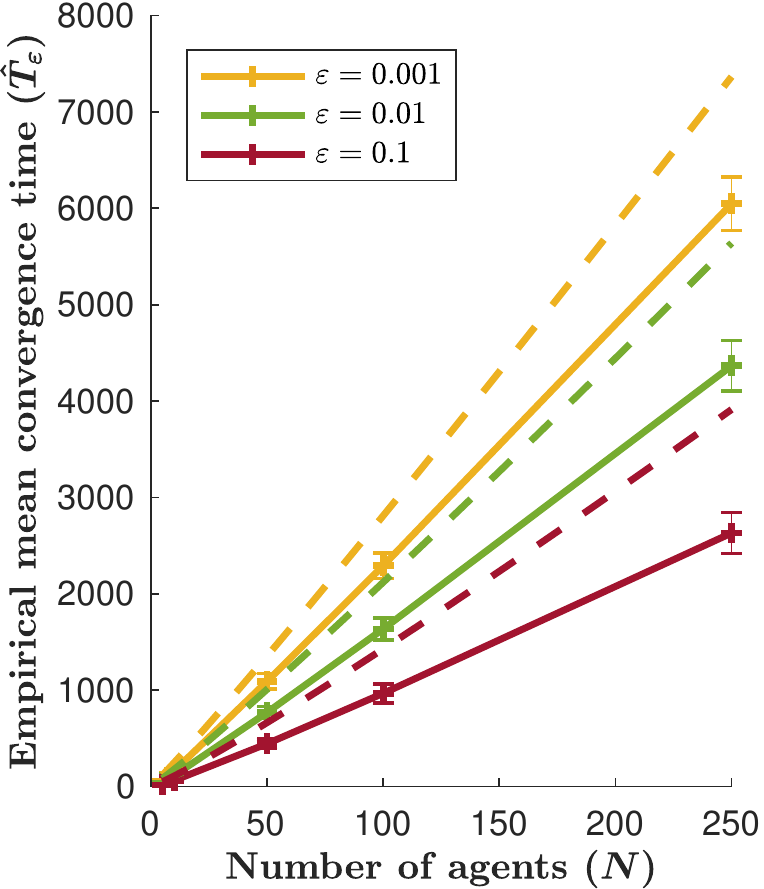}}
      \subfloat{\label{fig:4D mean cv N with bounds full 3 curves}\includegraphics[width=0.25\textwidth]{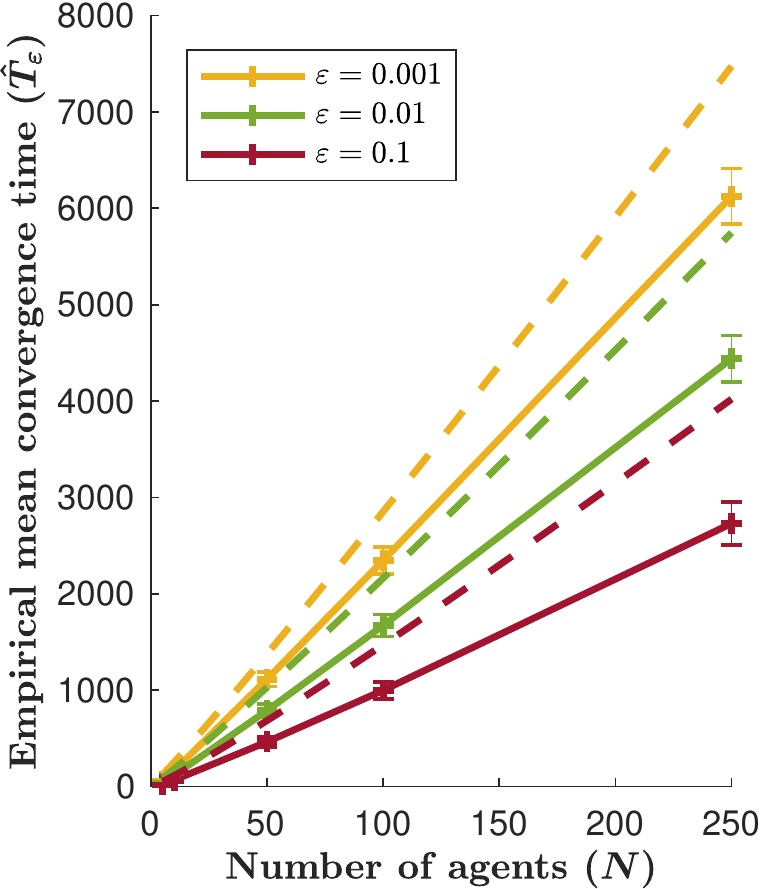}}\\
    \subfloat{\label{fig:2D mean cv N ln N with bounds full 3 curves}\includegraphics[width=0.25\textwidth]{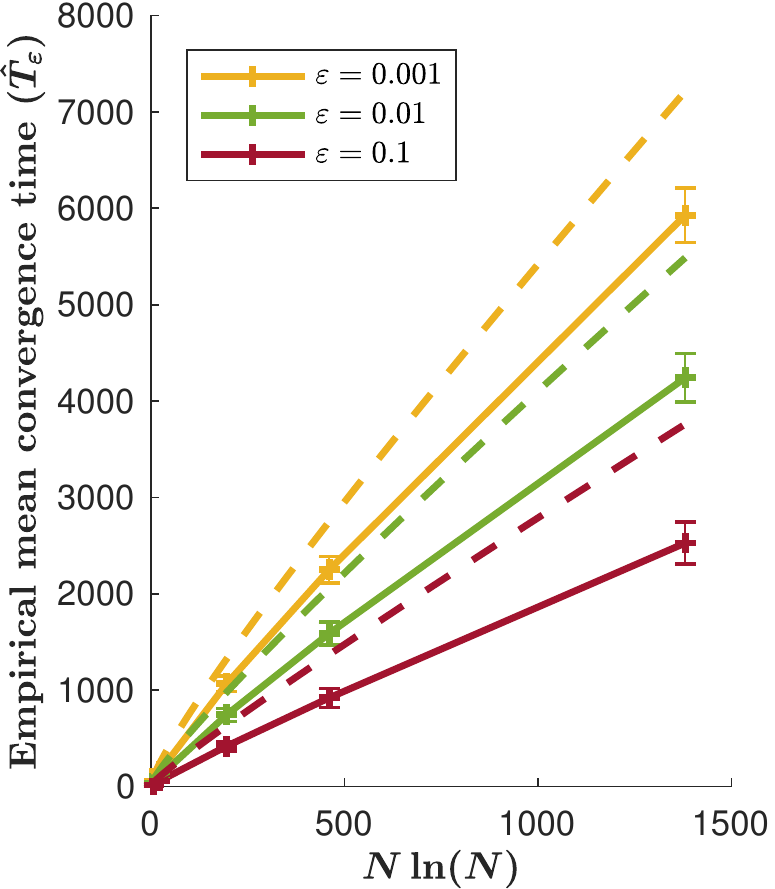}}
    \subfloat{\label{fig:3D mean cv N ln N with bounds full 3 curves}\includegraphics[width=0.25\textwidth]{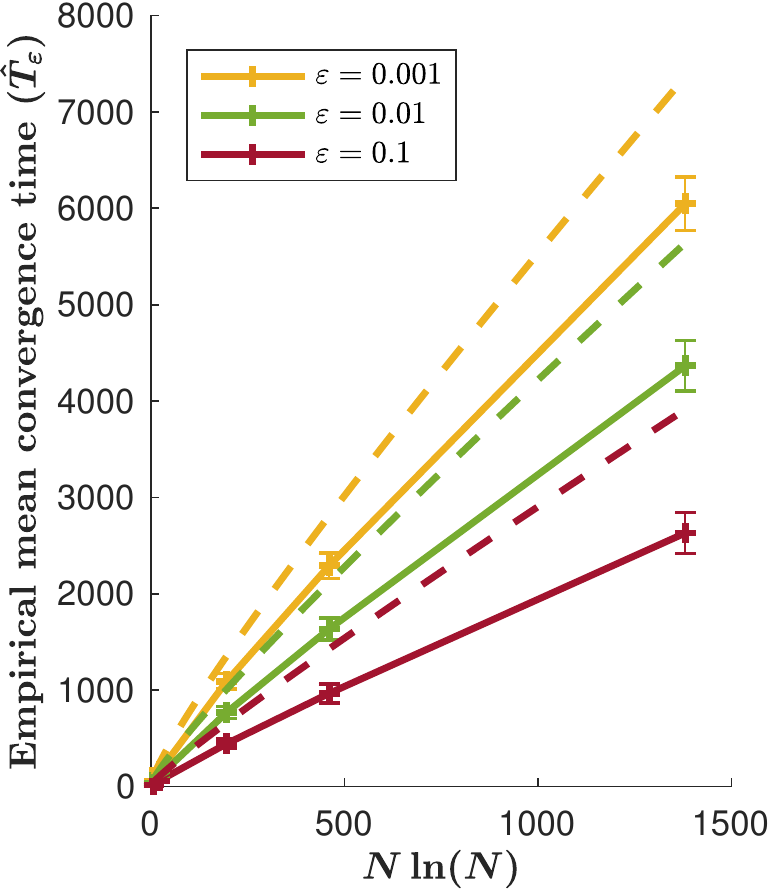}}
    \subfloat{\label{fig:4D mean cv N ln N with bounds full 3 curves}\includegraphics[width=0.25\textwidth]{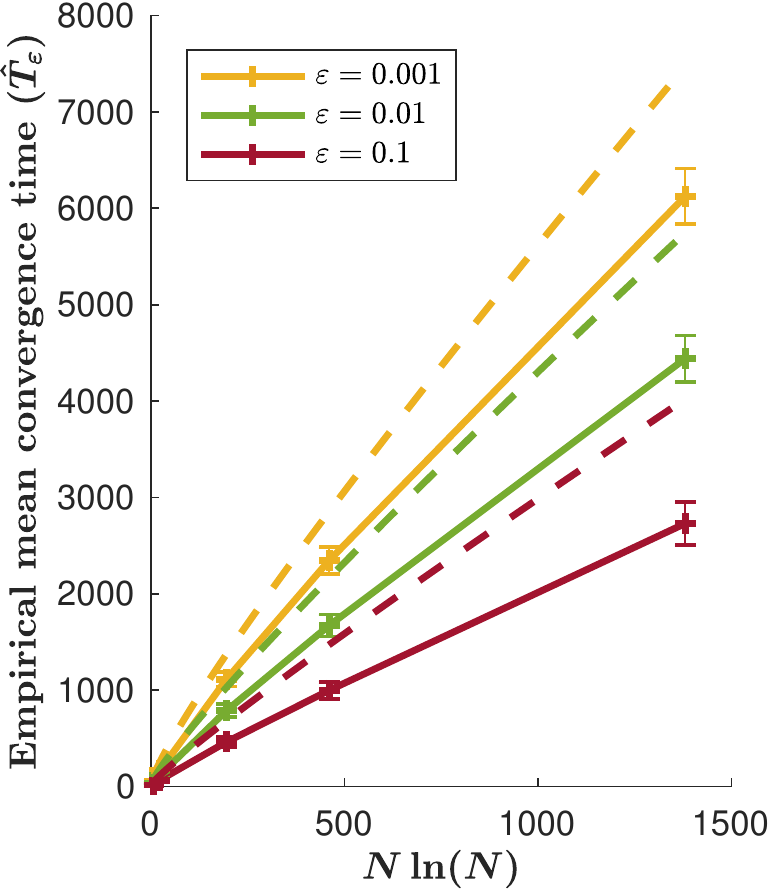}}\\  
    \subfloat{\label{fig:2D mean cv N with bounds full}\includegraphics[width=0.25\textwidth]{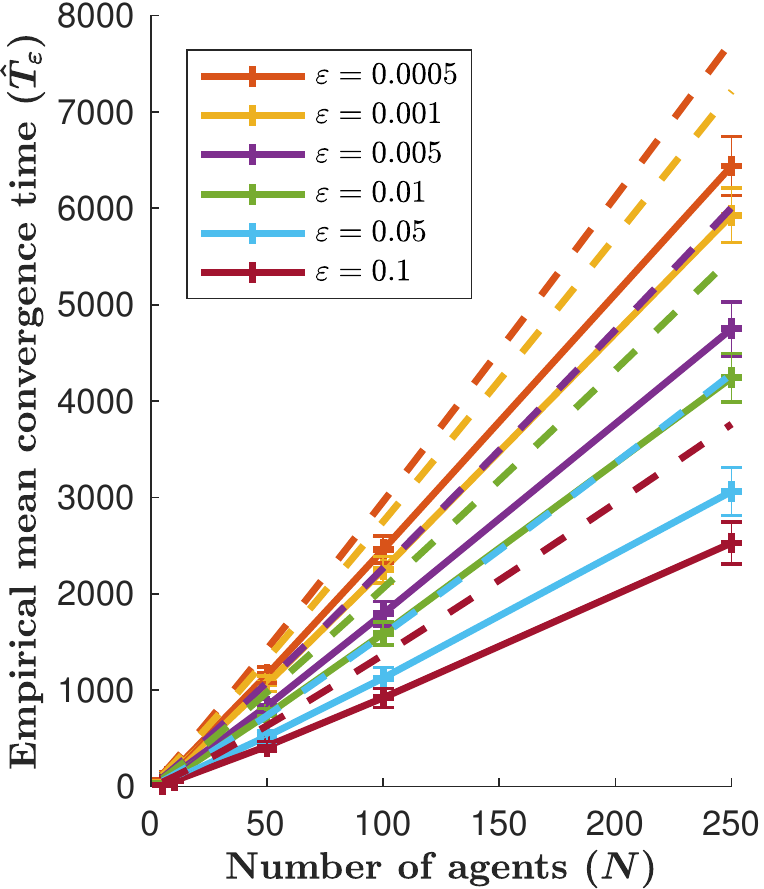}}
    \subfloat{\label{fig:3D mean cv N with bounds full}\includegraphics[width=0.25\textwidth]{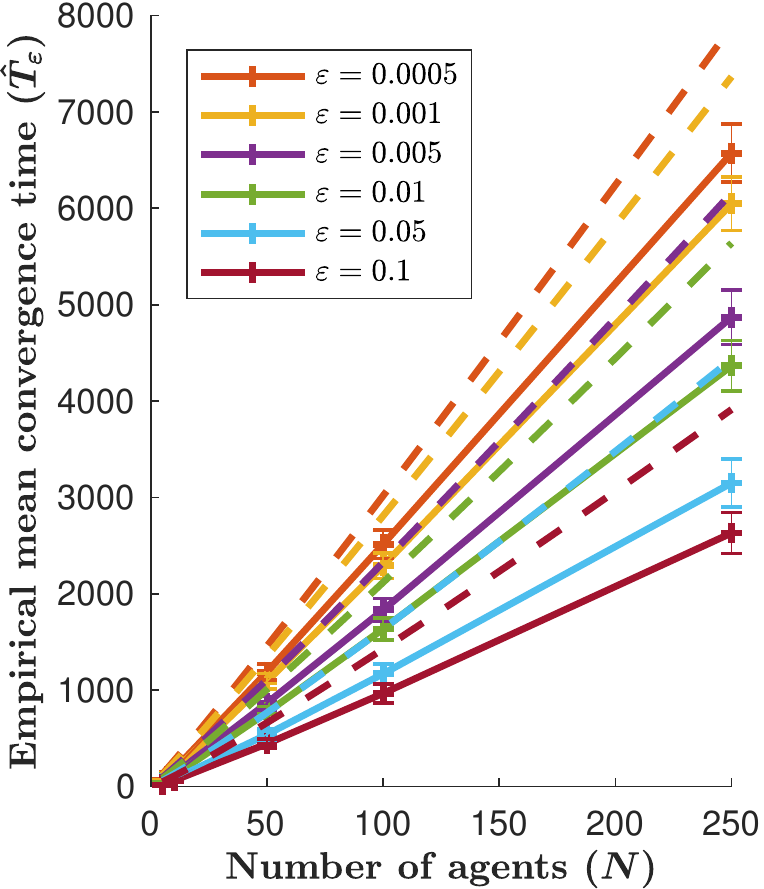}}
    \subfloat{\label{fig:4D mean cv N with bounds full}\includegraphics[width=0.25\textwidth]{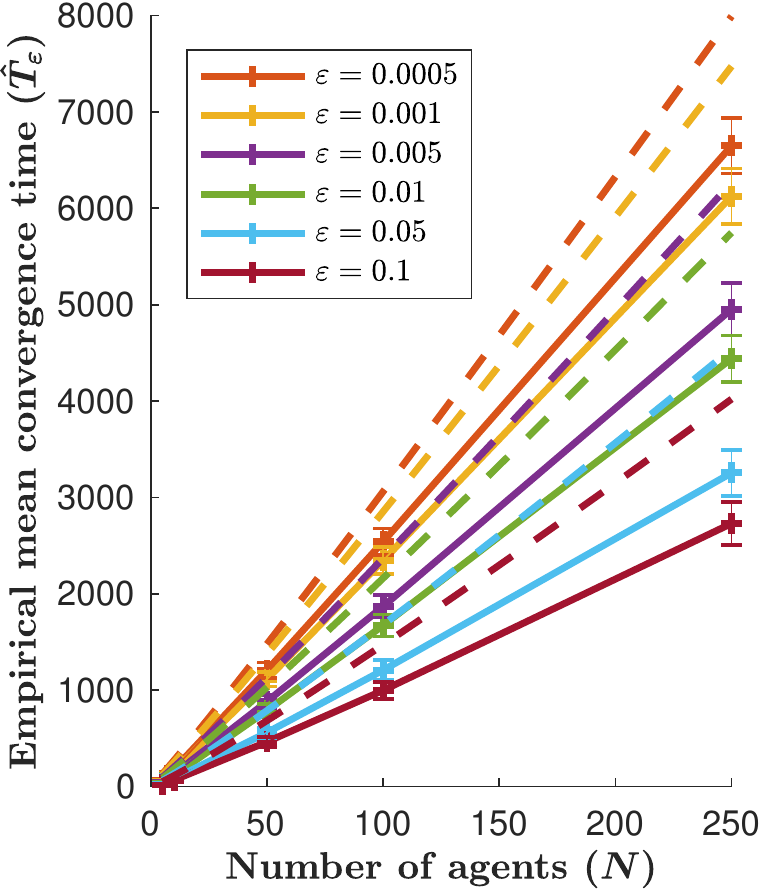}} \\
    \subfloat{\label{fig:2D mean cv N ln N with bounds full}\includegraphics[width=0.25\textwidth]{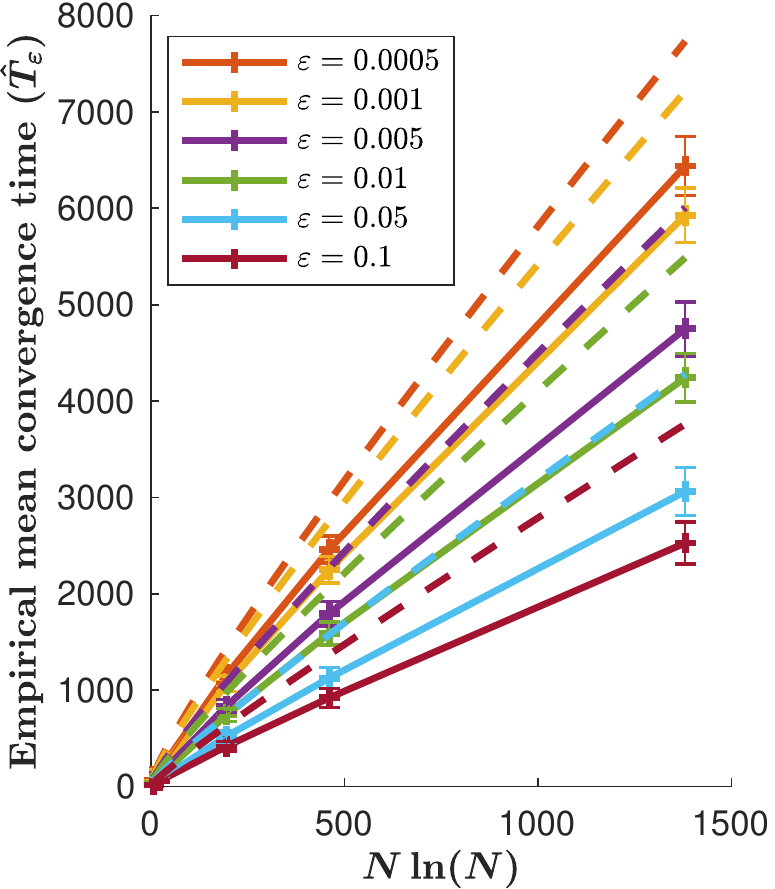}}
    \subfloat{\label{fig:3D mean cv N ln N with bounds full}\includegraphics[width=0.25\textwidth]{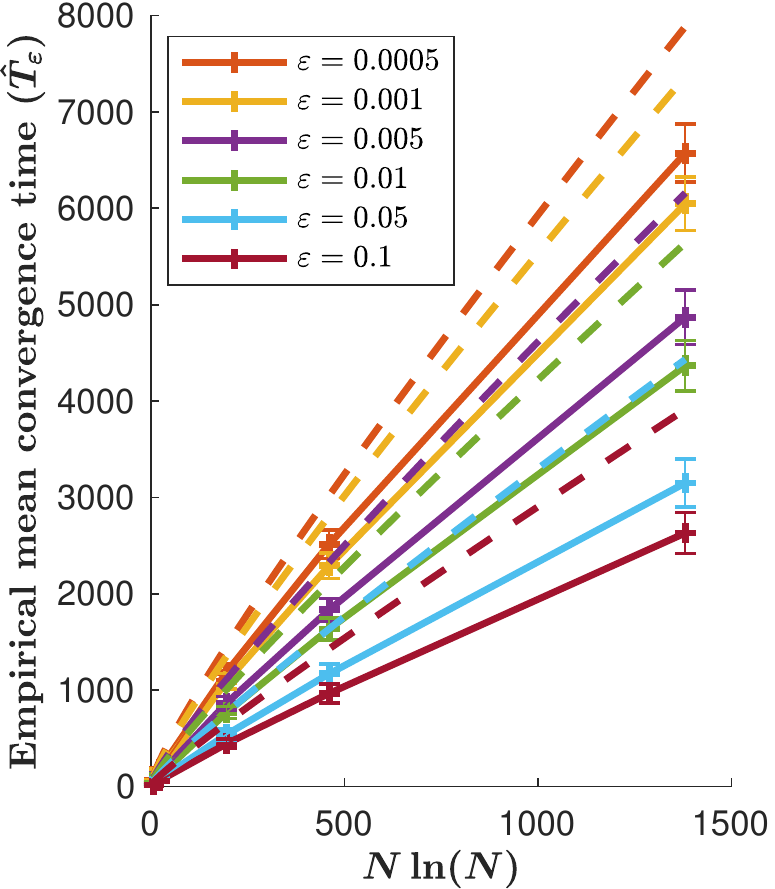}}
    \subfloat{\label{fig:4D mean cv N ln N with bounds full}\includegraphics[width=0.25\textwidth]{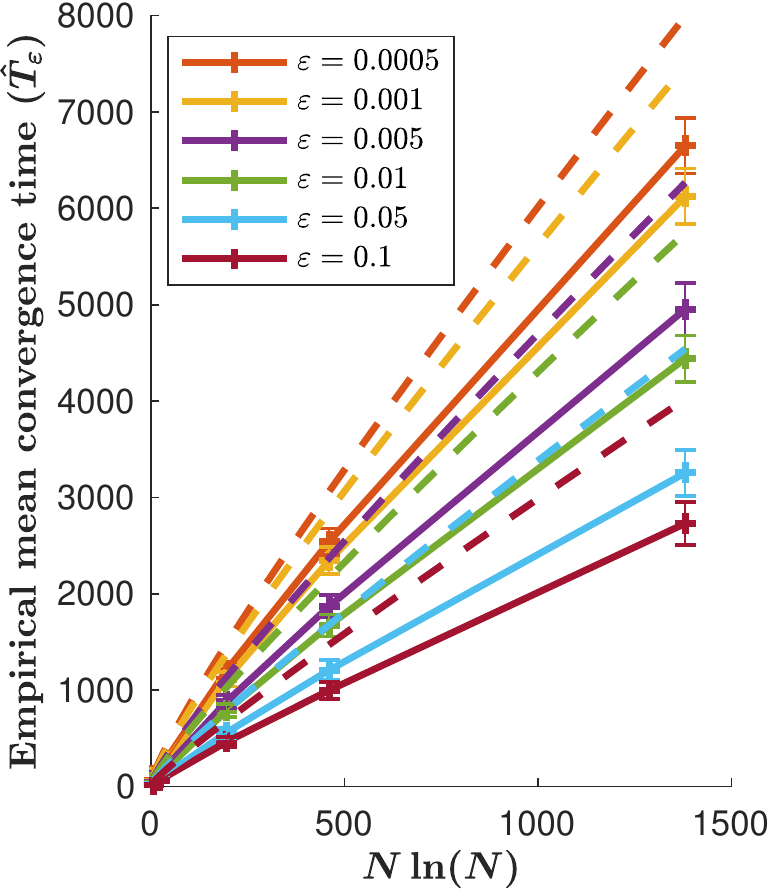}}\\    
    \caption{$D$-dimensional evolution: dependency of the empirical mean convergence time on the number of agents $N$, in the $2$, $3$, and $4$ dimensional cases from left to right. The first two lines are results on only three different tested values of $\varepsilon$ for figure clarity. The bottom two are the same but on all tested $\varepsilon$. The first and third line have $N$ abscissa. The second and fourth line have $N\ln N$ abscissa. The plain curves correspond to the empirical results whereas the dashed ones correspond to the theoretical bounds. We superimpose on the empirical curves the traditional unbiased estimator of the standard deviation of each data point.}
    \label{fig:ND cv N}
\end{figure}

\subsection{Empirical dependency on the number of agents in the constrained 2-dimensional case}
\label{sm: circle emp reg dep on N}
The circle evolution can be decomposed into two parts. First, when $k\le T_{HD}$, not all agents are confined within a half-disk. Second, all agents are within a half disk and then the dynamics are equivalent to those from the one dimensional case. We find that $\hat{T}_\varepsilon$ can be accurately modelled as $\hat{T}_\varepsilon \approx -3g_{N,circle} \ln{\varepsilon} + e_{N,circle} $ where $g_{N,circle}$ and $e_{N,circle}$ do not depend on $\varepsilon$, as shown in the regressed plot \cref{fig:circle mean cv minus log eps with regression}. 

The dependence on $\varepsilon$ is not related to the evolution prior to the half-disk configuration as $T_{HD}$ is independent from $\varepsilon$. Thus the dependence should be fixed by the one given by the one dimensional case. Therefore we expect $c_{N,circle} = \frac{g_{N,circle}}{N} = c_N$ where $c_N$ is given in the one dimensional case. We empirically find this behaviour when performing linear regression, see \cref{fig:circle c_N evolution of N}. 

Empirically, as shown in \cref{fig:circle mean hd N log N}, $\hat{T}_{HD}$ is quasi-linear and can be accurately modelled as, using regression, $\hat{T}_{HD}\approx a_{HD}N\ln{N} + f_{HD} = 0.92 N\ln{N} + 100$. Combining this information with our knowledge from the one dimensional case, we can then accurately model $e_{N,circle} \approx (a_\pi +a_{HD}) N\ln{N} + b_\pi N + c_\pi + f_{HD}$, where $a_\pi$, $b_\pi$, and $c_\pi$ are constants. The subindex $\pi$ is to emphasise that the ``initial conditions'' of the evolutions once we have reached convergence, i.e. at time $T_{HD}$, are not the same as the ones used for the previous results presented in the one dimensional case. First, the maximal span is $\pi$ instead of $1$. Second, this new initial distribution is not necessarily random uniform in an interval of length $\pi$. Using regression for this model, we find that $(a_\pi,b_\pi,f_\pi) \approx (0.25,2.3,-3.2)$, and $e_{N,circle} \approx 0.93N\ln{N}+4.1N-18$. We plot $\frac{e_{N,circle}}{N\ln{N}}$ and $e_{N,circle}$ and its regression in \cref{fig:circle e_N evol and regression}.

Note that $a_\pi < a \approx 0.89$ the equivalent coefficient from the one dimensional case. This suggests that the distribution once we have reached the half-disk configuration is indeed not uniformly random in a (random) interval of size $\pi$ but instead that the distribution is then biased. The following toy example can help understand why this happens. Assume there is one agent that is not within a half-disk shared by all other agents. Then the more we have agents, the more it will take a long time for that agent to be selected and for him to then go closer to those other agents and get into a half-disk configuration. However, while it is not selected, the other agents get closer to each other following the equivalent one dimensional evolution scheme. Therefore when that one agent will finally be selected, the distribution of the other agents will be somewhat already biased to convergence. 

Finally, we plot $\hat{T}_{\varepsilon}$ versus $N$ and $N\ln{N}$ in \cref{fig:circle cv N}. We find confirmation that asymptotically $\hat{T}_{\varepsilon}$ is quasi-linear, as predicted by the derived bound and the conjectured empirical regression models.

\begin{figure}[tbhp]
    \centering
    \subfloat[]{\label{fig:circle mean cv minus log eps with regression}\includegraphics[width=0.49\textwidth]{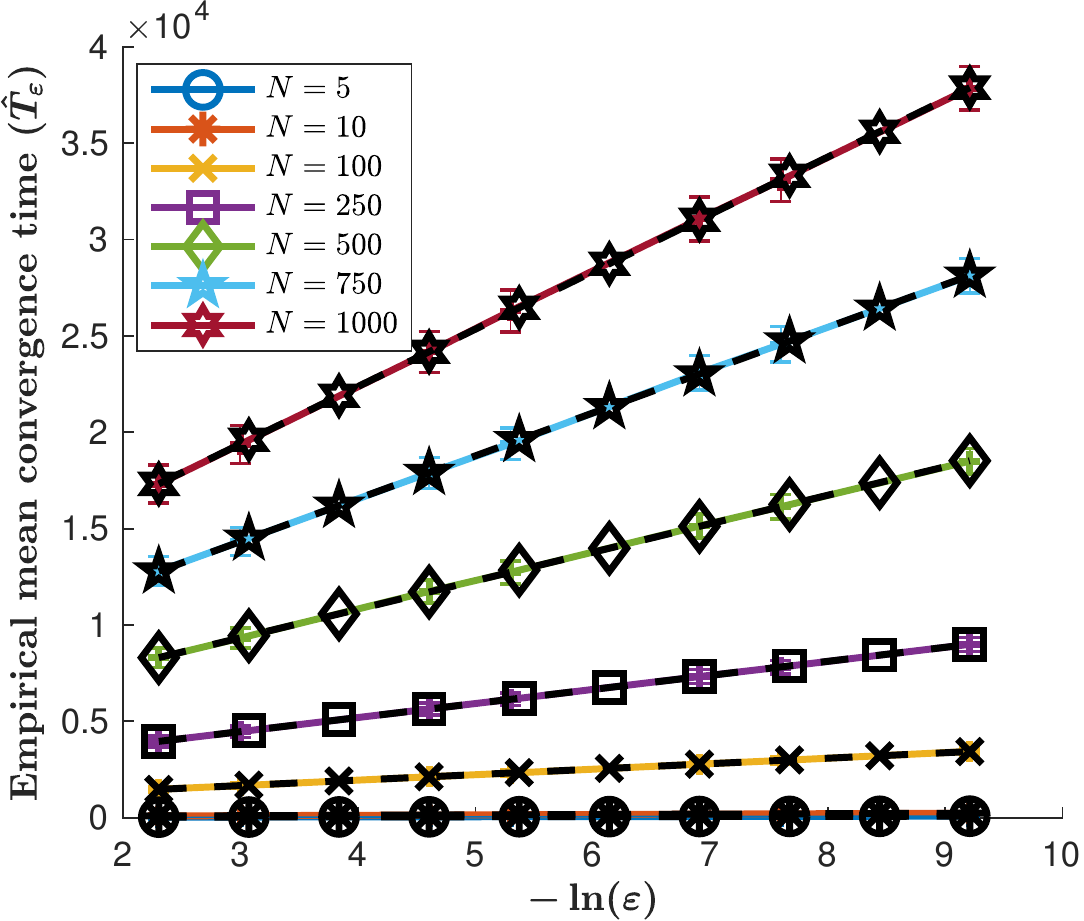}}
   \subfloat[]{\label{fig:circle c_N evolution of N}\includegraphics[width=0.49\textwidth]{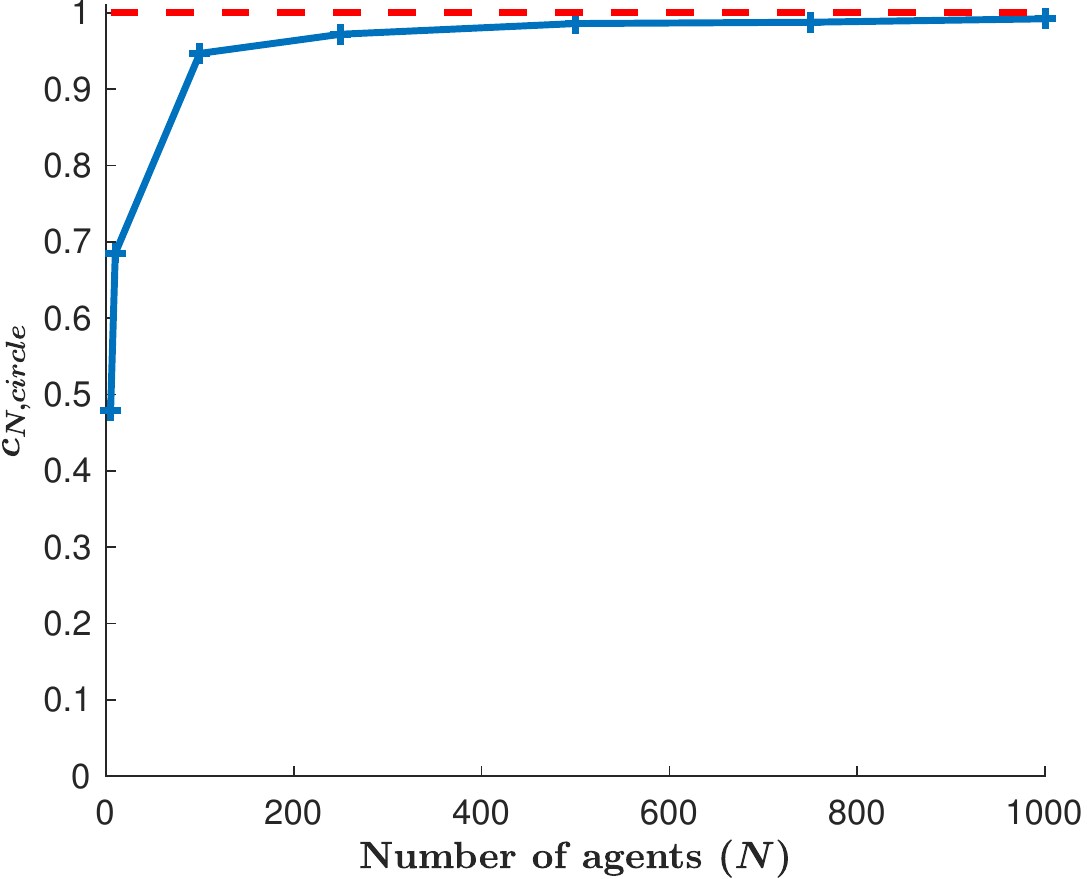}}\\
    \caption{Circle evolution: study of the linear part of the modelled dependency of the convergence time on the threshold level $\varepsilon$. Left: we superimpose on the empirical convergence time $\hat{T}_\varepsilon$ the regressed modelled one $-3g_{N,circle}\ln\varepsilon + e_{N,circle}$ in dashed black. Right: evolution of the regressed model coefficient $c_{N,circle} = \frac{g_{N,circle}}{N}$ with respect to $N$.}
    \label{fig:circle cv eps regression}
\end{figure}

\begin{figure}[tbhp]
    \centering
    \subfloat[]{\label{fig:circle e_N_over_NlnN_evol_N}\includegraphics[width=0.49\textwidth]{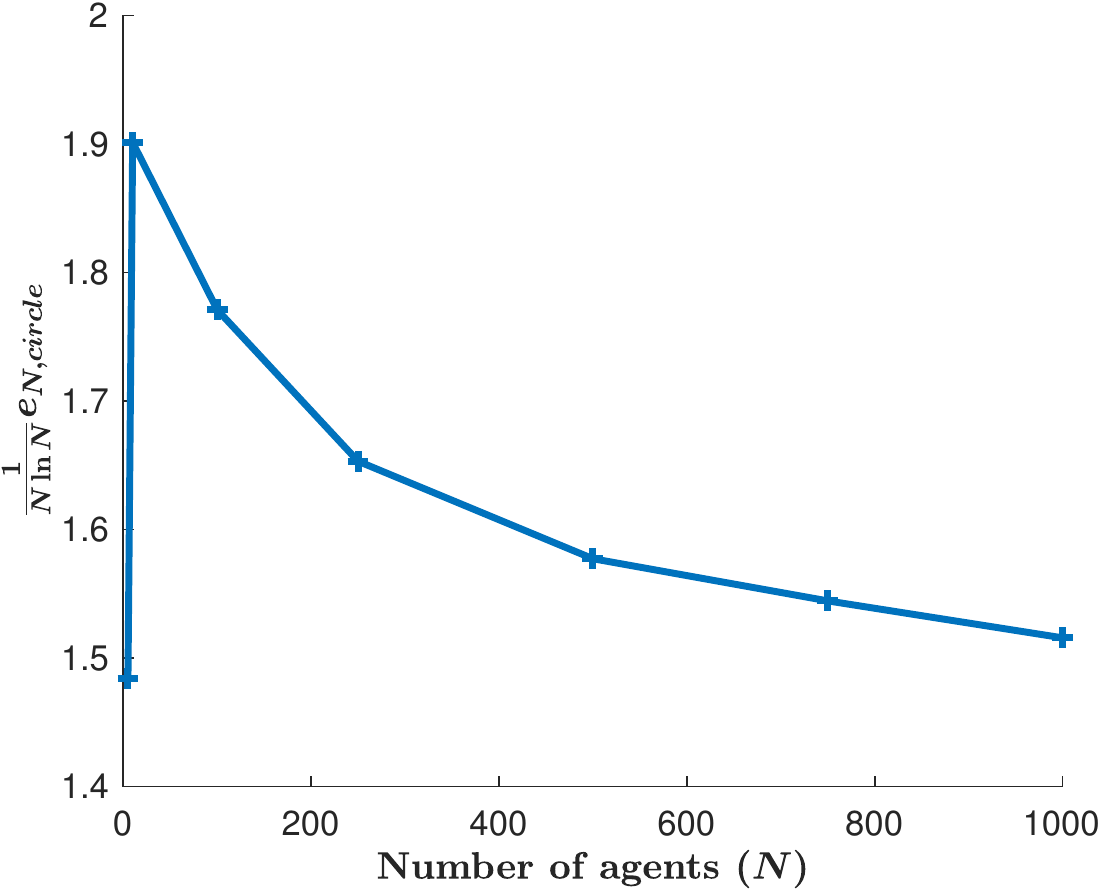}}
    \subfloat[]{\label{fig:circle e_N_evol_N_withRegression}\includegraphics[width=0.49\textwidth]{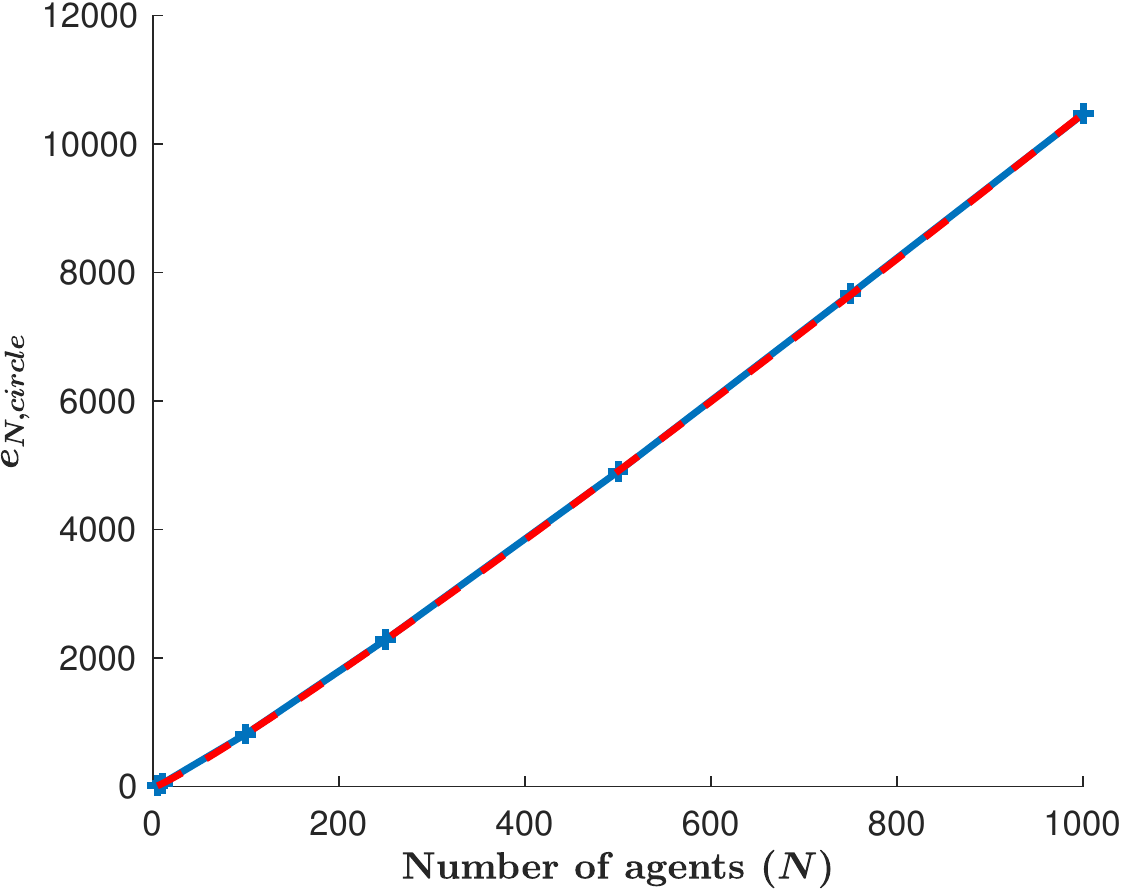}}\\
    \caption{Circle evolution: study of the offset part of the modelled dependency of the convergence time on the threshold level $\varepsilon$. Left: evolution of $\frac{e_{N,circle}}{N\ln{N}}$ with respect to $N$. Right: evolution of $e_{N,circle}$ with respect to $N$ on which we superimpose the regressed model $aN\ln{N}+bN+f$.}
    \label{fig:circle e_N evol and regression}
\end{figure}

\begin{figure}[tbhp]
    \centering 
    \subfloat{\label{fig:circle mean cv N with bounds full}\includegraphics[width=0.49\textwidth]{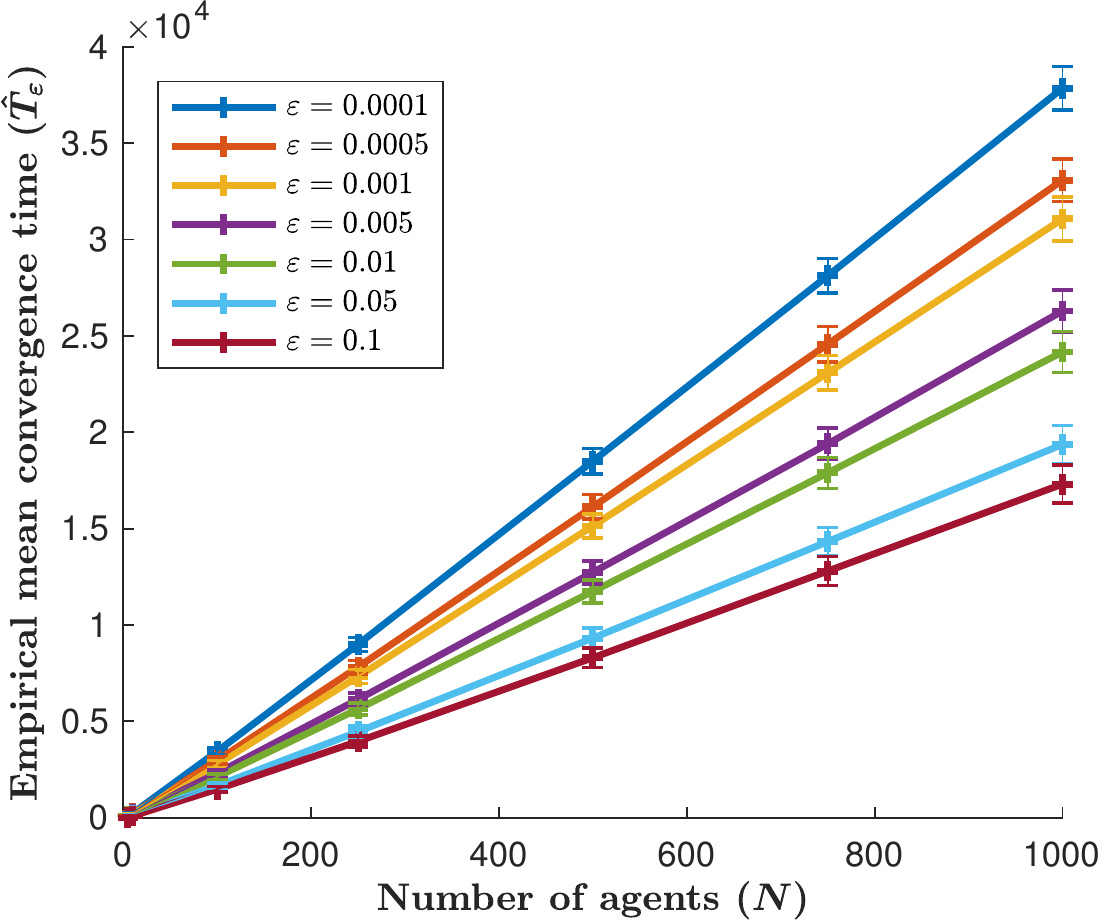}}
    \subfloat{\label{fig:circle mean cv N ln N with bounds full}\includegraphics[width=0.49\textwidth]{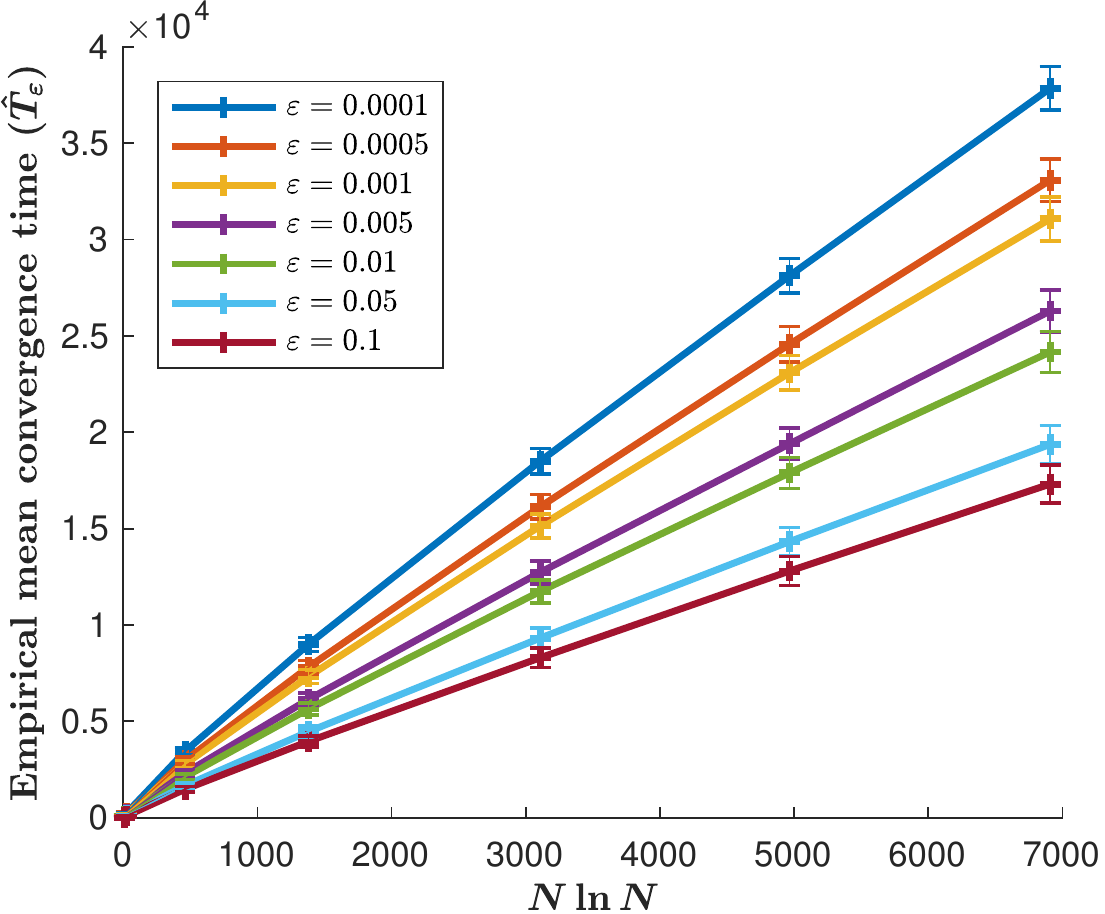}}\\    
    \caption{Circle evolution: dependency of the empirical mean convergence time on the number of agents $N$. Left: $N$ abscissa. Right: $N\ln N$ abscissa. The plain curves correspond to the empirical results whereas the dashed ones correspond to the theoretical bounds. We superimpose on the empirical curves the traditional unbiased estimator of the standard deviation of each data point.}
    \label{fig:circle cv N}
\end{figure}

\end{document}